\documentclass{LMCS}

\def\doi{8(4:6)2012}
\lmcsheading%
{\doi}
{1--46}
{}
{}
{Jan.~14, 2011}
{Oct.~11, 2012}
{}
 
\usepackage{prooftree,amsmath,amssymb,xcolor,xspace,url}
\usepackage{enumerate,hyperref}

\ifx\pdftexversion\undefined
\usepackage[dvips]{graphicx}
\else
\usepackage{graphicx}
\fi
\usepackage{xy}
\xyoption{all}

\usepackage{xspace}

\newcommand{\labelx}[1]{\label{#1}}

\newcommand{\ptilde}[1]{{\ensuremath{#1}}}
\newcommand{\kf}[1]{\textup{\textsf{#1}}\xspace}

\newcommand{\sr}[4]{\ensuremath{\bar{#1}[#2](#3).#4}}

\newcommand{\ssa}[4]{\ensuremath{#1[#2](#3).#4}}
\newcommand{\pp}{\ensuremath{\at{\p}}}
\newcommand{\si}[2]{\ensuremath{#1[#2]}}

\newcommand{\sj}[3]{\ensuremath{\bar{#1}[#2]:#3}}
\newcommand{\sii}{\si{\s}{\p}}
\newcommand{\siiv}{\si{\s}{\pv}}

\newcommand{\siq}{\si{\s}{\q}}

\newcommand{\ccc}{\ensuremath{c}}

\newcommand{\out}[4]{\ensuremath{#1!\langle \p,#2\rangle;#4}}

\newcommand{\outs}[4]{\ensuremath{#1!\langle #3,#2\rangle;#4}}
\newcommand{\e}{\ensuremath{e}}
\newcommand{\inp}[4]{\ensuremath{#1?( #3,#2);#4}}

\newcommand{\x}{\ensuremath{x}}
\newcommand{\participant}[1]{\ensuremath{\mathtt{#1}}}
\newcommand{\q}{\ensuremath{\participant{q}}}
\newcommand{\p}{\ensuremath{\participant{p}}}

\newcommand{\z}{\ensuremath{z}}
\newcommand{\pc}{\Par}
\newcommand{\s}{\ensuremath{s}}
\newcommand{\X}{\ensuremath{X}}
\newcommand{\Y}{\ensuremath{Y}}

\newcommand{\indexed}[4]{\ensuremath{\{#1_#3 : #2_#3\}_{#3 \in #4}}}

\newcommand{\values}{\ensuremath{\at{v}}}
\newcommand{\trival}[3]{\ensuremath{(#3,#2, #1)}}

\newcommand{\anglep}[2]{\ensuremath{\langle #1, #2\rangle}}
\newcommand{\valheap}[3]{\ensuremath{( #3,#2,#1 )}}

\newcommand{\delheap}[3]{\ensuremath{(#3,{#2},#1 )}}
\newcommand{\labheap}[3]{\ensuremath{( #3,#2,#1 )}}

\newcommand{\lsel}[4]{\ensuremath{#1 \oplus \anglep{#3}{#2};#4}}

\newcommand{\lbranchk}[2]{\ensuremath{#1 \&
({#2},\indexed{l}{\PP}{k}{K})}}
\newcommand{\Pifthenelse}[3]{\ensuremath{\kf{if}\ #1\ \kf{then}\ #2\ \kf{else}\ #3}}
\newcommand{\inact}{\ensuremath{\mathbf{0}}}
\newcommand{\nuc}[2]{\ensuremath{(\nu #1)#2}}

\newcommand{\true}{\kf{true}}
\newcommand{\false}{\kf{false}}

\newcommand{\h}{\ensuremath{h}}

\newcommand{\va}{\ensuremath{v}}
\newcommand{\at}[1]{\ensuremath{\ptilde{#1}}}
\newcommand{\atw}[1]{\ensuremath{\ptilde{#1}}}

\newcommand{\Par}{\ensuremath{\ |\ }}
\newcommand{\cas}{\ensuremath{r}}

\newcommand{\RECSEQ}[4]{\ensuremath{\mathbf{R} \ #1\ \lambda #2.\lambda #3.#4}}
\newcommand{\RECSEQP}[4]{\ensuremath{\mathbf{R} \ #1\ \lambda #2.\lambda #3.#4}}

\newcommand{\redsym}{\ensuremath{\longrightarrow}}
\newcommand{\red}[2]{\ensuremath{#1\redsym#2}}
\newcommand{\redM}[2]{\ensuremath{#1\redsym^*#2}}
\newcommand{\set}[1]{\ensuremath{\{#1\}}}
\newcommand{\sub}[2]{\ensuremath{\{#1/#2\}}}

\newcommand{\sep}{\ensuremath{~\mathbf{|}~ }}

\newcommand{\Implies}{\ensuremath{\quad \Rightarrow \quad }}

\newcommand{\qbot}{\ensuremath{\epsilon}}
\newcommand{\mqueue}[2]{\ensuremath{#1 : #2}}
\newcommand{\emptyqueue}[1]{\mqueue{\s}{\qbot}}
\newcommand{\queue}{\ensuremath{\h}}
\newcommand{\stdqueue}{\mqueue{\s}{\queue}}
\newcommand{\qcomp}[2]{\ensuremath{#1 \cdot#2}}
\newcommand{\qtail}[1]{\ensuremath{\qcomp{\queue}{#1}}}
\newcommand{\qhead}[1]{\ensuremath{\qcomp{#1}{\queue}}}

\newcommand{\qappend}[1]{\mqueue{\s}{\qtail{#1}}}
\newcommand{\qpop}[1]{\mqueue{\s}{\qhead{#1}}}

\newcommand{\subst}[2]{\ensuremath{\{#1 / #2\}}}

\newcommand{\freen}[1]{\ensuremath{\text{fn}(#1)}}

\newcommand{\G}{\ensuremath{G}}

\newcommand{\U}{\ensuremath{U}}

\newcommand{\pro}[2]{\ensuremath{#1\upharpoonright#2}}
\newcommand{\Ga}{\ensuremath{\Gamma}}
\newcommand{\D}{\ensuremath{\Delta}}

\newcommand{\T}{\ensuremath{T}}
\newcommand{\TQ}{\ensuremath{{\tt{T}}}}
\newcommand{\TG}{\ensuremath{{\mathsf{T}}}}

\newcommand{\ST}{\ensuremath{S}}

\newcommand{\SST}{\atw{S}}
\newcommand{\UT}{\ensuremath{U}}
\newcommand{\oT}[2]{\ensuremath{\;!\langle #2,#1\rangle}}
\newcommand{\iT}[2]{\ensuremath{? \langle #2,#1 \rangle}}
\newcommand{\oTG}[2]{\ensuremath{\;\natural\langle #2,#1\rangle}}
\newcommand{\oTGp}[2]{\ensuremath{\;\natural'\langle #2,#1\rangle}}

\newcommand{\de}[3]{\ensuremath{#1\vdash#2:#3}}
\newcommand{\der}[3]{\ensuremath{#1\vdash#2\triangleright#3}}
\newcommand{\dom}[1]{\ensuremath{dom( #1)}}
\newcommand{\ty}{\textbf{t}}
\newcommand{\End}{\kf{end}}
\newcommand{\Bool}{\kf{bool}}
\newcommand{\Nat}{\kf{nat}}

\newcommand{\seltype}{\ensuremath{\oplus \langle \p,\{l_i:\T_i\}_{i\in
I} \rangle }}

\newcommand{\seltypeT}{\ensuremath{\oplus\{l_i:\pro{\T_i}\q\}_{i\in I}}}
\newcommand{\seltypeTp}{\ensuremath{\oplus\{l_i:\T_i\}_{i\in I}}}

\newcommand{\seltypes}{\ensuremath{\oplus\langle\pv,l\rangle}}

\newcommand{\branchtype}{\ensuremath{\&\langle\p,\{l_k:\T_k\}_{k\in K}\rangle}}

\newcommand{\branchtypeT}{\ensuremath{\&\{l_i:\pro{\T_i}\q\}_{i\in I}}}
\newcommand{\branchtypeTp}{\ensuremath{\&\{l_i:\T'_i\}_{i\in I}}}
\newcommand{\branchtypes}{\ensuremath{\&\{l_i:\T_i\}_{i\in I}}}

\newcommand{\trule}[1]{\text{\footnotesize{\ensuremath{\lfloor\text{\sc{#1}}\rfloor}}}}

\newcommand{\scripttrule}[1]{\text{\scriptsize{\ensuremath{\lfloor\text{\sc{#1}}\rfloor}}}}

\newcommand{\tfrule}[1]{{\text{\scriptsize[\text{\sc{#1}}]}}}
\newcommand{\tftrule}[1]{{\text{\footnotesize[\text{\sc{#1}}]}}}

\newcommand{\equivT}[2]{\ensuremath{#1\approx #2}}
\newcommand{\derqq}[4]{\ensuremath{#1 \vdash_{#2} #3 \triangleright #4}}
\newcommand{\derq}[3]{\ensuremath{#1 \vdash_{\set\s} #2 \triangleright #3}}
\newcommand{\ms}[2]{\ensuremath{{#1}\setminus{#2}}}
\newcommand{\coe}[2]{\ensuremath{\mathsf{co}({#1},{#2})}}

\newcommand{\Dcomp}{\ensuremath{\ast}}
\newcommand{\Tcomp}{\ensuremath{;}}
\newcommand{\dual}[2]{\ensuremath{{#1}\bowtie{#2}}}

\newcommand{\sered}[2]{\ensuremath{{#1}~\Rightarrow~{#2}}}
\newcommand{\seredstar}[2]{\ensuremath{{#1}~\Rightarrow^*~{#2}}}

\newcommand{\E}{\ensuremath{\mathcal{E}}}

\newcommand{\outS}[3]{\ensuremath{#1!\langle #2\rangle;#3}}
\newcommand{\inpS}[3]{\ensuremath{#1?( #2);#3}}

\newcommand{\andl}{\ensuremath{~\wedge~}}

\newcommand{\ENCan}[1]{\langle #1 \rangle}
\newcommand{\TO}[2]{#1\to #2}
\newcommand{\GS}[3]{\TO{#1}{#2}\colon \!\ENCan{#3}}
\newcommand{\GB}[2]{\TO{#1}{#2}\colon\! \indexed{l}{\G}{k}{K}}
\newcommand{\GR}[4]{\RECSEQP{#1}{#2}{#3}{#4}}
\newcommand{\GM}[2]{\mu #1.#2}
\newcommand{\Lout}[3]{\ensuremath{!\langle #1,#2\rangle;#3}}
\newcommand{\Loutt}[2]{\ensuremath{!\langle #1,#2\rangle}}
\newcommand{\Lin}[3]{\ensuremath{?\langle #1,#2\rangle;#3}}
\newcommand{\Linn}[2]{\ensuremath{?\langle #1,#2\rangle}}
\newcommand{\Lsel}[2]{\ensuremath{\oplus\langle#1,\{l_k:#2\}_{k\in
K}\rangle}}
\newcommand{\LselSingle}[2]{\ensuremath{\oplus\langle{#1,#2}\rangle}}
\newcommand{\Lbranch}[2]{\ensuremath{\&\langle#1,\{l_k:#2\}_{k\in K}\rangle}}
\newcommand{\LselI}[4]{\ensuremath{\oplus\langle#1,\{l_{#3}:#2\}_{#3\in #4}\rangle}}
\newcommand{\LbranchI}[4]{\ensuremath{\&\langle#1,\{l_{#3}:#2\}_{{#3}\in #4}\rangle}}
\newcommand{\LR}[4]{\RECSEQP{#1}{#2}{#3}{#4}}
\newcommand{\LM}[2]{\ensuremath{\mu #1.#2}}
\newcommand{\Pout}[4]{\ensuremath{#1!\langle #2,#3\rangle;#4}}
\newcommand{\Pin}[4]{\ensuremath{#1?(#2,#3);#4}}

\newcommand{\Psel}[4]{\ensuremath{#1\oplus\langle#2,#3\rangle;#4}}
\newcommand{\Pbranch}[2]{\ensuremath{#1\&(#2,\{l_k:\PP_k\}_{k\in K})}}
\newcommand{\APP}{\;}

\newcommand{\mar}[1]{\ensuremath{\langle #1 \rangle}}
\newcommand{\ftv}{\kf{ftv}}
\newcommand{\fv}{\kf{fv}}
\newcommand{\fn}{\kf{fn}}
\newcommand{\pid}{\kf{pid}}
\newcommand{\IF}{\kf{if}}
\newcommand{\THEN}{\kf{then}}
\newcommand{\ELSE}{\kf{else}}

\newcommand{\y}{\ensuremath{y}}
\newcommand{\Ia}{\ensuremath{a}}

\newcommand{\Iu}{\ensuremath{u}}
\newcommand{\Iv}{\ensuremath{v}}

\newcommand{\ii}{\ensuremath{i}}
\newcommand{\jj}{\ensuremath{j}}
\newcommand{\kk}{\ensuremath{k}}
\newcommand{\Ll}{\ensuremath{l}}
\newcommand{\n}{\ensuremath{\mathrm{n}}}
\newcommand{\m}{\ensuremath{\mathrm{m}}}
\newcommand{\nn}{\ensuremath{n}}
\newcommand{\mm}{\ensuremath{m}}
\newcommand{\pv}{\ensuremath{\at{\hat{\p}}}}
\newcommand{\qq}{\ensuremath{\at{\q}}}
\newcommand{\qv}{\ensuremath{\at{\hat{\q}}}}
\newcommand{\uu}{\ensuremath{u}}
\newcommand{\xx}{\ensuremath{\mathbf{x}}}
\newcommand{\yy}{\ensuremath{\mathbf{y}}}
\newcommand{\zz}{\ensuremath{\mathbf{z}}}
\newcommand{\uuu}{\ensuremath{\mathbf{u}}}

\newcommand{\II}{\ensuremath{I}}
\newcommand{\K}{\ensuremath{\kappa}} 
\newcommand{\PP}{\ensuremath{P}}
\newcommand{\QQ}{\ensuremath{Q}}
\newcommand{\Q}{\ensuremath{Q}}

\newcommand{\Ty}{\ensuremath{\tau}}
\newcommand{\Names}{\ensuremath{\mathcal{P}}}
\newcommand{\Env}{\kf{Env}}
\newcommand{\Type}{\kf{Type}}
\newcommand{\GType}{\Type}
\newcommand{\LType}{\Type}
\newcommand{\SType}{\Type}
\newcommand{\PType}{\Type}
\newcommand{\PRType}{\Type}
\newcommand{\Alice}{\ensuremath{\mathtt{Alice}}}
\newcommand{\Bob}{\ensuremath{\mathtt{Bob}}}
\newcommand{\Carol}{\ensuremath{\mathtt{Carol}}}
\newcommand{\W}{\ensuremath{\mathtt{W}}}
\newcommand{\Worker}{\ensuremath{\mathtt{Worker}}}
\newcommand{\Buyer}{\ensuremath{\mathtt{Buyer}}}
\newcommand{\Supp}{\ensuremath{\mathtt{Supp}}}
\newcommand{\Manu}{\ensuremath{\mathtt{Manu}}}

\newcommand{\bit}[1]{\ensuremath{\text{bit}_{#1}}}

\newcommand{\subT}{\ensuremath{\leq}}
\newcommand{\RHD}{\,\ensuremath{\blacktriangleright}\,}
\newcommand{\bnfor}{~\ensuremath{~\vert~}~} 
\newcommand{\op}{~\texttt{op}~}

\newcommand{\minus}[1]{\ensuremath{#1^{-}}}
\newcommand{\proj}[1]{\ensuremath{\upharpoonright #1}}

\newcommand{\AT}[2]{#1\! : \! #2}

\newcommand{\tii}{\ensuremath{\mathtt{i}}}
\newcommand{\tjj}{\ensuremath{\mathtt{j}}}

\newcommand{\CONSTRAINT}[3]{\ensuremath{\{\AT{#1}{#2} \ | \ #3\}}} 
\newcommand{\PRED}{\ensuremath{\mathtt{P}}}

\newcommand{\eg}{e.g.~}

\newtheorem{DUM}{dummy}[section]{}{}
\newtheorem{DEFINITION}[DUM]{Definition}{}{}
{}{}
{}{}
\newtheorem{PROP}[DUM]{Proposition}{}{}
{}{}
{}{}
{}{}
{}{}

\newcommand{\proves}{\vdash}                        
\newcommand{\judg}{{J}}
\newcommand{\VEC}{\tilde}

\newcommand{\WHNF}[1]{\ensuremath{\text{whnf}(#1)}}

\newcommand{\WB}{\approx}

\newcommand{\MERGE}[1]{\ensuremath{#1}}
\newcommand{\FOREACH}[3]{\ensuremath{\mathtt{foreach}(#1 #2)\{#3\}}}

\newcommand{\mergeop}{\ensuremath{\bowtie}}
\newcommand{\mergecup}{\ensuremath{\sqcup}}

\newcommand{\ParT}{U_p}

\newcommand{\equivwf}{\equiv_{\text{wf}}}
\newcommand{\gequivwf}{\equiv}

\newcommand{\termsize}[1]{|#1|}
\newcommand{\inductiontermsize}[1]{\vert\vert #1 \vert\vert}
\newcommand{\judgementsize}[1]{w(#1)}
\newcommand{\reductionsize}[1]{\mu(#1)}

\newcommand{\inductionreductionsize}[1]{\mu^\star(#1)}

\newif\ifmc
\mcfalse 
\newcommand{\ny}[1]
{\ifmc{\color{magenta}{#1}}\else{#1}\fi}

\newcommand{\Poutend}[3]{\ensuremath{#1!\langle #2,#3\rangle}}

\newcommand{\valheaps}[3]{\ensuremath{( #3,#2,#1 )}}

\title[Parameterised Multiparty Session Types]{Parameterised Multiparty Session Types}
\thanks{The work is partially supported by EPSRC EP/G015635/1 and EP/F003757/1.}

\author[P.-M.~Deni\'elou]{Pierre-Malo Deni\'elou}

\author[N.~Yoshida]{Nobuko Yoshida}

\author[A.~Bejleri]{Andi Bejleri}

\author[R.~Hu]{Raymond Hu}
\address{Department of Computing, Imperial College London, 180 Queen's
  Gate, LONDON, SW7 2AZ, UK}
\email{\{malo, yoshida, rhu\}@doc.ic.ac.uk, andi.bejleri06@imperial.ac.uk}

\begin{document}

\begin{abstract}
\noindent
For many application-level distributed protocols and parallel algorithms, the
set of participants, the number of messages or the interaction structure are
only known at run-time.
This paper proposes a dependent type theory for multiparty sessions which can
statically guarantee type-safe, deadlock-free multiparty interactions among
processes whose specifications are parameterised by indices.
We use the primitive recursion operator from G\"odel's System $\mathcal{T}$ to
express a wide range of communication patterns while keeping type checking
decidable.
To type individual distributed processes, a parameterised global type is
projected onto a generic generator which represents a class of all possible
end-point types.
We prove the termination of the type-checking algorithm in the
full system with both multiparty session types and recursive types.
We illustrate our type theory through non-trivial
programming and verification examples
taken from parallel algorithms and web services usecases.
\end{abstract}

\keywords{Session Types, Dependent Types, The Pi-Calculus, G\"odel T, Parallel Algorithms, FFT, Web services}
\subjclass{F.3.3, D.1.3, F.1.1, F.1.2}

\maketitle

\section{Introduction}
\label{sec:introduction}
\noindent
As the momentum around communications-based computing grows, the need
for effective frameworks to
globally {\em coordinate} and {\em structure} the application-level
interactions is pressing. 
The structures of interactions are naturally distilled as
{\em protocols}. Each protocol describes a bare skeleton of how
interactions should proceed, through e.g.
sequencing, choices and repetitions.
In the theory of multiparty session types~\cite{CHY07,BC07,BettiniCDLDY08LONG},
such protocols can be captured as types for interactions,
and type checking can statically ensure runtime safety and fidelity
to a stipulated
protocol.

One of the particularly challenging aspects of protocol descriptions
is the fact that many actual communication protocols are
highly {\em parametric}
in the sense that the number of participants  and
even the interaction structure itself are not fixed at design time.
Examples include parallel algorithms such as
the Fast  Fourier Transform 
(run on any number of
communication nodes depending on resource availability)
and Web
services such as business negotiation involving an
arbitrary number of sellers and buyers.
This nature is important, for instance,
for the programmer of a parallel algorithm
where the size or shape of the communication topology,
or the number of available threads might
be altered depending on the number
of available cores in the machine.
Another scenario is web services where the
participant sets may be known at design time,
or instantiated later.
This paper introduces a robust dependent type theory
which can statically ensure communication-safe, deadlock-free process
interactions which follow parameterised multiparty protocols.
%

We illustrate the key ideas of our proposed parametric type structures
through examples. Let us first consider a simple protocol where
participant $\Alice$ sends a message of type $\Nat$ to participant
$\Bob$.  To develop the code for this protocol, we start by specifying
the global type, which can concisely and clearly describe a high-level
protocol for multiple participants
\cite{CHY07,BettiniCDLDY08LONG,esop09}, as follows ($\End$ denotes
protocol termination):
\[
\begin{array}{c}
G_1  \ = \ \GS{\Alice}{\Bob}{\Nat}.\End\\[-2ex]
\end{array}
\]
The flow of communication is indicated with the symbol $\rightarrow$ and
upon agreement on $G_1$ as a specification for $\Alice$ and $\Bob$,
each program can be implemented separately,
e.g. as $y!\ENCan{100}$ (output 100 to $y$) by $\Alice$ and
$y?(z);{\bf 0}$ (input at $y$) by $\Bob$.
For type-checking, $G_1$ is
{\em projected} into end-point session types:
one from $\Alice$'s point of view,
$!\ENCan{\Bob,\Nat}$ (output to $\Bob$ with $\Nat$-type),
and another from $\Bob$'s point of view, $?\ENCan{\Alice,\Nat}$
(input from $\Alice$ with $\Nat$-type), against which
the respective $\Alice$ and $\Bob$ programs are checked to be compliant.

The first step towards generalised type structures for multiparty
sessions is to allow modular specifications of protocols using
arbitrary compositions and repetitions of interaction units (this is a
standard requirement in multiparty contracts \cite{CDL}).  Consider
the type $G_2 \ = \ \GS{\Bob}{\Carol}{\Nat}.\End$.  The designer may
wish to compose sequentially $G_1$ and $G_2$ together to build a larger
protocol:
\[
\begin{array}{rclcl}
G_3
& = & \MERGE{G_1;G_2} 
& = & \GS{\Alice}{\Bob}{\Nat}.\GS{\Bob}{\Carol}{\Nat}.\End
\end{array}
\]
We may also want to iterate the composed protocols $\n$-times, which can be
written by $\FOREACH{i}{<\n}{\MERGE{G_1;G_2}}$, and moreover bind the number of
iteration $\n$ by a dependent product
to build a {\em family of global specifications}, as in
($\Pi n$ binds variable $n$):
\begin{eqnarray}
\label{alice-bob-carol}
\Pi n.\FOREACH{\ii}{<n}{\MERGE{G_1;G_2}}
\end{eqnarray}
%
%
%
Beyond enabling a variable number of exchanges between a fixed set of
participants, the ability to parameterise {\em
  participant identities} can represent
%
a wide class of the communication topologies found in the literature.  For
example, the use of indexed participants $\W[i]$ (denoting the $i$-th
worker) 
allows the specification of a family of session types such that
neither the number of participants nor
message exchanges are
known before the run-time instantiation of the parameters. The following type
and diagram both describe a sequence of messages from $\W[n]$ to $\W[0]$
(indices decrease in our $\mathtt{foreach}$, see \S~2):
\begin{eqnarray}
\label{ex:sequence}
\Pi n.(\FOREACH{\ii}{<n}{\GS{\W[\ii+1]}{\W[\ii]}{\Nat}})
\qquad
\begin{minipage}{10em}
\xymatrix@C=15pt{
  *+[F]{{\footnotesize \n}}\ar[r]
  &*+[F]\txt{\footnotesize \n-1}\ar[r]
  &{\ldots}\ar[r]
  &*+[F]\txt{\footnotesize 0}}
\end{minipage}
\end{eqnarray}
Here we face an immediate question:
{\em what is the underlying type structure for such parametrisation,
  and how can we type-check each (parametric) end-point program?}  The
type structure should allow the projection of a parameterised global
type to an end-point type {\em before} knowing the exact shape of the
concrete topology.

In (\ref{alice-bob-carol}), corresponding end-point
types are parameterised {\em families} of session types.
For example,
\Bob \ would be typed by $\Pi
j.\FOREACH{\ii}{<j}{?\ENCan{\Alice,\Nat};!\ENCan{\Carol,\Nat}}$, which
represents the product of session interactions with different lengths. The
choice is made when $j$ is instantiated, i.e. before execution.
The difficulty of the projection
arises in
(\ref{ex:sequence}): if $\n\!\geq\! 2$, there are three
distinct {\em communication patterns} inhabiting this
specification: the initiator $\W[\n]$ (send only), the $\n-1$ middle workers
(receive then send), and the last worker $\W[0]$ (receive only). This is no
longer the case when $\n= 1$ (there is only the initiator and the last
worker) or when $\n= 0$ (no communication at all).  Can we provide a
decidable projection and static type-checking by which we
can preserve the main properties of the session types such as progress
and communication-safety in parameterised process topologies?  The key
technique proposed in this paper is a projection method from a
dependent global type onto a {\em generic end-point generator} which
exactly captures the interaction structures of parameterised end-points
and which can represent the class of all possible end-point
types.

The main contributions of this paper follow:
\begin{iteMize}{$\bullet$}
\item {\em A new expressive framework to globally specify and program}
a wide range of
parametric communication protocols
(\S~\ref{sec:syntax_and_example}). We achieve this result by
combining dependent type theories derived from
G\"odel's System $\mathcal{T}$
\cite{DBLP:conf/mfps/Nelson91} (for expressiveness)
and indexed dependent types from \cite{DBLP:conf/popl/XiP99} (for parameter
control), with multiparty session types.

\item
{\em Decidable and flexible projection methods}
based on a generic end-point generator
and mergeability of branching types, enlarging the typability
(\S~\ref{subsec:endpoint}).

\item {\em A dependent typing system} that treats the full
multiparty session types integrated with dependent types (\S~\ref{sec:typing}).

\item {\em Properties of the dependent typing system}
which include
decidability of
  type-checking.
The resulting static typing system also
guarantees type-safety and deadlock-freedom (progress) for
well-typed processes involved in parameterised multiparty communication
protocols (\S~\ref{sec:property}).

\item {\em Applications} featuring various process topologies
(\S~\ref{sec:syntax_and_example},\S~\ref{subsec:typingexample}),
  including the complex butterfly network of the parallel
  FFT algorithm
(\S~\ref{subsec:fft}, \S~\ref{subsec:fft:typing}).
As far as we know, this is the first time
such a complex protocol is 
specified by
a single {\em type} and that its implementation can be automatically
type-checked
to prove communication-safety and deadlock-freedom.
We also extend the calculus with a new asynchronous
primitive for session initialisation and apply it to web services usecases
\cite{CDLRequirements} (\S~\ref{sec:applications}).
\end{iteMize}

\noindent Section~\ref{sec:syntax_and_example} gives the definition of the parameterised
types and processes, with their semantics. Section~\ref{sec:typing} describes
the type system. The main properties of the type system are presented in
Section~\ref{sec:property}. Section~\ref{subsec:typingexample} shows typing
examples. Section~\ref{sec:related} concludes and discusses related work.

This article is a full version expanded from \cite{YDBH10}, with complete
definitions and additional results with detailed proofs.
It includes more examples with detailed explanations and verifications,
as well as expanded related work.
Some additional material related to implementations and programming
examples will be discussed in \S~\ref{sec:related}.

\section{Types and processes for parameterised multiparty sessions}
\label{sec:syntax_and_example}

\subsection{Global types}
\label{sec:globaltypes}
\noindent 
Global types allow the description
of the parameterised conversations of
multiparty sessions as a type signature.
Our type syntax integrates
elements from three different theories:
(1) global types from
\cite{BettiniCDLDY08LONG}; (2)
dependent types with primitive recursive combinators based on
\cite{DBLP:conf/mfps/Nelson91};
and (3) parameterised dependent types 
from a simplified Dependent ML 
\cite{DBLP:conf/popl/XiP99,DependentBook}.  

\begin{figure}[ht]
\small
\begin{center}
\begin{tabular}{l@{}l}
\begin{tabular}{r@{\ }c@{\ }l@{\ }l}
$\tii$ & ::= & $\ii \bnfor \n \bnfor \tii \op \tii'$
& Indices\\ 
\PRED & ::= & $\PRED\andl \PRED \sep \tii\leq \tii'$& Propositions\\ 
\II & ::= & \Nat \sep \CONSTRAINT{\ii}{\II}{\PRED} & Index sorts \\ 
\Names & ::= & \Alice \sep \Worker \sep  \ldots & Participants \\
\p & ::= & $\p[\tii] \sep \Names$ & Principals\\
\ST & ::= & \Nat \sep \mar \G & Value type\\
\U & ::= & \ST \sep \T & Payload type \\
K & ::= & $\{\n_0, ..., \n_k\}$ & Finite integer set\\
\end{tabular} &
\begin{tabular}{r@{\ }c@{\ }l@{\ }l}
\G & ::= & & \hspace{-2em} Global types \\ 
   & \sep & $\GS{\p}{\p'}{\U}.\G$ & Message \\
   & \sep & $\GB{\p}{\p'}$  & Branching \\
   & \sep & $\GM{\xx}{\G}$ & Recursion \\
   & \sep & $\GR{\G}{\AT{\ii}{\II}}{\xx}{\G'}$  & Primitive recursion \\
   & \sep & \xx & Type variable \\
   & \sep & \G\APP \tii   & Application\\
   & \sep & \End & Null \\
\end{tabular}
\end{tabular}
\end{center}
\caption{Global types}\label{fig:global}
\end{figure}

\begin{figure}[ht]
\begin{center}
\begin{tabular}{rl cl}
  $\GR{\G}{\AT{\ii}{\II}}{\xx}{\G'}$ & $0$  & \redsym & \G \\
  $\GR{\G}{\AT{\ii}{\II}}{\xx}{\G'}$ & $(\n\!+\!1)$ & \redsym & 
    $\G'\sub{\n}{\ii}\sub{\GR{\G}{\AT{\ii}{\II}}{\xx}{\G'}\APP \n}{\xx}$ \\[-2ex]
\end{tabular}
\end{center}
\caption{Global type reduction}\label{fig:globalreduction}
\end{figure}

 

The grammar of global types ($G,G',...$) is given in Figure~\ref{fig:global}.
{\em Parameterised principals}, written $\p,\p',\q,\ldots$, 
can be indexed by one or more parameters, \eg $\Worker[5][\ii+1]$. Index
$\tii$ ranges over index variables $\ii,\jj,\nn$, naturals $\n$ or
arithmetic operations.  A global interaction can be a message exchange
($\GS{\p}{\p'}{\U}.\G$), where $\p,\p'$ denote the sending and receiving
principals, $\U$ the payload type of the message and $\G$ the subsequent
interaction.  Payload types $\U$ are either value types $\ST$ (which contain
base type $\Nat$ and session channel types $\mar{G}$), or 
{\em end-point types} $T$
(which correspond to the behaviour of one of the session participants and will
be explained in \S~\ref{sec:typing}) for
delegation. Branching ($\GB{\p}{\p'}$) allows the session to follow one of the different $\G_k$
paths in the interaction ($K$ is a ground and finite set of integers). 
$\GM{\xx}{\G}$ is a recursive type where   
type variable $\xx$ 
is guarded in the standard way 
(they only appear under some prefix)
 \cite{PierceBC:typsysfpl}.

The main novelty is the primitive recursive operator
$\GR{\G}{\AT{\ii}{\II}}{\xx}{\G'}$ from G\"odel's System
$\mathcal{T}$~\cite{GirardJY:protyp} 
whose reduction semantics is given in Figure~\ref{fig:globalreduction}.  
Its parameters are a global type $\G$, an index variable
$\ii$ with range $\II$,
 a type variable for recursion $\xx$ and a recursion body
$\G'$.\footnote{We distinguish recursion and primitive recursion in order to get
  decidability results, see \S~\ref{sec:subjectreduction}. } 
When applied to an index $\tii$, its semantics corresponds to the
repetition $\tii$-times of the body $\G'$, with the index variable $\ii$ value
going down by one at each iteration, from $\tii-1$ to $0$. The final
behaviour is 
given by $\G$ when the index reaches $0$. 
The index sorts comprise the set of natural numbers and its restrictions by 
predicates ($\PRED,\PRED',..$) that are, in our case, 
conjunctions of inequalities.  $\op$ represents first-order indices operators
(such as $+$, $-$, $\ast$,...). 
We often omit $\II$ and $\End$ in our examples. 
 
Using $\mathbf{R}$, we define the product, composition, repetition and test
operators as syntactic sugar (seen in \S~\ref{sec:introduction}):
{\footnotesize
\[
\begin{array}{@{}l@{}c@{\,}l|@{\;\;}l@{}c@{\,}l}
\Pi \AT{\ii}{\II}.\G & = & \GR{\End}{\ii}{\xx}{\G\sub{\ii+1}{\ii}}
& \FOREACH{\ii}{\!<\!\jj}{G}& = &
 \GR{\End}{\ii}{\xx}{\G\sub{\xx}{\End}} \ \jj\\[1mm]
\MERGE{G_1;G_2} & = &
\GR{\G_2}{\ii}{\xx}{\G_1\sub{\xx}{\End}} \ 1 \ 
& \IF\ \jj \ \THEN\ \G_1 \ \ELSE \ \G_2 & = &
\GR{\G_2}{\ii}{\xx}{\G_1}\APP \jj
\end{array}
\]}

\noindent
where we assume that $\xx$ is not free in $G$ and $G_1$, and
that the leaves of the syntax trees of $G_1$ and $G$ are 
$\End$. 
These definitions rely on a special substitution of each ${\small\End}$ by ${\small\xx}$ (for example, 
${\small \TO{\p}{\p'}\{ l_1\!\!: !\ENCan{\Nat};\End, l_2\!\!:\!\End\}\sub{\xx}{\End}=
\TO{\p}{\p'}\{ l_1\!\!: !\ENCan{\Nat};\xx, l_2\!\!:\! \xx\}}$). 
The composition operator (which we usually write `$;$')
appends the execution of $G_2$ to $G_1$; the repetition operator above repeats
$\G$ $\jj$-times\footnote{
This version of $\mathtt{foreach}$ uses decreasing indices. One can write an
increasing version, see \S~\ref{sub:globalexample}.}; the boolean values are integers $0$
(\false) and $1$ (\true).
Similar syntactic sugar is defined for local types and processes. 

Note that composition and repetition do not necessarily impose 
sequentiality: only the order of the asynchronous messages and the 
possible dependencies~\cite{CHY07} between receivers and subsequent senders
controls the sequentiality. For example, a parallel version of 
the sequence example of
(\S~\ref{sec:introduction} (\ref{ex:sequence})) can be written in our
syntax as follows: 
\begin{eqnarray}\label{parallel_seq}
 \Pi n.(\FOREACH{\ii}{<n}{\GS{\W[n-\ii]}{\W[n-\ii-1]}{\Nat}})
\end{eqnarray}
where each worker $\W[\jj]$ sends asynchronously a value $v_j$ to its next 
worker $\W[\jj-1]$ without waiting for the message from $\W[\jj+1]$ to arrive
first (i.e.~each choice of $v_j$ is independent from the others).

\label{subsec:examples}

\begin{figure}[t]\centering\small
\mbox{}\\[-4ex]
\begin{tabular}{@{}c@{\qquad}c}
{\bf \textsf (a) Ring}\\[-5ex]
 \begin{minipage}{8em}
 \xymatrix@C=15pt{
   *+[F]{\text{\small 0}}\ar[r]
   &*+[F]{\text{\small 1}}\ar[r]
   &{\ldots}\ar[r]
   &*+[F]{\text{\small $n$}}\ar@/^1.2pc/[lll]}
 \end{minipage}
&\hspace{-4em}
\begin{minipage}{25em}{\small
\begin{align*}
      \Pi \AT{n}{I}.(&\FOREACH{i<n}{}{\GS{\W[n-i-1]}{\W[n-i]}{\Nat}};\\
      & \GS{\W[n]}{\W[0]}{\Nat}.\End) \\
    \end{align*}}
\end{minipage}
\\
{\bf \textsf (b) Multicast}\\[1ex]
\begin{minipage}{8em}{\small
\xymatrix@C=15pt{
  &  & *+[F]{\small \Alice}\ar[dl]\ar[dll]\ar[dr] & & \\
   *+[F]{\small 0} &*+[F]{\small 1} & *{\cdots} & *+[F]{\small n-1} \\
}}
\end{minipage}
&\hspace{-4em} 
\begin{minipage}{25em}{\small\begin{align*}
\Pi \AT{n}{I}.&\FOREACH{i<n}{}{\GS{\Alice}{\W[n-1-\ii]}{\Nat}};\End
    \end{align*}}
\end{minipage}
\\[10ex]
{\bf \textsf (c) Mesh}\\[-5ex]
\begin{minipage}{8em}
\xymatrix@C=10pt@R=10pt{
  *{\W[n][m]\hspace{-1em}}&*[F]{\hole} \ar[r]\ar[d]
  &*[F]{\hole}\ar[r]\ar[d]
  &{\ldots}\ar[r]
  &*[F]{\hole}\ar[d]
  &\\
  &*[F]{\hole}\ar[r]\ar[d]
  &*[F]{\hole}\ar[r]\ar[d]
  &{\ldots}\ar[r]
  &*[F]{\hole}\ar[d]
  &\\
  &{:}\ar[d]
  &{:}\ar[d]
  &{\ddots}
  &{:}\ar[d]
  &\\
  &*[F]{\hole}\ar[r]
  &*[F]{\hole}\ar[r]
  &{\ldots}\ar[r]
  &*[F]{\hole}
  & *{\hspace{-1em}\W[0][0]}
}
\end{minipage}
&\hspace{-3em}
\begin{minipage}{20em}{\small
\begin{align*}
      &\Pi n.\Pi m. \\
      &\FOREACH{i}{<n}{ \\
        &\quad \FOREACH{j}{<m}{ \\
          &\qquad \GS{\W[i+1][j+1]}{\W[i][j+1]}{\Nat}.\\
          &\qquad \GS{\W[i+1][j+1]}{\W[i+1][j]}{\Nat}}; \\
        &\quad \GS{\W[i+1][0]}{\W[i][0]}{\Nat}};\\
      &\FOREACH{k}{<m}{\GS{\W[0][k+1]}{\W[0][k]}{\Nat}}
   \end{align*}}
\end{minipage}
\end{tabular}
\vspace{-4.5ex}
\caption{Parameterised multiparty protocol on a mesh topology}\label{fig:examples}
\end{figure}

\subsection{Examples of parameterised global types}
\label{sub:globalexample}
\noindent                                                                                                                                              
We present some examples of global types that implement some communication
patterns specific to typical network topologies found in classical parallel
algorithms textbooks \cite{FThomson}.

\paragraph{\bf Ring - Figure~\ref{fig:examples}(a)}
The ring pattern consists of $n+1$ workers (named $\W[0]$, $\W[1]$,\ldots
,$\W[n]$) that each talks to its
two neighbours: the worker $\W[\ii]$ communicates with the worker $\W[\ii-1]$
and $\W[\ii+1]$ ($1\leq\ii\leq n-1)$, with the exception of $\W[0]$ and $\W[n]$
who share a direct link. 
%
The type specifies that the first message is sent by $\W[0]$ to $\W[1]$, and the
last one is sent from $\W[n]$ back to $\W[0]$. To ensure the presence of all
three roles in the workers of this topology, the parameter domain $I$ is set to
$n \geq 2$.



\paragraph{\bf Multicast - Figure~\ref{fig:examples}(b)}
The multicast session consists of $\Alice$ sending a message to $n$
workers \W. The first message is thus sent from \Alice\ to $\W[0]$,
then to $\W[1]$, until $\W[n-1]$.  
Note that, while the index $\ii$ bound by the iteration
$\FOREACH{i}{<n}{\GS{\Alice}{\W[n-1-\ii]}{\Nat}}$ decreases from 
$n-1$ to $0$, the index $\n-1-\ii$ in $\W[n-1-\ii]$ increases from $0$ to
$n-1$. 


\paragraph{\bf Mesh  - Figure~\ref{fig:examples}(c)}
The session presented in Figure~\ref{fig:examples}(c) describes a particular
protocol over a standard mesh topology~\cite{FThomson}. In this two dimensional
array of workers \W, each worker receives messages from his left and top
neighbours (if they exist) before sending messages to his right and bottom (if
they exist). Our session takes two parameters $n$ and $m$ which represent the
number of rows and the number of columns. Then we have two iterators that repeat
$\GS{\W[\ii+1][\jj+1]}{\W[\ii][\jj+1]}{\Nat}$ and
$\GS{\W[\ii+1][\jj+1]}{\W[\ii+1][\jj]}{\Nat}$ for all $\ii$ and $\jj$. 
The communication flow goes from the top-left worker $\W[n][m]$ and
converges towards the bottom-right worker $\W[0][0]$ 
in $n+m$ steps of asynchronous message exchanges.

\subsection{Process syntax}
\label{sec:usersyntax}

\begin{figure}[t]
\centering
\begin{tabular}{l@{\quad}|l}
\begin{tabular}{@{}r@{\;}c@{\;}l@{\quad}l}
\ccc & ::=  & \y \sep \s[\p] & Channels \\[1mm]
\Iu & ::=  & \x \bnfor \Ia & Identifiers\\[1mm]
\Iv & ::=  & $\Ia \bnfor \n$ 
 & Values\\[1mm]
\end{tabular} &
\begin{tabular}{r@{\;}c@{\;}l@{\quad}l}
 \pv, \qv & ::=  & $\pv[\n] \sep \Names$ & Principal values\\[1mm]
 \mm & ::=  & (\qv,\pv,\va) \sep  (\qv,\pv,\s[\pv'])  \sep
 (\qv,\pv,\Ll) 
& Messages in transit \\[1mm]
 \h  & ::=  & $\qbot \sep \qcomp{\mm}{\h}$ & Queue types \\
\end{tabular}\\[1mm] 
%
\multicolumn{2}{l}{ $e$   ::=  $\tii \bnfor v \bnfor \x \bnfor \s[\p] \bnfor
  \e\op\e' $ \hspace{1em} 
  Expressions}\\[1mm]
\end{tabular}
\begin{tabular}{@{}l@{}l}
\begin{tabular}{@{}r@{\,}cl@{\ }l}
 \PP & ::=  & & \hspace{-2em} Processes \\
     & \sep & \sr\uu{\p_0,..,\p_\n} \y\PP   &   {Init}\\
     & \sep & \ssa\uu\p\y\PP   &   {Accept}\\
     & \sep & \ny{\sj{a}\pv\s}   &   {Request}\\
     & \sep & \Pout{\ccc}{\p}{\e}{\PP}& Value sending\\
     & \sep & \Pin{\ccc}{\p}{\x}{\PP}& Value reception \\
     & \sep & \Psel{\ccc}{\p}{l}{\PP}& Selection \\
     & \sep & \Pbranch{\ccc}{\p}& Branching \\
     & \sep & \nuc{a}\PP       & Shared channel restriction
\end{tabular} &
\begin{tabular}{@{}r@{}cl@{\quad}l}
      & \sep & $\mu \X.\PP$ & {Recursion}\\
     & \sep & \inact & {Inaction}\\
     & \sep & \PP \pc \QQ  &{Parallel}\\
     & \sep & $\RECSEQP{\PP}{\ii}{\X}{\QQ}$ & {Primitive recursion}\\
     & \sep & \X &{Process variable}\\
     & \sep & (\PP\APP\tii) &{Application}\\
     & \sep & \nuc{\s} \PP & Session restriction \\
     & \sep & \s:\h & Queues \\[1mm]
\end{tabular}
\end{tabular}
\caption{Syntax for user-defined and run-time processes}\label{fig:syntax}
\end{figure}

The syntax of expressions and processes is given in Figure~\ref{fig:syntax},
extended from \cite{BettiniCDLDY08LONG}, adding the primitive recursion operator and
a new request process.  Identifiers $\Iu$ can be variables $\x$ or
channel names $\Ia$. Values $\Iv$ are either channels $\Ia$ or natural numbers $\n$.
Expressions $e$ are built out of indices
$\tii$, values $\Iv$, variables $\x$, session end points (for delegation) and
operations over expressions.
Participants $\pp$ can include indices which 
are substituted by values and evaluated during reductions (see 
the next subsection). 
In processes, sessions are asynchronously initiated by $\sr\uu{\p_0,..,\p_\n}\y\PP$. 
It spawns, for each of the $\{\p_0,..,\p_\n\}$, 
\footnote{Since the set of principals is parameterised, we allow some syntactic
sugar to express ranges of participants that depend on parameters.}
a request that is accepted by the participant through $\ssa\uu\p\y\PP$. 
Messages
are sent by $\Pout{\ccc}{\p}{\e}{\PP}$ to the participant $\p$ 
and received by $\Pin{\ccc}{\q}{\x}{\PP}$ from the participant $\q$.  
Selection $\Psel{\ccc}{\p}{l}{\PP}$, and branching $\Pbranch{\ccc}{\q}$, allow a
participant to choose a branch from those supported by another.  Standard
language constructs include recursive processes $\mu \X.P$, restriction
$\nuc{\Ia}{\PP}$ and
$\nuc{s}{\PP}$,  and
parallel composition $\PP \pc \QQ$.
The primitive recursion operator $\RECSEQP{\PP}{\ii}{\X}{\QQ}$ takes as
parameters a process $\PP$, a function taking an index parameter $\ii$ and a
recursion variable $\X$. A queue $\s:\h$ stores the asynchronous messages in
transit.


\label{par:annotated}
An {\em annotated} $P$ is the result of annotating $P$'s bound names and
variables by their types or ranges as in e.g.~$(\nu a\!:\!\ENCan{G})Q$ or
$s?(\pp,x\!:\! U);Q$ or 
$\RECSEQP{\QQ}{\ii\!:\!\II}{\X}{\QQ'}$.
We omit the annotations 
unless needed.  We often omit $\inact$ and the participant $\p$ from the session
primitives.  Requests, session restriction and channel queues appear only at
runtime, as explained below.

\subsection{Semantics}
\label{subsec:semantics}

\iffalse
\begin{figure}[t]
\centering
\small
\begin{tabular}{cr}
     \red{\Iv\op\Iv'}
     {\Iv''} \quad \text{with $\Iv''$ corresponding to $\op$}
     & [Op]
 \\[1.5mm]
  \red{\GR{P}{\ii}{X}{Q}\APP 0}{P} & [ZeroR]
  \\[1.5mm]
  \red{\GR{P}{\ii}{X}{Q}\APP \n+1}{P\sub{\n}{\ii}\sub{\GR{P}{\ii}{X}{Q}\APP \n}{X}} & [SuccR]
  \\[1.5mm]
        $\sr\Ia{\p_0,..,\p_\n}{\y}{\PP}\redsym (\nu \s)(
        s : \qbot \pc \PP\sub{\si\s {\p_0}}{\y} \pc \sj{\Ia}{\p_1}{\s} \pc ...\pc
        \sj{\Ia}{\p_\n}{\s})$  & [Init]
\\[1.5mm]
        $\sj{\Ia}{\p_k}{\s} \pc \ssa\Ia{\p_k}{\y_k}{\PP_k}\redsym 
        \PP_k\sub{\si \s {\p_k}}{\y_k}$  & [Join]
\\[1.5mm]
    \red{\outs{\si{\s}{\p}}{\va}{\q}{\PP} \pc \mqueue{\s}{\queue}}
    {\PP \pc\mqueue{\s}{\qtail{\valheap{\va}{\q}{\p}}}}
    & [Send]
\\[1.5mm]
    \red{\lsel{\sii}{l}{\q}{\PP} \Par \stdqueue}
    {\PP \Par \qappend{\labheap{l}{\q}{\p}}}
    & [Label]
\\[1.5mm]
    $\inp{\sii}{\x}{\q}{\PP} \Par \qpop{\valheap{\va}{\p}{\q}}$
    $\redsym \PP\subst{\ptilde{\va}}{\ptilde{\x}} \Par
    \mqueue{\s}{\queue}$ &
    [Recv]
\\[1.5mm]
    \lbranchk{\sii}{\q} \Par \qpop{\labheap{l_{k_0
}}{\p}{\q}}
    $\redsym \PP_{i_0
} \Par \mqueue{\s}{\queue}$ \ \ $(i_0 \in I)$ & [Branch]
\\[1.5mm]
   \red{\PP}{\PP'} \Implies \red{\PP\APP\e}{\PP'\APP\e}\quad \quad
   \red{\PP}{\PP'} \Implies
 \red{\nuc{\cas}{\PP}}{\nuc{\cas}{\PP'}}
   &[App,Scop]\\[1.5mm]
   \red{\PP}{\PP'} \Implies \red{\PP \Par \Q}{\PP' \Par \Q}
   &[Par]\\[1.5mm]
   $P\equiv P'\ \text{and}\ \red{P'}{Q'}\ \text{and}\ Q\equiv Q' \Implies
   \red{P}{Q}$ 
 & [Str]
 \\[1.5mm]
  $\red{\e}{\e'} \Implies 
 \red{\E[\e]}{\E[\e']}$  & [Context]
 \\[1.5mm]
\end{tabular}
\caption{Reduction rules (all principals $\p, \q, \p'$ in the rules above are
  principal values $\pv, \qv, \pv'$)}\label{fig:reduction}
\end{figure}

\else

\begin{figure}
\centering \small
\begin{tabular}{cr}
\small
    \red{
\GR{P}{\ii}{X}{Q}\APP 0
}
{
P
}
& [ZeroR]
\\[1.5mm]
 \red{\GR{P}{\ii}{X}{Q}\APP \n+1}
{Q\sub{\n}{\ii}\sub{\GR{P}{\ii}{X}{Q}\APP \n}{X}}
    & [SuccR]
\\[1.5mm]
        $\sr\Ia{\pv_0,..,\pv_\n}{\y}{\PP}\redsym (\nu \s)(
        \PP\sub{\si\s {\pv_0}}{\y} \pc s : \qbot \pc \sj{\Ia}{\pv_1}{\s} \pc ...\pc
        \sj{\Ia}{\pv_\n}{\s})$  & [Init]
\\[1.5mm]
        $\sj{\Ia}{\pv_k}{\s} \pc \ssa\Ia{\pv_k}{\y_k}{\PP_k}\redsym 
        \PP_k\sub{\si \s {\pv_k}}{\y_k}$  & [Join]
\\[1.5mm]
    \red{\outs{\si{\s}{\pv}}{\va}{\qv}{\PP} \pc \mqueue{\s}{\queue}}
    {\PP \pc\mqueue{\s}{\qtail{\valheap{\va}{\qv}{\pv}}}}
    & [Send]
\\[1.5mm]
    \red{\lsel{\siiv}{l}{\qv}{\PP} \Par \stdqueue}
    {\PP \Par \qappend{\labheap{l}{\qv}{\pv}}}
    & [Label]
\\[1.5mm]
    $\inp{\siiv}{\x}{\qv}{\PP} \Par \qpop{\valheap{\va}{\pv}{\qv}}$
    $\redsym \PP\subst{\ptilde{\va}}{\ptilde{\x}} \Par
    \mqueue{\s}{\queue}$ &
    [Recv]
\\[1.5mm]
    \lbranchk{\siiv}{\qv} \Par \qpop{\labheap{l_{k_0
}}{\pv}{\qv}}
    $\redsym \PP_{k_0
} \Par \mqueue{\s}{\queue}$ \ \ $(k_0 \in K)$ & [Branch]
\\[1.5mm]
   \red{\PP}{\PP'} \Implies \red{\PP\APP\e}{\PP'\APP\e}\quad \quad
   \red{\PP}{\PP'} \Implies
 \red{\nuc{\cas}{\PP}}{\nuc{\cas}{\PP'}}\quad \quad \quad 
   &[App,Scop]\\[1.5mm]
   \red{\PP}{\PP'} \Implies \red{\PP \Par \Q}{\PP' \Par \Q}
   &[Par]\\[1.5mm]
   $P\equiv P'\ \text{and}\ \red{P'}{Q'}\ \text{and}\ Q\equiv Q' \Implies
   \red{P}{Q}$ 
 & [Str]
 \\[1.5mm]
  $\red{\e}{\e'} \Implies \red{\E[\e]}{\E[\e']}$  & [Context]
 \\[1.5mm]
\end{tabular}
\caption{Reduction rules 
}\label{fig:reduction}
\end{figure}

\fi

\begin{figure}[t]
\centering
\begin{tabular}{r@{\ }c@{\ }l@{\quad}l}
  \E[\_, \ldots, \_] & ::= &  &\hspace{-2em} Evaluation contexts \\
   & \sep & $\_\op \ e$ \sep $v \op\_$ & Expression\\
   & \sep & (\PP\APP\_) & Application\\
   & \sep & $\sr\Ia{\pv_1,..,\pv_n,\_, \p_{n+1},..,\p_m}{\y}{\PP} $ & Request \\
   & \sep & $\ssa\Ia{\_}{\y}{\PP}$ & Accept \\
   & \sep & $\outs{\si{\s}{\_}}{e}{\p}{\PP}$ 
            \sep
            $\outs{\si{\s}{\pv}}{e}{\_}{\PP}$ 
            \sep
            $\outs{\si{\s}{\pv}}{\_}{\qv}{\PP}$ 
            & Send\\
   & \sep & $\lsel{\si{\s}{\_}}{l}{\p}{\PP} $ 
\sep $\lsel{\si{\s}{\pv}}{l}{\_}{\PP} $ 
& Selection\\
   & \sep & $\inp{\si{\s}{\_}}{\x}{\p}{\PP}$ 
\sep $\inp{\si{\s}{\pv}}{\x}{\_}{\PP} $
& Receive\\
   & \sep & $\lbranchk{\si{\s}{\_}}{\p} $ 
     \sep   $\lbranchk{\si{\s}{\pv}}{\_} $ 
& Branching\\
\end{tabular}
\caption{\ny{Evaluation contexts}}\label{fig:context}
\end{figure}

The semantics is defined by the reduction relation $\red{}{}$ presented in
Figure~\ref{fig:reduction}. The standard definition of evaluation contexts (that
allow 
e.g. $\W[3+1]$ to be reduced to $\W[4]$) is in Figure~\ref{fig:context}.
The metavariables $\pv,\qv,..$ range over principal values (where all indices
have been evaluated).
%
Rules [ZeroR] and [SuccR] are standard and
identical to their global type counterparts.  Rule [Init] describes the
initialisation of a session by its first participant
$\sr\Ia{\p_0,..,\p_\n}{\y_0}{\PP_0}$. It spawns asynchronous requests
$\sj{\Ia}{\pv_k}{\s}$ that allow delayed acceptance by the other session
participants (rule [Join]).  After the connection, the participants share the
private session name \s, and the queue associated to \s~(which is initially
empty by rule [Init]).  The variables $\y_\p$ in each participant $\p$ are then
replaced with the corresponding session channel, $\si{\s}{\p}$. An equivalent,
but symmetric, version of [Init] (where any participant can start the session, 
not only $\p_0$) can be also used. Rule [Init] would then be replaced by the
following:
$$
\bar{\Ia}[\pv_0,..,\pv_\n] \redsym (\nu \s)(
        s : \qbot \pc \sj{\Ia}{\pv_0}{\s} \pc ...\pc
        \sj{\Ia}{\pv_\n}{\s})
$$
%



\noindent  The rest of the session reductions are standard \cite{BettiniCDLDY08LONG,CHY07}.  The
output rules [Send] 
and [Label] push values, channels and labels into the queue of the session
\s{}.  Rules [Recv] 
and [Branch] perform the complementary operations.  Note that these operations
check that the sender and receiver match.
Processes are considered modulo structural equivalence, denoted
by $\equiv$ (in particular, we note 
$\mu X.P \equiv P\subst{\mu X.P}{X}$),  
whose definition is found in Figure~\ref{tab:structcong}.  
Besides the standard rules
\cite{MilnerR:commspc}, we have
a rule for rearranging messages when the senders or the receivers
are different, and a rule for the garbage-collection of unused and empty
queues. 

\begin{figure}[h!]
\begin{tabular}{c}
  $P \Par \textbf{0}\equiv P \quad
  P \Par Q\equiv Q \Par P \quad
  (P \Par Q) \Par R\equiv P \Par (Q \Par R) \quad
  (\nu \cas\cas')\ P\equiv (\nu \cas'\cas)\ P$ \\[2mm]
  $(\nu \cas)\ \textbf{0}\equiv \textbf{0} \quad 
  (\nu s)\ s:\qbot \equiv \textbf{0} \quad 
  (\nu \cas)\ P \Par Q\equiv (\nu \cas)\ (P \Par Q)\ \ \ \ \ \ \text{if}\ \cas\notin
  \freen{\Q}$\\[2mm]
   $\mqueue{\s}{
   \qhead
    {
        \qcomp
        {\trival{z}{\pv}{\qv}}
        {\trival{z'}{\pv'}{\qv'}}
    }
   } \equiv
   \mqueue{\s}{
   \qhead
    {
        \qcomp
        {\trival{z'}{\pv'}{\qv'}}
        {\trival{z}{\pv}{\qv}}
    }
   } \qquad \text{if $\pv \ne \pv'$ or $\qv \ne \qv'$}$\\[2mm]
   
   $\mu X. \PP \equiv \PP\sub{  \mu X. \PP}{X}$\\[2mm]
\end{tabular}
\\
\cas\ ranges over \Ia, \s. 
\quad $z$ ranges over \at{\va}, \si\s\pv\ and $l$.
\caption{Structural equivalence}\label{tab:structcong}
\end{figure}


\subsection{Processes for parameterised multiparty protocols}
\label{subsec:processexample}
\noindent 
We give here the processes corresponding to the interactions described in
\S~\ref{sec:introduction} and \S~\ref{subsec:examples}, then introduce a
parallel implementation of the Fast Fourier Transform algorithm.
There are various ways to implement end-point 
processes from a single global type, and we show one instance for each example
below.

\paragraph{\bf Repetition}
A concrete definition for the protocol (\ref{alice-bob-carol}) in \S~1 
is:
\[
\begin{array}{l}
\Pi n.(\RECSEQ{\End}{i}{\xx}{}\GS{\Alice}{\Bob}{\Nat}.
\GS{\Bob}{\Carol}{\Nat}.\xx \APP n)
\end{array}
\]
Then $\Alice$ and $\Bob$ 
can be implemented
with recursors as follows (we abbreviate 
\Alice~by $\participant{a}$, \Bob~by $\participant{b}$ and 
\Carol~by $\participant{c}$). 
\[
\begin{array}{lll}
\Alice(\nn) = \sr{a}{\participant{a},\participant{b},\participant{c}}
\y(\RECSEQP{\inact}{\ii}{\X}{\outS{\y}{\participant{b},\e[i]}}\X\APP n)\\[1mm]
\Bob(\nn) = \ssa{a}{\participant{b}}
\y(\RECSEQP{\inact}{\ii}{\X}{\inpS{y}{\participant{a},z}\outS{\y}{\participant{c},\z}}\X\APP
n)\\[1mm]
\Carol(\nn) = \ssa{a}{\participant{c}}
\y(\RECSEQP{\inact}{\ii}{\X}{\inpS{\y}{\participant{b},z}}\X\APP n)
\end{array}
\]
$\Alice$ repeatedly sends a message $e[i]$ to $\Bob$ $n$-times.
Then $\nn$ can be bound by $\lambda$-abstraction, allowing 
the user to dynamically assign the number of the repetitions.  
\begin{center}
{\small
$
\begin{array}{l}
\lambda \nn.(\nuc{\Ia}{(\Alice(\nn) \pc \Bob(\nn) \pc
\Carol(\nn))}) \ 1000
\end{array}
$}
\end{center}
\paragraph{\bf Sequence from \S~\ref{sec:introduction} (\ref{ex:sequence})}
The process below generates all participants using a recursor:
{\small
\[
\begin{array}{llrlll}
\Pi \nn.(\Pifthenelse{\nn=0 & }{ & \inact & \\
& }{ &(\mathbf{R} &
(\sr{a}{\W[\nn],..,\W[0]}{y}{\outS{\y}{\W[\nn-1],v}\inact}
\\[1mm]
& & & \pc \ssa{a}{\W[0]}{y}{\inpS{y}{\W[1],z}{\inact}})\\[1mm]
& & & \lambda {\ii}.\lambda{\X}.(
\ssa{a}{\W[\ii+1]}{y}{\inpS{y}{\W[\ii+2],z}\outS{\y}{\W[\ii],z}\inact \pc X}) {\quad \nn-1})}
\end{array}
\]}
\noindent When $n=0$ no message is exchanged. In the other case, the recursor
creates the $n-1$ workers through the main loop and finishes by spawning the initial and final ones.

As an illustration of the semantics, we show the reduction of the above process
for $n=2$. After several applications of the [SuccR] and [ZeroR] rules, we have:
{\small
\[
\begin{array}{l}
\sr{a}{\W[2],\W[1],\W[0]}{y}{\outS{\y}{\W[1],v}}\inact \pc
\ssa{a}{\W[0]}{y}{\inpS{y}{\W[1],z}}\inact \pc
\ssa{a}{\W[1]}{y}{\inpS{y}{\W[2],z}\outS{\y}{\W[0],z}}\inact
\end{array}
\]
} which, with [Init], [Join], [Send], [Recv], gives:
{\small
\[
\begin{array}{ll}
\red{}{} &(\nu \s)(s : \qbot \pc \outS{\si\s {\W[2]}}{\W[1],v}\inact \pc \sj{\Ia}{\W[1]}{\s}\pc \sj{\Ia}{\W[0]}{\s}\pc \\
 & \qquad
\ssa{a}{\W[0]}{y}{\inpS{y}{\W[1],z}}\inact \pc \ssa{a}{\W[1]}{y}{\inpS{y}{\W[2],z}\outS{\y}{\W[0],z}}\inact)\\
\red{}{} & (\nu \s)(s : \qbot \pc \outS{\si\s {\W[2]}}{\W[1],v}\inact \pc \sj{\Ia}{\W[1]}{\s}\pc\\
 &  \qquad 
\inpS{\si\s {\W[0]}}{\W[1],z}\inact \pc
\ssa{a}{\W[1]}{y}{\inpS{y}{\W[2],z}\outS{\y}{\W[0],z}}\inact)\\
\red{}{}^* & (\nu \s)(s : \qbot \pc \outS{\si\s {\W[2]}}{\W[1],v}\inact \pc 
\inpS{\si\s {\W[0]}}{\W[1],z}\inact \pc
\inpS{\si\s {\W[1]}}{\W[2],z}\outS{\si\s {\W[1]}}{\W[0],z}\inact)\\
\red{}{}^* &
(\nu \s)(s :  \qbot \pc \inpS{\si\s {\W[0]}}{\W[1],z}\inact \pc
\outS{\si\s {\W[1]}}{\W[0],v}\inact) \\
\red{}{}^* & \equiv \inact
\end{array}
\]
}



\paragraph{\bf Ring - Figure \ref{fig:examples}(a)}
The process that generates all the
roles using a recursor is as follows:
{\small
\[
\begin{array}{lllll}
\Pi \nn.(\mathbf{R} &
\sr{a}{\W[0], ..., \W[\nn]}{y}{\outS{\y}{\W[1],v}\inpS{y}{\W[\nn],z}P}\\[1mm]
& \ssa{a}{\W[\nn]}{y}{\inpS{y}{\W[\nn-1],z}{\outS{\y}{\W[0],z}Q}}\\[1mm]
& \quad \lambda {\ii}.\lambda{\X}.(
\ssa{a}{\W[\ii+1]}{y}{\inpS{y}{\W[\ii],z}\outS{\y}{\W[\ii+2],z} \pc X}) {\quad \nn-1} )
\end{array}
\]
}
We take the range of $\nn$ to be $\nn\geq 2$.

\paragraph{\bf Mesh - Figure~\ref{fig:examples} (c)}
In this example, 
when \n~and \m~are bigger than $2$, there are
9 distinct patterns of communication.  

We write below these processes. We assume the existence of a function
$f(z_1,z_2,i,j)$ which computes from $z_1$ and $z_2$ the value to be transmitted
to $\W[\ii][\jj]$. We then designates the processes based on their position in
the mesh. The initiator $\W[\nn][\mm]$ is in the top-left corner of the mesh and
is implemented by $\PP_{\text{top-left}}$. The workers that are living in the
other corners are implemented by $\PP_{\text{top-right}}$ for $\W[\nn][0]$,
$\PP_{\text{bottom-left}}$ for $\W[0][\mm]$ and $\PP_{\text{bottom-right}}$ for
the final worker $\W[0][0]$. The processes $\PP_{\text{top}}$,
$\PP_{\text{left}}$, $\PP_{\text{bottom}}$ and $\PP_{\text{right}}$ respectively
implement the workers from the top row, the leftmost column, the bottom row and
the rightmost column. The workers that are in the central part of the mesh are
played by the $\PP_{\text{center}}(\ii, \jj)$ processes.\\[2mm]
{\small
$\begin{array}{ll}
\PP_{\text{top-left}}(z_1,z_2,\nn, \mm)  & =  \sr{a}{\W[\nn][\mm], ..., \W[0][0]}{y}{
\outS{y}{\W[n-1][m], f(z_1,z_2,n-1,m)}}\\
 & ~~~\outS{y}{\W[n][m-1], f(z_1,z_2,n,m-1)}{\inact}\\[1mm]
\PP_{\text{top-right}}(z_2,\nn) & = \ssa{a}{\W[n][0]}{y}{\inpS{y}{\W[n][1], z_1}\outS{y}{\W[n-1][0], f(z_1,z_2,n-1,0)}\inact}\\[1mm]
\PP_{\text{bottom-left}}(z_1,\mm) & =  \ssa{a}{\W[0][m]}{y}{\inpS{y}{\W[1][m], z_2}\outS{y}{\W[0][m-1], f(z_1,z_2,0,m-1)}\inact}\\[1mm]
\PP_{\text{bottom-right}}(\mm) & =   \ssa{a}{\W[0][0]}{y}{\inpS{y}{\W[1][0], z_1}\inpS{y}{\participant{\W[0][1]}, z_2}\inact}\\[1mm]
\PP_{\text{top}}(z_2,\nn, \kk) & = 
\ssa{a}{\W[n][k+1]}{y}{\inpS{y}{\W[n][k+2], z_1}}\\
 & ~~~ \outS{y}{\W[n-1][k+1], f(z_1,z_2,n-1,k+1)}\outS{y}{{\W[n][k]}, f(z_1,z_2,n,k)}\inact\\[1mm] 
\PP_{\text{bottom}}(\kk) & = \ssa{a}{\W[0][k+1]}{y}{\inpS{y}{\W[1][k+1], z_1}\inpS{y}{\W[0][k+2], z_2}}\\
& ~~~ \outS{y}{\W[0][k], f(z_1,z_2,0,k)}\inact\\[1mm]
\PP_{\text{left}}(z_1,\mm, \ii) & = \ssa{a}{\W[i+1][m]}{y}{\inpS{y}{{\W[i+2][m]}, z_2}\outS{y}{{\W[i][m]}, f(z_1,z_2,i,m)}}\\
& ~~~ \outS{y}{{\W[i+1][m-1]}, f(z_1,z_2,i+1,m-1)}\\[1mm] 
\PP_{\text{right}}(\ii) & =  \ssa{a}{\W[i+1][0]}{y}{\inpS{y}{{\W[i+2][0]}, z_1}\inpS{y}{{\W[i+1][1]}, z_2}}\\
& ~~~ \outS{y}{{\W[i][0]}, f(z_1,z_2,i,0)}\inact\\[1mm]
\PP_{\text{center}}(\ii, \jj) & = \ssa{a}{\W[i+1][j+1]}{y}{\inpS{\y}{\W[i+2][j+1], z_1} \inpS{\y}{\W[i+1][j+2], z_2}}\\
 &~~~ \outS{y}{\W[i][j+1], f(z_1,z_2,i,j+1)}\outS{y}{\W[i+1][j], f(z_1,z_2,i+1,j)}\inact\\[1mm]
\end{array}$
}

The complete implementation can be generated using the following process:

{\small
\[
\begin{array}{l}
\Pi \nn.\Pi \mm.(\mathbf{R} ~(\mathbf{R}~ 
\PP_{\text{top-left}}(z_1,z_2,\nn, \mm) | \PP_{\text{bottom-right}}(\mm)  | \PP_{\text{top-right}}(z_2,\nn)  | \PP_{\text{bottom-left}}(z_1,\mm))\\
\qquad\qquad\qquad \lambda \kk.\lambda Z.(\PP_{\text{top}}(z_2,\nn, \kk) | \PP_{\text{bottom}}(\kk) | Z)\\
\qquad\qquad\quad m-1)\\
\qquad\qquad~~\lambda \ii.\lambda \X. (\mathbf{R}~ \PP_{\text{left}}(z_1,\mm, \ii) | \PP_{\text{right}}(\ii) | \X \\[1mm]
\qquad\qquad\qquad\qquad~~  \lambda \jj.\lambda \Y.(\PP_{\text{center}}(\ii, \jj) | \Y)\\
\qquad\qquad\qquad\qquad \mm-1)\\
\qquad\qquad \nn-1)\\
\end{array}
\]
}

\begin{figure}[t]
\centering
\footnotesize
\begin{tabular}{@{}c@{}l@{}}
\begin{minipage}{17.5em}
{\bf \textsf (a) Butterfly pattern}

\hspace{0em}
\xymatrix@C=15pt@R=15pt{
{\text{\footnotesize $x_{k-N/2}$}}
\ar[dr]\ar@{.>}[r]
 & {\hole}\save[]+<1.3cm,-0.2cm>*\txt<4cm>{\small
   $X_{k-N/2}=x_{k-N/2}+$\\ \hspace{6em} $x_{k}*\omega_N^{k-N/2}$}\restore\\
{\text{\footnotesize $x_{k}$}}\ar[ur]
\ar@{.>}[r] 
 & {\hole}\save[]+<1.6cm,0cm>*{\text{\small 
   $X_k=x_{k-N/2}+x_{k}*\omega_N^{k}$}}\restore\\
}

\smallskip
{\bf \textsf (b) FFT diagram}\\ 
$ 
\begin{minipage}{15em}
\centering 
\xymatrix@C=18pt@R=15pt{
  \ar[r]^ {x_0}
  & *+[o][F]{0} \ar[dr]\ar@{.>}[r]^1
  & *+[o][F]{0} \ar[ddr]\ar@{.>}[r]^2
  & *+[o][F]{0} \ar[ddddr]\ar@{.>}[r]^3
  & *+[o][F]{0} \ar[r]^{X_0}
  & \\
  \ar[r]^{x_4} 
  & *+[o][F]{1} \ar[ur]\ar@{.>}[r]
  & *+[o][F]{1} \ar[ddr]\ar@{.>}[r]
  & *+[o][F]{1} \ar[ddddr]\ar@{.>}[r]
  & *+[o][F]{1} \ar[r]^{X_1}
  & \\
  \ar[r]^{x_2} 
  & *+[o][F]{2} \ar[dr]\ar@{.>}[r]
  & *+[o][F]{2} \ar[uur]\ar@{.>}[r]
  & *+[o][F]{2} \ar[ddddr]\ar@{.>}[r]
  & *+[o][F]{2} \ar[r]^{X_2}
  & \\
  \ar[r]^{x_6} 
  & *+[o][F]{3}  \ar[ur]\ar@{.>}[r]
  & *+[o][F]{3}  \ar[uur]\ar@{.>}[r]
  & *+[o][F]{3} \ar[ddddr]\ar@{.>}[r]
  & *+[o][F]{3} \ar[r]^{X_3}
  & \\
  \ar[r]^{x_1} 
  & *+[o][F]{4}  \ar[dr]\ar@{.>}[r]
  & *+[o][F]{4}  \ar[ddr]\ar@{.>}[r]
  & *+[o][F]{4} \ar[uuuur]\ar@{.>}[r]
  & *+[o][F]{4} \ar[r]^{X_4}
  & \\
  \ar[r]^{x_5} 
  & *+[o][F]{5}  \ar[ur]\ar@{.>}[r]
  & *+[o][F]{5} \ar[ddr]\ar@{.>}[r]
  & *+[o][F]{5} \ar[uuuur]\ar@{.>}[r]
  & *+[o][F]{5} \ar[r]^{X_5}
  & \\
  \ar[r]^{x_3} 
  & *+[o][F]{6}  \ar[dr]\ar@{.>}[r]
  & *+[o][F]{6}  \ar[uur]\ar@{.>}[r]
  & *+[o][F]{6}  \ar[uuuur]\ar@{.>}[r]
  & *+[o][F]{6} \ar[r]^{X_6}
  & \\
  \ar[r]^{x_7} 
  & *+[o][F]{7}  \ar[ur]\ar@{.>}[r]
  & *+[o][F]{7}  \ar[uur]\ar@{.>}[r]
  & *+[o][F]{7} \ar[uuuur]\ar@{.>}[r]
  & *+[o][F]{7} \ar[r]^{X_7}
  & \\
}
\end{minipage} 
$
\end{minipage}

&

\begin{minipage}{30em}
\vspace{-6ex}
{\bf \textsf (c) Global type} \ {$\G=$
\vspace{-1ex}
$$\begin{array}{l@{\hspace{-4em}}l}
&\Pi n.  \\
& \FOREACH{\ii}{<2^n}{\GS{i}{i}{\Nat}};\\
& \FOREACH{l}{<n}{\\
  & \hspace{0.5em} \FOREACH{i}{<2^l}{ \\
    & \hspace{0.9em} \FOREACH{j}{<2^{n-l-1}}{ \\
      &  \hspace{1.2em} \GS{i*2^{n-l}+j}{i*2^{n-l}+2^{n-l-1}+j}{\Nat}  \\
      &  \hspace{1.2em} \GS{i*2^{n-l}+2^{n-l-1}+j}{i*2^{n-l}+j}{\Nat}  \\
      &  \hspace{1.2em} \GS{i*2^{n-l}+j}{i*2^{n-l}+j}{\Nat}  \\
      &  \hspace{1.2em} \GS{i*2^{n-l}+2^{n-l-1}+j}{i*2^{n-l}+2^{n-l-1}+j}{\Nat} 
}}}\\
\end{array}$$
}


{\bf \textsf (d) Processes} \ {\small $\PP(n,\pp,x_{\overline{\pp}},y,r_{\p}) =$}
\vspace{-1ex}
{\small
$$\begin{array}{@{\hspace{-3em}}l} 
 \ \Pout{\y}{\p}{\x_{\overline{\pp}}}{}\\
 \ \FOREACH{l}{<n}{\\
 \quad
    \Pifthenelse{\bit{n-l}(\pp)=0\\
 \quad}{\Pin{\y}{\p}{\x}{\Pout{\y}{\p+2^{n-l-1}}{\x}{\\
 \hspace{3.3em}\Pin{y}{\p+2^{n-l-1}}{\z}{\Pout{\y}{\p}{\x+\z\,\omega_N^{g(l,\pp)}}{}}}}\\
\quad}{\Pin{\y}{\p}{\x}{\Pin{y}{\p-2^{n-l-1}}{\z}{\\
\hspace{3.3em}\Pout{\y}{\p-2^{n-l-1}}{\x}
      {\Pout{\y}{\p}{\z+\x\,\omega_N^{g(l,\pp)}}{}}}}}};\\
 \ \Pin{\y}{\p}{\x}{\Pout{r_{\p}}{0}{x}}\inact
\end{array}$$}
\vspace{-2ex}

\noindent where $g(l,\pp)=\pp \mod 2^l$
\end{minipage}
\end{tabular}
\caption{Fast Fourier Transform on a butterfly network topology}\label{fig:fft}
\end{figure}

\subsection{Fast Fourier Transform}
\label{subsec:fft}
\noindent 
We describe a
parallel implementation of the Fast Fourier Transform algorithm (more
precisely the radix-2 variant of the Cooley-Tukey 
algorithm~\cite{CT65}).
%
%
We start by a quick reminder of the discrete fourier transform definition,
followed by the description of an FFT algorithm that implements it over a
butterfly network. We then give the corresponding global session type. From the
diagram in (b) and the session type from (c), it is finally straightforward to
implement the FFT as simple interacting processes.

\paragraph{\bf The Discrete Fourier Transform}  
The goal of the FFT is to compute the Discrete Fourier Transform
(DFT) of a vector of complex numbers.
Assume the input consists in $N$ complex numbers $\vec{x}=x_0, \ldots,
x_{N-1}$ that can be interpreted as the coefficients of a polynomial
$f(y)=\sum_{j=0}^{N-1}x_j\,y^j$.  The DFT transforms $\vec{x}$ in a vector
$\vec{X}=X_0, \ldots, X_{N-1}$ defined by:
$$
X_k = f(\omega_N^k) 
$$\\[-3ex]
\noindent
with $\omega_N^k=\e^{ \imath\frac{ 2 k\pi }{N}}$ one of the $n$-th primitive
roots of unity. The DFT can be seen as a polynomial interpolation on the
primitive roots of unity or as the application of the square matrix
$(\omega_N^{ij})_{i,j}$ to the vector $\vec{x}$.
%

\paragraph{\bf FFT and the butterfly network} 
We present 
the radix-2 variant of the Cooley-Tukey 
algorithm~\cite{CT65}. 
It uses a divide-and-conquer strategy based on the following equation (we use
the fact that $\omega_N^{2k}=\omega_{N/2}^k$):

\[
\begin{array}{rcl}
  X_k& = & \sum_{j=0}^{N-1}x_j\,\omega_N^{jk} \\[1mm]
  & = &\sum_{j=0}^{N/2-1}x_{2j}\,\omega_{N/2}^{jk}
  + \omega_N^k\sum_{j=0}^{N/2-1}x_{2j+1}\,\omega_{N/2}^{jk}
\end{array}
\]
Each of the two separate sums are DFT of half of the original vector members,
separated into even and odd. Recursive calls can then divide the input set
further based on the value of the next binary bits. The good complexity of this
FFT algorithm comes from the lower periodicity of $\omega_{N/2}$: we have
$\omega_{N/2}^{jk}=\omega_{N/2}^{j(k-N/2)}$ and thus computations of $X_k$ and
$X_{k-N/2}$ only differ by the multiplicative factor affecting one of the two
recursive calls.

Figure~\ref{fig:fft}(a) illustrates this
recursive principle, called {\em butterfly}, where two
different intermediary values can be computed in constant time from the results
of the same two recursive calls. 

The complete algorithm is illustrated by the diagram from
Figure~\ref{fig:fft}(b).  It features the application of the FFT on a network of
$N=2^3$ machines on an hypercube network computing the discrete Fourier
transform of vector $x_0, \ldots, x_7$.  Each row represents a single machine at
each step of the algorithm. Each edge represents a value sent to another
machine. The dotted edges represent the particular messages that a machine sends
to itself to remember a value for the next step.  Each machine is successively
involved in a butterfly with a machine whose number differs by only one
bit. Note that the recursive partition over the value of a different bit at each
step requires a particular bit-reversed ordering of the input vector: the
machine number $\pp$ initially receives $x_{\overline{\pp}}$ where
$\overline{\pp}$ denotes the bit-reversal of $\pp$.





\paragraph{\bf Global Types}
Figure~\ref{fig:fft}(c) gives the global session type corresponding to the
execution of the FFT.  The size of the network is specified by the index
parameter $n$: for a given $n$, $2^n$ machines compute the DFT of a vector of
size $2^n$. The first iterator
$\FOREACH{\ii}{<2^n}{\GS{i}{i}{\Nat}};$
concerns the initialisation:
each of the machines sends the $x_\pp$ value to themselves. Then we have an
iteration over variable $l$ for the $n$ successive steps of the algorithm. The
iterators over variables $i,j$ work in a more complex way: at each step, the
algorithm applies the butterfly pattern between pairs of machines whose numbers
differ by only one bit (at step $l$, bit number $n-l$ is concerned). The iterators
over variables $i$ and $j$ thus generate all the values of the other bits: for
each $l$, $i*2^{n-l}+j$ and $i*2^{n-l}+2^{n-l-1}+j$ range over all pairs of
integers from $2^n-1$ to $0$ that differ on the $(n-l)$th bit. The four repeated
messages within the loops correspond to the four edges of the
butterfly pattern.

\paragraph{\bf Processes}
The processes that are run on each machine to execute the FFT algorithm are
presented in Figure~\ref{fig:fft}(d). 
When $\pp$ is the machine number, $x_{\overline{\pp}}$ the initial value, and
$\y$ the session channel, the machine starts by sending $x_{\overline{\pp}}$ to
itself: $\y!\langle\x_{\overline{\pp}}\rangle$. The main loop corresponds to the
iteration over the $n$ steps of the algorithm. At step $l$, each machine is
involved in a butterfly corresponding to bit number $n-l$, i.e. whose number
differs on the $(n-l)$th bit. In the process, we thus distinguish the two cases
corresponding to each value of the $(n-l)$th bit (test on $\bit{n-l}(\pp)$). In
the two branches, we receive the previously computed value $\inpS{\y}\x{..}$,
then we send to and receive from the other machine (of number $\p+2^{n-l-1}$ or
$\p-2^{n-l-1}$, i.e. whose $(n-l)$th bit was flipped). We finally compute the
new value and send it to ourselves: respectively by
$\outS{\y}{\x+\z\,\omega_N^{g(l,\pp)}}{\X}$ or
$\outS{\y}{\z+\x\,\omega_N^{g(l,\pp)}}{\X}$. Note that the two branches do not
present the same order of send and receive as the global session type specifies
that the diagonal up arrow of the butterfly comes first. At the end of the
algorithm, the calculated values are sent to some external channels:
$\Poutend{r_\pp}{0}{\x}$.

\section{Typing parameterised multiparty interactions}
\label{sec:typing}
\noindent 
This section introduces the type system, by which we can statically type
parameterised global specifications.

\subsection{End-point types and end-point projections}
\label{subsec:endpoint}
\begin{figure}[ht]
\begin{tabular}{lll}
\begin{tabular}{lllllll}
\T & ::=   &\hspace{-2em}End-point types \\ 
   & \sep \ \Lout{\p}{\U}{\T} & Output \\
   & \sep \ \Lin{\p}{\U}{\T} & Input \\
   & \sep  \ \Lsel{\p}{\T_k} \quad & Selection \\
   & \sep  \ \Lbranch{\p}{\T_k} & Branching\\
\end{tabular}
&
\quad 
\begin{tabular}{lllllll}
   & \sep  \ \LM{\xx}{\T} & Recursion \\
   & \sep  \ $\LR{\T}{\AT{\ii}{\II}}{\xx}{\T'}$  \quad & Primitive recursion \\
   & \sep  \ \xx & Type variable \\
   & \sep  \ \T\APP \tii  & Application\\
   & \sep  \ \End  & End \\ 
\end{tabular}
\end{tabular}
\caption{End-point types}\label{fig:local}
\end{figure}

\noindent 
A global type is projected to an {\em end-point type} according to each
participant's viewpoint.
The syntax of end-point types is given in 
Figure~\ref{fig:local}. 
{\rm Output} expresses the sending to
$\p$ of a value or channel
of type \UT, followed by the interactions \T.
{\rm Selection} represents the transmission to 
$\p$ of a label $l_k$ chosen in $\{l_k\}_{k\in K}$ 
followed by $\T_k$. 
{\rm Input} and {\rm branching} are their dual counterparts.
The other types are similar to their global versions.

\paragraph{\bf End-point projection: a generic projection}
The relation between end-point types and global types is formalised by the projection
relation. 
Since the actual participant characteristics might only be determined at
runtime, we cannot straightforwardly use the definition
from~\cite{CHY07,BettiniCDLDY08LONG}.  %
Instead, we rely on the expressive power of the primitive recursive operator:
{\em a generic end-point projection of \G\ onto \q}, written \pro\G\q,
represents the family of all the possible end-point types that a principal $\q$ can
satisfy at run-time. %

\begin{figure}[ht]\small
\centering
\begin{tabular}{@{}rl@{}}
$\GS{\p}{\p'}{\U}.\G$\proj{\qq} \ = 
&
\IF\ \qq=\pp=\pp'\ \THEN\ \Lout{\p}{\U}{\Lin{\p}{\U}{\G\proj{\qq}}}\\
 & \ELSE\IF\ \qq=\pp\ \THEN\ \Lout{\p'}{\U}{\G\proj{\qq}}\\
 & \ELSE\IF\ \qq=\pp' \THEN\ \Lin{\p}{\U}{\G\proj{\qq}}\\
&  \ELSE\ {\G\proj{\qq}} \\[1mm]
 $\GB{\p}{\p'}$\proj{\qq} \ = 
&  \IF\ \qq=\pp\ \THEN\ \Lsel{\p'}{\G_k\proj{\qq}}\\
 & \ELSE\IF\ \qq=\pp' \THEN\ \Lbranch{\p}{\G_k\proj{\qq}} \\
 & \ELSE\ $\sqcup_{k\in K}  \G_k\proj{\qq}$ \\[1mm]
 $\GR{\G}{\AT{\ii}{\II}}{\xx}{\G'}$\proj{\qq} \ = 
 & $\GR{\G\proj{\qq}}{\AT{\ii}{\II}}{\xx}{\G'\proj{\qq}}$\\
 $(\GM{\xx}{\G})$\proj{\pp} \ =  
&  $\GM{\xx}{\G\proj{\pp}}$ \\
 \xx \proj{\pp} \ =  
&  \xx \\
 (\G\APP \tii) \proj{\pp}\ = 
& (\G\proj{\pp})\APP \tii\\
\End\proj{\pp} \ = 
& \End 
\end{tabular}
\caption{Projection of global types to end-point types}\label{fig:projection}
\end{figure}

The general endpoint generator is defined in Figure~\ref{fig:projection} using
the derived condition construct $\Pifthenelse{\_}{\_}{\_}$.  The projection
$\GS{\p}{\p'}{\U}.\G\proj{\qq}$ leads to a case analysis: if the participant
$\q$ is equal to $\p$, then the end-point type of $\q$ is an output of type $\U$
to $\p'$; if participant $\q$ is $\p'$ then $\q$ inputs $U$ from $\p'$; else we
skip the prefix. The first case corresponds to the possibility for the sender
and receiver to be identical.  
Projecting the branching global type is similarly
defined, but for the operator $\sqcup$ explained below. For the other cases (as
well as for our derived operators), the projection is homomorphic.
We also identify $\mu \xx.\xx$ as $\End$ ($\mu \xx.\xx$
is generated when a target participant is not included 
under the recursion, for example, $\GS{\p}{\p'}{\U}.\mu
\xx.\GS{\q}{\q'}{\U}.\xx\proj{\p}=\Lout{\p}{\U}{\mu \xx.\xx}$) 
and $\mu \xx.T$ as $T$ if $\xx\not\in \ftv({T})$.

\paragraph{\bf Mergeability of branching types} 
We first recall the example from~\cite{CHY07}, which explains that
na\"ive branching projection leads to inconsistent end-point types. 

{\small
$$
\begin{array}{lll}
\TO{\W[0]}{\W[1]}: & \{\mathsf{ok}:\TO{\W[1]}{\W[2]}:\ENCan{\Bool}, 
                     & \ \mathsf{quit}:\TO{\W[1]}{\W[2]}:\ENCan{\Nat}\}
\end{array}
$$}

We cannot project the above type onto $\W[2]$ because, while the branches behave
differently, $\W[0]$ makes a choice without informing $\W[2]$ who thus cannot
know the type of the expected value. 
A solution is to define projection only when the branches are identical, i.e. we change the above $\Nat$ to
$\Bool$ in our example above. 

In our framework, this restriction is too strong since each branch may contain
different parametric interaction patterns.  To overcome this, below we propose
a method called {\em mergeability} of branching types.\footnote{The idea of mergeability is introduced informally 
in the tutorial paper  \cite{Tutorial09}.}
\begin{DEFINITION}[Mergeability] \rm
\label{def:mergeability}
The mergeability relation $\mergeop$ is the smallest congruence relation over 
end-point types such that:
$$
\begin{prooftree}
{\forall i\in (K \cap J). T_i\mergeop T_i' \quad 
\forall k\in (K \setminus J), \forall j.(J \setminus K).l_k \not = l_j
}
\justifies 
{\langle\p,\{l_k:T_k\}_{k\in K}\mergeop 
\&\langle\p,\{l_j:T_j'\}_{j\in J}\rangle
}
\end{prooftree}
$$
When $T_1\mergeop T_2$ is defined, 
we define the operation $\mergecup$ as a partial commutative 
operator over two  types such that $T\mergecup T=T$ for all types and that:
\[
\begin{array}{lll}
\Lbranch{\p}{\T_k}\mergecup 
\&\langle\p,\{l_j:T_j'\}_{j\in J}\rangle \ = \\ 
\quad \&\langle\p,\{l_i:T_i\mergecup T_i'\}_{i\in K\cap J}
\cup \{l_k:T_k\}_{k\in K\setminus J}
\cup \{l_j:T_j'\}_{j\in J\setminus K}\rangle\\[1mm]
\end{array}
\]
and homomorphic for other types (i.e. 
$\mathcal{C}[T_1] \sqcup \mathcal{C}[T_2]=\mathcal{C}[T_1\sqcup
T_2]$ where $\mathcal{C}$ is a context for local types). 
\end{DEFINITION}
The mergeability relation states that two types are identical up to their
branching types where only branches with distinct labels are allowed to be
different.  By this extended typing condition, we can modify our previous global
type example to add $\mathsf{ok}$ and $\mathsf{quit}$ labels to notify
$\W[2]$. We get:

{\small\[
\begin{array}{llll}
\TO{\W[0]}{\W[1]}:& \{\mathsf{ok}:\TO{\W[1]}{\W[2]}:
\{ \mathsf{ok}:\TO{\W[1]}{\W[2]}\ENCan{\Bool} \ \}, \\
&\
\mathsf{quit}:\TO{\W[1]}{\W[2]}:\{\mathsf{quit}:\TO{\W[1]}{\W[2]}\ENCan{\Nat}\}\}\}
\end{array}
\]}

Then $\W[2]$ can have the type $\&\ENCan{\W[1], \ \{
  \mathsf{ok}:\ENCan{\W[1],\Bool}, \ \mathsf{quit}:\ENCan{\W[1],\Nat} \}}$ which
could not be obtained through the original projection rule
in~\cite{CHY07,BettiniCDLDY08LONG}.
This projection is sound 
up to branching subtyping 
(it will be proved in Lemma \ref{lem:mergeability} later).

\subsection{Type system (1): environments, judgements and kinding}
\noindent 
This subsection introduces the environments and kinding systems. 
Because free indices appear both in terms (e.g.
participants in session initialisation) and in types, the formal definition of
what constitutes a valid term and a valid type are interdependent and
both in turn require a careful definition of a valid global type.

\paragraph{\bf Environments}
One of the main differences with previous session type systems is that session
environments $\Delta$ 
can contain dependent {\em process types}.  The grammar of environments, process
types and kinds are given below.
\begin{center}
\small
\begin{tabular}{ll}
\begin{tabular}{rcl@{\quad}l}
\D & ::= & $\emptyset$ \sep \D, \ccc:\T\\[1mm]  
\end{tabular}
&
\begin{tabular}{rcl@{\quad}l}
$\Gamma$ & ::= & $\emptyset \sep\Gamma, \PRED \sep \Gamma, u:\ST \sep 
\Gamma, \ii:\II \sep \Gamma, \X:\Ty$
\quad 
$\Ty \ ::= \ \D \sep \Pi \AT{\ii}{\II}.\Ty$\\
\end{tabular}
\end{tabular}
\end{center}
$\D$ is the {\em session environment} 
which associates
channels to session types. 
$\Ga$ is the {\em standard environment}
which contains predicates and which associates variables to sort types, service names to global
types, indices to index sets and process variables 
to session types. 
$\Ty$ is a {\em process type} which 
is either a session environment or a 
dependent type. 
We write $\Ga,\uu:S$ only if $\uu\not\in\dom{\Gamma}$ 
where
$\dom{\Gamma}$ denotes the domain of $\Gamma$.  
We use the same convention for others.

\begin{figure*}
\center
\begin{tabular}{lllllllll}
$\Ga \proves \Env$ & \quad\quad & well-formed environments\\
$\Ga \proves \kappa$ & \quad \quad& well-formed kindings\\
$\Ga \proves \alpha\RHD \kappa$ & \quad\quad & well-formed types\\
$\Ga \proves \alpha\equiv \beta$ & \quad\quad & type
equivalence\\
\end{tabular}\quad 
\begin{tabular}{lllllllll}
$\Ga \proves \alpha\WB \beta$ & \quad\quad & type isomorphism\\
$\Ga \proves \e \rhd U$ & \quad\quad & expression\\
$\Ga \proves \p \rhd \ParT$ & \quad\quad & participant \\
$\Ga \proves \PP \rhd \tau$ & \quad\quad & processes\\
\end{tabular}
\caption{Judgements ($\alpha,\beta,...$ range over any types)}\label{fig:judgements}
\end{figure*}

\paragraph{\bf Judgements}
Our type system uses the judgements listed in
Figure~\ref{fig:judgements}.

Following \cite{DBLP:conf/popl/XiP99},  we assume given in the typing rules two
semantically defined judgements: $\Ga \models \PRED$ (predicate $\PRED$ is a
consequence of $\Gamma$) and $\Ga \models \tii : I$ ($\tii : I$ follows from the
assumptions of $\Ga$).

We write $\Gamma \proves \judg$ for arbitrary judgements and 
write $\Gamma \proves \judg,\judg'$ to stand for both 
$\Gamma \proves \judg$ and $\Gamma \proves \judg'$. 
In addition, we use two additional judgements for 
the runtime systems (one for queues $\derq{\Ga}{\stdqueue}{\D}$ and one for 
runtime processes $\derqq{\Ga}{\Sigma}{\PP}{\D}$) which are similar with those 
in \cite{BettiniCDLDY08LONG} and listed in the Appendix. 
We often omit $\Sigma$ from $\derqq{\Ga}{\Sigma}{\PP}{\D}$ if it is
not important.
%
%

\paragraph{\bf Kinding}
The definition of kinds is given below: 
\[
\begin{array}{l}
\kappa ::= \Pi \jj:I.\kappa \sep \Type 
\quad \quad \quad 
\ParT :: = \Nat \sep \Pi\AT{\ii}{\II}.{U_p}\\
\end{array}
\]
We inductively define well-formed types using a 
kind system. The judgement $\Gamma\vdash \alpha \!\RHD\!
\kappa$ means type $\alpha$ has kind $\kappa$.  Kinds include
proper types for global, value, principal, end-point and process types (denoted
by $\Type$), and the kind of type families, written by $\Pi\AT{\ii}{\II}.\K$.
%
The kinding rules are defined in 
Figure~\ref{fig:kindsystembase} 
and Figure~\ref{fig:kind_all} 
in this section and 
Figure~\ref{fig:localkindsystem} in the Appendix.
The environment well-formedness rules are in Figure
\ref{fig:wellformed_env}. 

The kinding rules for types, value types, principals, index sets and
process types are listed in  Figure~\ref{fig:kind_all}.  
In \trule{KMar} in the value types, $\ftv(G)$ denotes a set of 
free type variables in $G$. The condition $\ftv(G)=\emptyset$ 
means that shared channel types are always closed. 
Rule \trule{KIndex} forms the index 
sort which contains only natural number (by the condition $0\leq i$). 
Other rules in Figure~\ref{fig:kind_all} and the rules in 
Figure~\ref{fig:wellformed_env} are standard.   

We next explain the global type kinding rules from
Figure~\ref{fig:kindsystembase}. 
The local type kinding in Figure~\ref{fig:localkindsystem}
in Appendix is similar. 

Rule \trule{KIO} states that if both participants have $\Nat$-type,  
that the carried type $\U$ and the rest of the global type $\G'$ 
are kinded by $\Type$, and that $\U$ does not contain any free type variables, 
then the resulting type is well-formed. This prevents these types from being
dependent. The rule \trule{KBra} is similar, while rules \trule{KRec,KTVar} are
standard.  

Dependent types are introduced when kinding recursors
in \trule{KRcr}.  
In \trule{KRcr}, we need an updated index range for $\ii$ in the premise $\Gamma,
\ii:\II^- \vdash \G' \RHD \Type$ since the index substitution uses the
predecessor of $\ii$. We define $\minus\II$ using the abbreviation 
$[0..\tjj]  =  \CONSTRAINT{\ii}{\Nat}{\ii\leq \tjj}$:
\begin{center}
${}\qquad \qquad \qquad \minus{[0..0]}  =  \emptyset 
\quad \text{and}\quad \minus{[0..\tii]}  = [0..\tii-1] 
$
\end{center}
Note that the second argument ($\lambda {\AT{\ii}{\II^-}}.\lambda
{\xx}.{\G'}$) is closed (i.e. it does not contain 
free type variables).
We use \trule{KApp} for both index applications. 
Note that \trule{KApp} checks whether the argument $\tii$ satisfies 
the index set $\II$.
Other rules are similarly understood including those 
for process types (noting $\Delta$ is a well-formed environment 
if it only contains types $T$ of kind $\PType$).

\begin{figure}
\[
\begin{array}{c}
\begin{prooftree}
{- }
\justifies
{\emptyset \vdash \Env} \using\scripttrule{ENul}
\end{prooftree}
\quad 
\begin{prooftree}
{\Gamma \models \mathtt{P}}
\justifies
{\Gamma,\mathtt{P} \vdash \Env} \using\scripttrule{EPre}
\end{prooftree}
\quad 
\begin{prooftree}
{\Gamma \vdash S \RHD \GType \quad u\not\in \dom{\Gamma} }
\justifies
{\Gamma,\AT{u}{S} \vdash \Env} \using\scripttrule{ESort}
\end{prooftree}
\\ 
\\ 
\begin{prooftree}
{\Gamma \vdash I \quad i\not\in \dom{\Gamma}}
\justifies
{\Gamma,i:I \vdash \Env} \using\scripttrule{EIndex}
\end{prooftree}
\quad 
\begin{prooftree}
{\Gamma \vdash \tau\RHD \kappa \quad X\not\in \dom{\Gamma} }
\justifies
{\Gamma,X:\tau \vdash \Env} \using\scripttrule{VEnv}
\end{prooftree}
\end{array}
\]
\caption{Well-formed environments} \label{fig:wellformed_env}
\end{figure}


\begin{figure}
\begin{tabular}{l}
{\bf Type}\\ 
\\ 
$
\begin{prooftree}
{\Gamma \vdash \Env}
\justifies
{\Gamma \vdash \GType}
\using\scripttrule{KBase}
\end{prooftree}
\quad
\begin{prooftree}
{\Gamma,\AT{\ii}{\II}  \vdash \K}
\justifies
{
\Gamma  \vdash \Pi\AT{\ii}{\II}.\K
} \using\scripttrule{KSeq}
\end{prooftree}
$
\\
\\
{\bf Value Types}\\ 
\\ 
$
\begin{array}{l}
\begin{prooftree}
{\Gamma \vdash \G \RHD \GType \quad \ftv(\G)=\emptyset}
\justifies
{\Gamma \vdash \mar{\G} \RHD \SType}  \using\scripttrule{KMar}
\end{prooftree}
\ 
\begin{prooftree}
{\Gamma \vdash \Env}
\justifies
{\Gamma \vdash \Nat \RHD \SType}  \using\scripttrule{KNat}
\end{prooftree}
\ 
\begin{prooftree}
{\Gamma \vdash \Env}
\justifies
{\Gamma \vdash \Bool \RHD \SType}  \using\scripttrule{KBool}
\end{prooftree}
\end{array}
$
\\
\\
\begin{tabular}{ll}
{\bf Principals} 
\\ 
\\
$\begin{prooftree}
{\Gamma \vdash \Env}
\justifies
{\Gamma \vdash \Nat \RHD \PRType}  \using\scripttrule{KPNat}
\end{prooftree}
\quad
\begin{prooftree}
{\Gamma,\AT{i}{I} \vdash U_p\RHD \kappa}
\justifies
{\Gamma \vdash U_p \RHD \Pi\AT{i}{I}.\kappa}  \using\scripttrule{KProd}
\end{prooftree}
$
\\
\\
{\bf Index Sets}\\ 
\\
$
\begin{array}{l}
\begin{prooftree}
{\Gamma \vdash \Env}
\justifies
{\Gamma \vdash \Nat}  \using\scripttrule{KINat}
\end{prooftree}
\quad
\begin{prooftree}
{\Gamma,\AT{\ii}{\II} \models \PRED\andl 0\leq \ii
}
\justifies
{\Gamma \vdash \CONSTRAINT{\ii}{\II}{\PRED\andl 0\leq \ii}} \using\scripttrule{KIIndex}
\end{prooftree}
\end{array}
$
\end{tabular}
\\
\\
{\bf Process Types}\\ 
\\ 
$
\begin{array}{l}
\begin{prooftree}
{\Gamma \vdash \Env }
\justifies
{\Gamma \vdash \emptyset \RHD \PType} \using\scripttrule{KPNul}
\end{prooftree}
\quad
\begin{prooftree}
{\Gamma \vdash \D \RHD \PType \quad \Gamma \vdash \T  \RHD \LType}
\justifies
{\Gamma \vdash \D, \ccc:\T  \RHD \PType} \using\scripttrule{KPCh}
\end{prooftree}
\\[3ex]
\begin{prooftree}
{\Gamma, \ii:\II \vdash \Ty \RHD \K}
\justifies
{\Gamma \vdash \Pi\AT{\ii}{\II}.\Ty \RHD \Pi\AT{\ii}{\II}.\K} \using\scripttrule{KPProd}
\end{prooftree}
\end{array}
$
\end{tabular}
\caption{Kinding system for types, values, principals, index sets and
process types}
\label{fig:kind_all}
\end{figure}

\begin{figure}
\small
\centering
\begin{tabular}{c}
\begin{prooftree}
{
\begin{array}{c}
\Gamma \vdash \p \rhd \Nat\quad \Gamma \vdash \p' \rhd \Nat \quad 
\Gamma   \vdash G' \RHD \Type \quad 
\Gamma \vdash \U  \RHD \Type\\[1mm]
\end{array}
}
\justifies
{\Gamma \vdash \GS{\p}{\p'}{\U}.\G' \RHD \Type} \using\trule{KIO}
\end{prooftree}\\
\\
\begin{prooftree}
{\Gamma \vdash \p \rhd \Nat, \Gamma \vdash \p' \rhd \Nat \quad 
  \forall k \in K,\ \Gamma \vdash \G_k \RHD \Type}
\justifies
{\Gamma \vdash \GB{\p}{\p'} \RHD \Type} \using\trule{KBra}
\end{prooftree}\\
\\
\begin{prooftree}
{\Gamma \vdash \G \RHD \K\subst{0}{j} 
\quad \Gamma, \ii:\II^- \vdash \G' \RHD \K\subst{i+1}{j} }
\justifies
{\Gamma \vdash \GR{\G}{\AT{\ii}{\II^-}}{\xx}{\G'}  
\RHD \Pi\AT{\jj}{\II}. \K} \using\trule{KRcr}
\end{prooftree}\\
\\
\begin{prooftree}
{\Gamma \vdash \G \RHD \Type}
\justifies
{\Gamma \vdash \GM{\xx}{\G} \RHD \Type} \using\trule{KRec}
\end{prooftree}
\quad 
\begin{prooftree}
{\Gamma \vdash \kappa}
\justifies
{\Gamma \vdash \xx \RHD \kappa} \using\trule{KVar}
\end{prooftree}
\quad 
\begin{prooftree}
{\Gamma \vdash \Env}
\justifies
{\Gamma \vdash \End \RHD \Type} \using\trule{KEnd}
\end{prooftree}\\
\\
\begin{prooftree}
{\Gamma \vdash \G \RHD  \Pi\AT{\ii}{\II}.\K \quad \Gamma \models \tii : I}
\justifies
{\Gamma \vdash \G\APP\tii \RHD \K\sub{\tii}{\ii}} \using\trule{KApp}
\end{prooftree}
\end{tabular}\\[2mm]
\caption{Kinding rules for global types} \label{fig:kindsystembase}
\end{figure}

\subsection{Type system (2): type equivalence}
Since our types include dependent types and recursors, 
we need a notion of type equivalence. 
We extend the standard method of \cite[\S2]{DependentBook} with the
recursor. 
The rules are found in Figure~\ref{fig:type-wf} and applied 
following the order appeared in Figure~\ref{fig:type-wf}.   
For example, 
\trule{WfRec} has a higher priority than 
\trule{WfRecF}. 
We only define the rules for $G$.   
The same set of rules can be applied to $T$ and $\tau$. 

Rule \trule{WfBase} 
is the main rule defining $\G_1\equiv\G_2$ and relies on the
existence of a common weak head normal form  for the two types. 

Rules \trule{WfIO} and \trule{WfBra} say if 
subterms are equated and each type satisfies the kinding rule, 
then the resulting global types are equated. 

Rule \trule{WfPRec} says the two recursive types 
are equated only if the bodies are equated. 
Note that we do not check whether unfolded recursive types are 
equated or not. 

Rule \trule{WfRVar} and \trule{WfEnd} are the base cases. 

Two recursors are equated if either (1) 
each global type subexpression is equated by $\equiv$ 
(rule \trule{WfRec}), or if not, 
(2) they reduce to the same normal forms when applied to a finite
number of indices (rule \trule{WfRecF}). 
Note that rule \trule{WfRec} has a higher priority 
than rule \trule{WfRecF} (since it is more efficient without 
reducing recursors). If $\GR{\G_1}{\AT{\ii}{\II}}{\xx}{\G_1'}
\equivwf \GR{\G_2}{\AT{\ii}{\II}}{\xx}{\G_2'}$ 
is derived by applying \trule{WfRec} under finite $I$, 
then the same equation can be derived using \trule{WfRecF}. 
Thus, when the index range is finite, 
\trule{WfRec} subsumes \trule{WfRecF}. On the other hand, 
\trule{WfRec} can be used for infinite index sets. 

Similarly, \trule{WfBase} is staged with \trule{WfApp} to ensure that 
the premise of \trule{WfRecF} always matches with \trule{WfBase}, 
not with \trule{WfApp} 
(it avoids 
the infinite application of rules \trule{WfRecF} and \trule{WfApp}).  
A use of these rules are given in the examples later. 
Other rules are standard. 

\begin{figure}
\centering
\begin{tabular}{c}
\begin{prooftree}
{
\Gamma \vdash \WHNF{G_1}\equivwf  \WHNF{G_2}
}
\justifies
{\Gamma \vdash G_1 \gequivwf G_2} 
\using\trule{WfBase}
\end{prooftree}
\\
\\
\begin{prooftree}
{
\begin{array}{c}
\Gamma \vdash \U_1 \gequivwf \U_2 
\quad 
\Gamma \vdash \G_1 \gequivwf \G_2 \quad 
\Gamma \vdash \GS{\p}{\p'}{\U_i}.\G_i \RHD \GType\\
\end{array}
}
\justifies
{\Gamma \vdash \GS{\p}{\p'}{\U_1}.\G_1\equivwf 
\Gamma \vdash \GS{\p}{\p'}{\U_2}.\G_2
} \using\trule{WfIO}
\end{prooftree}\\
\\
\begin{prooftree}
{
\begin{array}{c}
\forall k\in K.\ \Gamma \vdash \G_{1k} \gequivwf \G_{2k}  \quad 
\Gamma \vdash 
\TO{\p}{\q}\colon \{l_k: \G_{jk}\}_{k\in K}\RHD \GType \ 
\ (j=1,2)
\end{array}
}
\justifies
{\Gamma \vdash \TO{\p}{\q}\colon \{l_k: \G_{1k}\}_{k\in K}\equivwf  
\TO{\p}{\q}\colon \{l_k: \G_{2k}\}_{k\in K}} 
\using\trule{WfBra}
\end{prooftree}
\\
\\
\begin{prooftree}
{\Gamma \vdash \G_1 \gequivwf \G_2}
\justifies
{\Gamma \vdash \GM{\xx}{\G_1} \equivwf \GM{\xx}{\G_2}}  
\using\trule{WfPRec}
\end{prooftree}
\quad
\begin{prooftree}
{\Gamma \vdash \Env}
\justifies
{\Gamma \vdash \xx \equivwf \xx} \using\trule{WfRVar}
\end{prooftree}
\quad
\begin{prooftree}
{\Gamma \vdash \Env}
\justifies
{\Gamma \vdash \End \equivwf \End} \using\trule{WfEnd}
\end{prooftree}\\
\\
\begin{prooftree}
{
\begin{array}{c}
\Gamma \vdash
\G_1
\gequivwf
\G_2
\quad 
\Gamma,\AT{\ii}{\II} \vdash
\G_1'
\gequivwf
\G_2'
\end{array}
}
\justifies
{
\Gamma \vdash
\GR{\G_1}{\AT{\ii}{\II}}{\xx}{\G_1'}
\equivwf
\GR{\G_2}{\AT{\ii}{\II}}{\xx}{\G_2'}
}
\using\trule{WfRec}
\end{prooftree}
\\
\\
\begin{prooftree}
{
\begin{array}{l}
\Gamma \vdash
\G_1
\equiv
\G_2\\
\Gamma \vdash
\GR{\G_1}{\AT{\ii}{\II}}{\xx}{\G_1'}\ \n
\gequivwf
\GR{\G_2}{\AT{\ii}{\II}}{\xx}{\G_2'}\ \n
\quad \Gamma\models \II=[0..\m]
\quad 1\leq \n \leq \m 
\end{array}
}
\justifies
{
\Gamma \vdash
\GR{\G_1}{\AT{\ii}{\II}}{\xx}{\G_1'}
\equivwf
\GR{\G_2}{\AT{\ii}{\II}}{\xx}{\G_2'}
}
\using\trule{WfRecF}
\end{prooftree}
\\
\\
\begin{prooftree}
{\begin{array}{c}
\Gamma\vdash \G_1\equivwf \G_2 
\quad 
\Gamma\models \tii_1:I = \tii_2:I \quad 
\Gamma\vdash \G_i\tii_i \RHD \kappa 
\quad 
(i=1,2)\\[1mm]
\end{array}
}
\justifies
{\Gamma \vdash \G_1 \tii_1 \equivwf \G_2 \tii_2}
\using\trule{WfApp}
\end{prooftree}
\end{tabular}\\[2mm]
\caption{Global (decidable) type equivalence rules} \label{fig:type-wf}
\end{figure}

\begin{figure}
\[
\begin{prooftree}
{
\begin{array}{c}
\Gamma \vdash
\G_1
\equiv
\G_2\\
\forall \n \in I. 
\Gamma \vdash
\GR{\G_1}{\AT{\ii}{\II}}{\xx}{\G_1'}\ \n
\equiv 
\GR{\G_2}{\AT{\ii}{\II}}{\xx}{\G_2'}\ \n
\end{array}
}
\justifies
{
\Gamma \vdash
\GR{\G_1}{\AT{\ii}{\II}}{\xx}{\G_1'}
\equiv 
\GR{\G_2}{\AT{\ii}{\II}}{\xx}{\G_2'}
}
\using\trule{WfRecExt}
\end{prooftree}
\]
\caption{Meta-logical global type equivalence rule} \label{fig:meta-equivalence}
\end{figure}

\paragraph{\bf Type equivalence with meta-logic reasoning}
The set of rules in Figure~\ref{fig:type-wf} are designed with algorithmic
checking in mind (see \S~\ref{subsec:termination}). In some examples, in order
to type processes with types that are not syntactically close, it is
interesting to extend the equivalence classes on types, at the price of the
decidability of type checking.

We propose in Figure~\ref{fig:meta-equivalence} an additional equivalence rule
that removes from rule \trule{WfRecF} the finiteness assumption on $\II$. It
allows to prove the equivalence of two recursor-based types if it is possible to
prove meta-logically that they are extensionally equivalent. This technique can
be used to type several of our examples (see \S~\ref{subsec:typingexample}).

\subsection{Typing processes}
\begin{figure}
\begin{center}
\begin{tabular}{@{\hspace{-.2em}}c}
\begin{prooftree}
{\Gamma \models 0\leq \tii \op \tii' }
\justifies
{\Gamma \vdash \tii \op \tii' \rhd \Nat} \using\tfrule{TIOp}
\end{prooftree}
\quad 
\begin{prooftree}
{\Gamma \vdash \Env}
\justifies
{\Gamma \vdash \n \rhd \Nat} \using\tfrule{TNat}
\end{prooftree}
\\
\\
\begin{prooftree}
{\Gamma, \ii : \II  \vdash \Env}
\justifies
{\Gamma, \ii : \II \vdash \ii \rhd \Nat} \using\tfrule{TVari}
\end{prooftree}
\quad
\begin{prooftree}
{\Gamma  \vdash \kappa}
\justifies
{\Gamma\vdash \Alice\rhd \kappa} \using\tfrule{TId}
\end{prooftree}
\quad
\begin{prooftree}
{\Gamma  \vdash \p\rhd\Pi\AT{\ii}{\II}.\kappa\quad \Gamma\models \AT{\tii}{\II}}
\justifies
{\Gamma\vdash \p[\tii]\rhd \kappa\sub{\tii}{\ii}} \using\tfrule{TP}
\end{prooftree}
\\ 
\\
\begin{prooftree}
{
\begin{array}{c}
{\Gamma, \ii : \minus{\II},\X:\Ty\sub{\ii}{\jj}  \vdash \QQ  \rhd
 \Ty\sub{\ii+1}{\jj}\quad \Gamma\vdash \PP  \rhd\Ty\sub{0}{\jj}}
\quad 
\Gamma,\AT{\jj}{\II}\vdash \Ty \RHD \K
\end{array}
}
\justifies
{\Gamma \vdash \RECSEQP{\PP}{\ii}{\X}{\QQ} \rhd \Pi \AT{\jj}{\II}.\Ty} 
\using\tfrule{TPRec}
\end{prooftree}\\
\\
\begin{prooftree}
{
{\Ga\! \proves\! \PP  \!\rhd\! \tau \quad \Ga \!\proves\! \tau \!\equiv\! \tau'}
}
\justifies
{  
\Ga \proves \PP  \rhd \tau'
}
\using\tfrule{TEq}
\end{prooftree}
\begin{prooftree}
{\Gamma \vdash \PP\rhd \Pi \AT{\ii}{\II}.\Ty 
\quad \Gamma \models {\tii:\II} 
}
\justifies
{\Gamma \vdash \PP\APP\tii \rhd \Ty\sub{\tii}{\ii}} \using\tfrule{TApp}
\end{prooftree}
\\
\\
\begin{prooftree}
{\Gamma, \X:\Ty \vdash P \rhd \Ty}
\justifies
{\Gamma \vdash \mu{\X}.P \rhd \Ty} \using\tfrule{TRec}
\end{prooftree}
\begin{prooftree}
{\Gamma, \X:\Ty \vdash \Env \quad \Gamma\vdash \Ty\WB \Ty'}
\justifies
{\Gamma, \X:\Ty \vdash \X \rhd \Ty'} \using\tfrule{TVar}
\end{prooftree}
\\
\\
\begin{prooftree}
{
\begin{array}{@{}c@{}}
\Gamma \vdash \uu:\mar{\G} \quad 
\Gamma \vdash \PP \rhd \D, \y:\G \proj{\p_0}\\[1mm]
\Gamma \proves \p_i\rhd \Nat \quad 
\Gamma \models \pid(\G)=\{\p_0..\p_\n\}\\[1mm]
\end{array}
}
\justifies
{\Gamma \vdash \sr\uu{\p_0,..,\p_\n} \y\PP \rhd \D} 
\using\tfrule{TInit}
\end{prooftree}
\begin{prooftree}
{
\begin{array}{@{}c@{}}
\Gamma \vdash \uu:\mar{\G} \quad 
\Gamma \vdash \PP \rhd \D, \y:\G \proj{\pp}\\[1mm]
\Gamma \proves \pp\rhd \Nat \quad 
\Gamma \models \pp\in\pid(\G)\\[1mm]
\end{array}
}
\justifies
{\Gamma \vdash \ssa\uu\pp\y\PP \rhd \D } \using\tfrule{TAcc}
\end{prooftree}\\
\\
\begin{prooftree}
{
\Gamma \vdash \Ia:\mar{\G} \ \, \Gamma \proves \pp\rhd \Nat \ \, 
\Gamma \models \pp\in\pid(\G)
} 
\justifies
{\Gamma \vdash \sj{\Ia}{\pp}{\s} \rhd \s[\pp]:\G \proj{\pp}} \using\tfrule{TReq}
\end{prooftree}
\\
\\
\begin{prooftree}
{\Gamma \vdash \e \rhd \ST \quad \Gamma \vdash \PP\rhd \D,\ccc:\T}
\justifies
{\Gamma \vdash \Pout{\ccc}{\p}{\e}{\PP} \rhd \D,\ccc:\Lout{\p}{\ST}{\T}} \using\tfrule{TOut}
\end{prooftree}
\quad
\begin{prooftree}
{\Gamma,\x:\ST \vdash \PP\rhd \D,\ccc:\T}
\justifies
{\Gamma \vdash \Pin{\ccc}{\p}{\x}{\PP} \rhd \D,\ccc:\Lin{\p}{\ST}{\T}} \using\tfrule{TIn}
\end{prooftree}
\\
\\
\begin{prooftree}
{\Gamma \vdash \PP\rhd \D,\ccc:\T}
\justifies
{\Gamma \vdash \Pout{\ccc}{\p}{\ccc'}{\PP}\rhd \D,\ccc:\Lout{\p}{\T'}{\T},\ccc':\T'} \using\tfrule{TDeleg}
\end{prooftree}
\begin{prooftree}
{\Gamma \vdash \PP\rhd \D,\ccc:\T,\y:\T'}
\justifies
{\Gamma \vdash \Pin{\ccc}{\p}{\y}{\PP}\rhd \D,\ccc:\Lin{\p}{\T'}{\T}}
\using\tfrule{TRecep}
\end{prooftree}
\\
\\
\begin{prooftree}
{\Gamma \vdash \PP\rhd \D,\ccc:\T_\jj\quad \jj\in K}
\justifies
{\Gamma \vdash \Psel{\ccc}{\p}{l_\jj}{\PP}\rhd \D,\ccc:\Lsel{\p}{\T_k}} 
\using\tfrule{TSel}
\end{prooftree}\\
\\
\begin{prooftree}
{\forall k\in K,\Gamma \vdash \PP_k\rhd \D,\ccc:\T_k}
\justifies
{\Gamma \vdash \Pbranch{\ccc}{\p}\rhd \D,\ccc:\Lbranch{\p}{\T_k} } \using\tfrule{TBra}
\end{prooftree}\\
\\
\begin{prooftree}
{\Gamma, \Ia:\ENCan{G} \vdash \PP \rhd \Delta}
\justifies
{\Gamma \vdash \nuc{\Ia}{\PP} \rhd \Delta} \using\tfrule{TNu}
\end{prooftree}
\quad
\begin{prooftree}
{\Gamma \vdash \D \quad \D\ \End\text{ only}}
\justifies
{\Gamma \vdash \inact \rhd \D } \using\tfrule{TNull}
\end{prooftree}
\quad
\begin{prooftree}
{\Gamma \vdash \PP \rhd \D \quad \Gamma \vdash \QQ \rhd \D'}
\justifies
{\Gamma \vdash \PP \pc \QQ \rhd \D,\D' } \using\tfrule{TPar}
\end{prooftree}
\end{tabular}
\end{center}
\caption{Initial expression and process typing} \label{fig:process-typing1}
\end{figure}

We explain here (Figure~\ref{fig:process-typing1}) the 
typing rules for the initial processes. 
Rules \tftrule{TNat} and \tftrule{TVari} are
standard. Judgement $\Gamma\proves \Env$ (defined in Figure
\ref{fig:wellformed_env}) in the premise 
means that $\Gamma$ is well-formed. 
For participants, we check their typing by \tftrule{TId} and \tftrule{TP}
in a similar way as~\cite{DBLP:conf/popl/XiP99}.  
Rule \tftrule{TPRec} deals with the
changed index range within the recursor body. More precisely, we first check
$\tau$'s kind.  Then we verify for the base case ($\jj = 0$) that $P$ has type
$\Ty\sub{0}{\jj}$. Last, we check the more complex inductive case: $Q$ should
have type $\Ty\sub{\ii+1}{\jj}$ under the environment $\Gamma,
\ii\!\!:\!\!\minus{\II},\X\!\!:\!\!\Ty\sub{\ii}{\jj}$ where $\Ty\sub{\ii}{\jj}$
of $X$ means that $X$ satisfies the predecessor's type (induction
hypothesis). Rule \tftrule{TApp} is the elimination rule for dependent types. 

Rule \tftrule{TEq} states that typing works up to type
equivalence where $\equiv$ is defined in the previous subsection.
Recursion \tftrule{TRec} rule is standard. 
In rule \tftrule{TVar}, $\Delta\WB \Delta'$ denotes 
the standard isomorphism rules for recursive types 
(i.e.~we identify $\LM{\xx}{\T}$ and $T\sub{\LM{\xx}{\T}}{\xx}$), 
see Appendix \ref{app:kind}. Note that 
we apply isomorphic rules 
only when recursive variables are introduced. This way, we can separate 
type isomorphism for recursive types and type equalities 
with recursors. 

Rule \tftrule{TInit} types a session initialisation on shared channel \uu, binding
channel \y\ and 
requiring participants $\{\pp_0,..,\pp_n\}$. 
The premise verifies that the type of \y\
is the first projection of the global type \G\ of \uu\ and that the participants
in \G\ (denoted by $\mathtt{pid}(\G)$) can be semantically derived as
$\{\p_0,..,\p_\n\}$. 

Rule \tftrule{TAcc} allows to type the \p-th participant to the session initiated on
\uu. The typing rule checks that the type of \y\ is the \p-th projection of the
global type \G\ of \uu\ and that \G\ is fully
instantiated.   
The kind rule ensures that $\G$ 
 is fully instantiated (i.e. $G'$'s kind is $\Type$).
Rule \tftrule{TReq} types the process that waits for an accept from a
participant: its 
type corresponds to the end-point projection of $G$. 
 
The next four rules are associate the input/output processes with the
input/output types, and delegation input/output processes to 
session input/output types. 
Then the next two rules are branching/selection rules. 

Rule \tftrule{TNull} checks that $\Delta$ is well-formed and only contains
$\End$-type for weakening (the condition $\Delta\ \End\text{ only}$ means $\forall c \in \dom{\Delta}.\Delta(c)=\End$). Rule \tftrule{TPar} puts in parallel two processes
only if their sessions environments have disjoint domains. 
Other rules are standard. 

\section{Properties of typing}
\label{sec:property}
\noindent 
We study the two main properties of the typing system: one is the
termination of type-checking and another is type-soundness. The proofs
require a careful analysis due to the subtle interplay between
dependent types, recursors, recursive types and branching types.
\subsection{Basic properties}
\label{sec:basic}
\noindent We prove here a series of consistency lemmas concerning permutations
and weakening.
They are invariably deduced by induction on the derivations in the
standard manner.  

We use the following additional notations: $\Ga
\subseteq \Ga'$ iff $\AT{u}{S}\in \Ga$ implies $\AT{u}{S}\in \Ga'$ and
similarly for other mappings. In other words, $\Ga \subseteq \Ga'$
means that $\Ga'$ is a permutation of an extension of $\Gamma$.

\begin{lem} \label{lem:basic}\hfill
\begin{enumerate}[\em(1)]
\item (Permutation and Weakening)\ 
\label{lem:per}
Suppose $\Ga \subseteq \Ga'$ and $\Ga'\proves \Env$. Then 
$\Gamma \proves \judg$ 
implies
$\Gamma' \proves \judg$.  

\item (Strengthening)\ 
\label{lem:st}
$\Gamma,\AT{u}{U},\Gamma' \proves \judg$ 
and $u\not\in \fv(\Gamma', \judg)\cup \fn(\Gamma', \judg)$. Then $\Gamma,\Gamma' \proves \judg$.
Similarly for other mapping. 

\item (Agreement)
\label{lem:agree}
\begin{enumerate}[\em(a)]
\item 
\label{lem:agreeEnv}
$\Gamma \proves J$ implies $\Gamma \proves \Env$. 
\item 
\label{lem:agreeK}
$\Gamma \proves \G \RHD \K$ implies 
$\Gamma \proves \K$. Similarly for other judgements. 

\item 
\label{lem:agreeEq}
$\Gamma \proves \G \equiv \G'$ implies 
$\Gamma \proves G \RHD \K$. Similarly for other judgements. 

\item 
\label{lem:agreeTy}
$\Gamma \proves P \rhd \Ty$ implies $\Gamma \proves \Ty \RHD \K$.
Similarly for other judgements. 
\end{enumerate}
\item (Exchange) 
\label{lem:agreeEx}
\begin{enumerate}[\em(a)] 
\item 
\label{lem:ExT}
$\Gamma,\AT{u}{U},\Gamma' \proves \judg$ and $\Gamma \proves U\equiv
U'$. Then $\Gamma,\AT{u}{U'},\Gamma' \proves \judg$. 
Similarly for other mappings. 
\item 
\label{lem:ExI}
$\Gamma,\AT{\ii}{\II},\Gamma' \proves \judg$ and $\Gamma\models
\AT{\ii}{\II}=
\AT{\ii}{\II'}$. Then 
$\Gamma,\AT{\ii}{\II'},\Gamma' \proves \judg$. 

\item 
\label{lem:ExP}
$\Gamma,\PRED,\Gamma' \proves \judg$ and $\Gamma\models
\PRED=\PRED'$. Then 
$\Gamma,\PRED',\Gamma' \proves \judg$. 
\end{enumerate}
\end{enumerate}
\end{lem}
\begin{proof}
By induction on the derivations. We note that the proofs are 
done simultaneously.  For the rules which use substitutions 
in the conclusion of the rule, such as \tftrule{TApp} in
Figure~\ref{fig:process-typing1}, 
we require to use the next substitution lemma simultaneously. 
We only show the most interesting case 
with a recursor. \\[1mm]
{\bf Proof of (3)(b). Case \trule{KRcr}:}  Suppose 
$\Gamma \vdash \GR{\G}{\AT{\ii}{\II^-}}{\xx}{\G'}  
\RHD \Pi\AT{\jj}{\II}. \K$ is derived by \trule{KRcr} in 
Figure~\ref{fig:kindsystembase}.  We prove $\Gamma \vdash  
\Pi\AT{\jj}{\II}. \K$. 
From $\Gamma, \ii:\II^- \vdash \G' \RHD \K\subst{i+1}{j}$ in the
premise of \trule{KRcr}, 
we have $\Gamma, \ii:\II^- \vdash \K\subst{i+1}{j}$ by inductive
hypothesis.  By definition of $\II^-$, this implies 
$\Gamma, \jj:\II \vdash \K$. Now by  \trule{KSeq}, 
we have $\Gamma\vdash \Pi\AT{\jj}{\II}.\K$, as desired. 
\end{proof}
The following lemma which states that well-typedness is
preserved by substitution of appropriate values for variables, is the
key result underlying Subject Reduction. This also guarantees that 
the substitution for the index which affects to a shared environment and
a type of a term, and the substitution for a process variable are
always well-defined.  
Note that 
substitutions may change session types and environments in the index case.
\newpage

\begin{lem}[Substitution lemma]\label{lem:substitution}\hfill
\begin{enumerate}[\em(1)]
\item \labelx{subs01} 
If ${\Ga,\ptilde{\ii}:\II,\Ga'}\proves_\Sigma {\judg}$ and 
$\Ga\models \n: \II$, then 
${\Ga,(\Ga'\subst{\n}{\ii})}\proves_\Sigma \judg\subst{\n}{\ii}$.
\item \labelx{subs02} 
If \derqq{\Ga,\ptilde{\X}:\ptilde{\D_0}}{\Sigma}{\PP}{\Ty} and 
\de\Ga{\ptilde{\QQ}}{\D_0}, then
\derqq{\Ga}{\Sigma}{\PP\subst{\QQ}{\ptilde{X}}}{\Ty}.
\item \label{subs1} 
If \derqq{\Ga,\ptilde{\x}:\ptilde{\ST}}{\Sigma}{\PP}{\D} and \de\Ga{\ptilde{v}}{\ptilde{\ST}}, then
\derqq{\Ga}{\Sigma}{\PP\subst{\ptilde{\va}}{\ptilde{\x}}}{\D}.
\item \labelx{subs2} If
\derqq{\Ga}{\Sigma}{\PP}{\D,\y:\T}, then
\derqq{\Ga}{\Sigma}{\PP\sub{\si\s \pv}{\y}}{\D,\si\s \pv:\T}.
\end{enumerate}
\end{lem}
\begin{proof}
By induction on the derivations.
We prove the most interesting case: 
if \derqq{\Ga,\ptilde{\ii}:\II,\Ga'}{\Sigma}{\PP}{\Ty} and 
\der{\Ga}{\n}{\Nat} with $\Ga\models \n: \II$, then 
$\derqq{\Ga,(\Ga'\subst{\n}{\ii})}{\Sigma}{\PP\subst{\n}{\ii}}{\Ty\subst{\n}{\ii}}$
when the last applied rule is \tftrule{TPRec}. 
Assume 
\[
\Gamma,k:J, \Gamma' \vdash \RECSEQP{\PP}{\ii}{\X}{\QQ} \rhd \Pi
\AT{\jj}{\II}.\Ty 
\quad \mbox{with} \quad \Ga\models \n: J
\]
is derived from \tftrule{TPRec}. This is derived by: 
{\small
\begin{eqnarray}
& & 
\Gamma,k:J, \Gamma', \ii : \minus{\II},\X:\Ty\sub{\ii}{\jj}  \vdash \QQ  \rhd
\Ty\sub{\ii+1}{\jj} \label{sub:q-1}\\
& & 
\Gamma, k:J, \Gamma'\vdash \PP  \rhd\Ty\sub{0}{\jj} 
\label{sub:p-1}\\
& & \Gamma, k:J, \Gamma',\AT{\jj}{\II} \vdash \Ty \RHD \K \label{sub:ty-1}\\
\text{From (\ref{sub:q-1})}
& & 
\Gamma,\Gamma'\sub{\n}{k}, \ii : \minus{\II}\sub{\n}{k},\X:\Ty\sub{\ii}{\jj}\sub{\n}{k}  \vdash \QQ\sub{\n}{k}  \rhd
\Ty\sub{\ii+1}{\jj}\sub{\n}{k} \label{sub:q-2}\\
\text{From (\ref{sub:p-1})}
& & \Gamma, \Gamma'\sub{\n}{k} \vdash 
\PP\sub{\n}{k}  \rhd\Ty\sub{0}{\jj}\sub{\n}{k} 
\label{sub:p-2}\\
\text{From (\ref{sub:ty-1})}
& & 
\Gamma,\Gamma'\sub{\n}{k},\AT{\jj}{\II}\sub{\n}{k} \vdash
\Ty\sub{\n}{k} \RHD 
\K\sub{\n}{k} 
\label{sub:ty-2}
\end{eqnarray}
}
From (\ref{sub:q-2}), (\ref{sub:p-2})
and (\ref{sub:ty-2}), by \tftrule{TPRec}, we obtain 
$\Gamma,\Gamma'\sub{\n}{k} \vdash
(\RECSEQP{\PP}{\ii}{\X}{\QQ})\sub{\n}{k} \rhd 
(\Pi \AT{\jj}{\II}.\Ty)\sub{\n}{k}$ as required.
\end{proof}


\label{sec:subjectreduction}
\subsection{Termination of equality checking and type checking} 
\label{subsec:termination}
This subsection proves the termination of the type-checking 
(we assume that we use the type equality rules in Figure~\ref{fig:type-wf}).  
\noindent Ensuring termination of type-checking with dependent types is not an
easy task since type equivalences are often derived from term equivalences.  We
rely on the strong normalisation of System $\mathcal{T}$
\cite{GirardJY:protyp} for the termination proof.
 
\begin{PROP}[Termination and confluence]
\label{prop:SNconfluent}
The reduction relation $\redsym$ on global and end-point types 
(i.e. $\G\redsym \G'$
and $\T\redsym \T'$ for closed types in Figure \ref{fig:globalreduction}) 
are strong normalising and confluent on well-formed kindings. 
More precisely, 
if $\Gamma\proves G \RHD \K$, then there exists a unique term 
$G'= \WHNF{G}$ such that $G \redsym^\ast G'\not\redsym$. 
%
\end{PROP}
\begin{proof}
By strong normalisation of System $\mathcal{T}$ \cite{GirardJY:protyp}. 
For the confluence, we first note that the reduction relation on 
global types defined in Figure~\ref{fig:globalreduction} 
and on expressions with the first-order operators in the types 
is {\em deterministic}, i.e. if $\G\redsym \G_i$  
by rules in Figure~\ref{fig:globalreduction}, then $\G_1=\G_2$. 
Hence it is locally confluent, i.e. if  $\G\redsym \G_i$ ($i=1,2$) then 
$\G_i\redsym^\ast \G'$. Then we achieve the result
by Newman's Lemma. The second clause is a direct consequence 
from the fact that $\redsym$ coincides with the head reduction.
\end{proof}


\begin{figure}
\centering
\begin{tabular}{l@{\ }l@{\ }l@{\quad}l}
{\bf Judgements} $\judgementsize{\cdot}$&  
\ 
$\judgementsize{\Gamma \proves G_1 \gequivwf G_2}=
\judgementsize{\Gamma \proves G_1 \equivwf G_2}+1$
\\[1mm]
& 
$\judgementsize{\Gamma \proves G_1 \equivwf G_2}=
\omega\cdot(\inductionreductionsize{G_1}+\inductionreductionsize{G_2})+
\inductiontermsize{G_1}+\inductiontermsize{G_2}+1
$
\\[1mm]
{\bf Types $|\cdot|$} \\[2mm]
Value \quad& $\termsize{\Nat}=1$, 
             $\termsize{\mar\G}= \termsize{G}+1$\\[1mm]
Global \quad & $\termsize{\GS{\p}{\p'}{\U}.\G}=
   2+\termsize{\U}+\termsize{\G}$\\[1mm]
& $\termsize{\GB{\p}{\p'}}=2+\Sigma_{k\in K} (1+ \termsize{G_k})$\\[1mm]
& $\termsize{\GM{\xx}{\G}}=\termsize{\G}+2$ \quad \quad
  $\termsize{\xx}=\termsize{\n}=\termsize{\End}=1$\\[1mm]
& 
$\termsize{\G\APP \tii}=\termsize{G}+2$ \quad ($\fv(\tii)=\emptyset$) 
\quad\quad 
$\termsize{\G\APP \tii}=\inductiontermsize{G}+2$ \quad 
($\fv(\tii)\not=\emptyset$) 
\\[1mm]
&
$\termsize{\GR{\G}{\AT{\ii}{\II}}{\xx}{\G'}}=
   4+\inductiontermsize{G}+\inductiontermsize{G'}$\\[1mm]
Local\quad &
  $\termsize{\Lout{\p}{\U}{\T}}=3+\termsize{\U}+\termsize{\T}$, 
  \quad 
  $\termsize{\Lin{\p}{\U}{\T}}=3+\termsize{\U}+\termsize{\T}$\\[1mm] 
  & $\termsize{\Lsel{\p}{\T_k}}= \termsize{\Lbranch{\p}{\T_k}}
     =2+\Sigma_{k\in K}(1+\termsize{T_k})$\\[1mm] 
   & $\termsize{\LM{\xx}{\T}} =\termsize{\T}+2$, 
\quad $\termsize{\xx}=\termsize{\n}=\termsize{\End}=1$\\[1mm]
& 
$\termsize{T\APP \tii}=\termsize{T}+2$ \quad
($\fv(\tii)=\emptyset$) 
\quad 
$\termsize{T\APP \tii}=\inductiontermsize{T}+2$ \quad
($\fv(\tii)\not=\emptyset$) 
\\[1mm] & 
$\termsize{\GR{\T}{\AT{\ii}{\II}}{\xx}{\T'}}=
   4+\inductiontermsize{T}+\inductiontermsize{T'}$\\[1mm]
Principals & 
$\termsize{\Pi\AT{\ii}{\II}.{U_p}}= 2 + \termsize{U_p}$\\[1mm]
Processes  & 
$\termsize{\emptyset}=0$, 
$\termsize{\Delta,\AT{c}{\T}}=\termsize{\Delta}+\termsize{\T}+1$,
$\termsize{\Pi\AT{\ii}{\II}.\tau}= 2 + \termsize{\tau}$ 
\\[1mm]
{\bf Types $\vert\vert\cdot\vert\vert$}\\[2mm]
Global\quad &
  $\inductiontermsize{\GR{\G}{\AT{\ii}{\II}}{\xx}{\G'}}=
   \Sigma_{\n \in I}\termsize{{\GR{\G}{\AT{\ii}{\II}}{\xx}{\G'}} \n}$
\ ($I$ finite)
\\[1mm]
Local\quad &
$\inductiontermsize{\GR{\T}{\AT{\ii}{\II}}{\xx}{\T'}}=
   \Sigma_{\n \in I}\termsize{{\GR{\G}{\AT{\ii}{\II}}{\xx}{\G'}} \n}$
\ ($I$ finite)
\\[1mm]
&  Others are 
$\inductiontermsize{G}=\termsize{G}$. \\[2mm]
{\bf Types $\reductionsize{\cdot}$}\\[2mm]
Global 
& 
$\reductionsize{\GS{\p}{\p'}{\U}.\G}=\inductionreductionsize{G}$\\
& 
$\reductionsize{\GB{\p}{\p'}}=\Sigma_{i\in I}\inductionreductionsize{G_i}$\\
&
$\reductionsize{\GM{\xx}{\G}}=\inductionreductionsize{G}$\\
&
$\reductionsize{\G\APP \tii}=m+\inductionreductionsize{G}$ \\
& 
  $\reductionsize{\GR{\G}{\AT{\ii}{\II}}{\xx}{\G'}}=
   \inductionreductionsize{G}+\inductionreductionsize{G'}$\\
& $\reductionsize{\xx}=\reductionsize{\End}=0$ 
\\[1mm]
Local 
& 
$\reductionsize{\T\APP \tii}=n+\inductionreductionsize{T}$\\
&  $\reductionsize{\GR{\T}{\AT{\ii}{\II}}{\xx}{\T'}}=
   \inductionreductionsize{T}+\inductionreductionsize{T'}$\\[1mm]
& Others are similar to $\reductionsize{G}$. \\[1mm]
{\bf Types $\inductionreductionsize{\cdot}$}\\[2mm]
Global & $\inductionreductionsize{\GR{\G}{\AT{\ii}{\II}}{\xx}{\G'}}=
  \Sigma_{\n \in I} \reductionsize{\GR{\G}{\AT{\ii}{\II}}{\xx}{\G'} \n}$\\[1mm]
Local & $\inductionreductionsize{\GR{\T}{\AT{\ii}{\II}}{\xx}{\T'}}=
  \Sigma_{\n \in I} \reductionsize{\GR{\T}{\AT{\ii}{\II}}{\xx}{\T'} \n}$\\[1mm]
&  Others are 
$\inductionreductionsize{G}=\reductionsize{G}$ and homomorphic. \\[2mm]
\end{tabular}\\[2mm]
$m$ and $n$ denotes the upper bound on the length of any
$\longrightarrow$ from $\G\APP \tii$ and $\T\APP \tii$, respectively. 

\caption{Size of types and judgements and the upper bound of 
reductions with unfolding} \label{fig:size}
\end{figure}

Proving the termination of type equality checking requires 
a detailed analysis 
since the premises of the mathematical induction 
rules compare the two types whose syntactic sizes are larger 
than those of their conclusions. 
The size of the judgements are defined 
in Figure \ref{fig:size}, 
using the following functions from \cite[\S~2.4.3]{DependentBook}, 
taking rule \trule{WfRecF} in Figure~\ref{fig:type-wf} into account. 
\begin{enumerate}[(1)]
\item $\termsize{G}$ is the size of the structure of $G$ 
where we associate $\omega^2$-valued weight to each judgement 
to represent a possible reduction to a weak head normal form. 

\item $\inductiontermsize{G}$ is the size of the structure of $G$ 
where unfolding of recursors with finite index sets is considered,
taking \trule{WfRecF} in Figure~\ref{fig:type-wf} into account. 

\item $\reductionsize{\G}$ 
denotes an upper bound on the length of any 
$\redsym$ reduction sequence (see Figure~\ref{fig:globalreduction})  
starting from $G$ and its subterms.

\item $\inductionreductionsize{\G}$ 
denotes an upper bound on the length of any 
$\redsym$ reduction sequence (see Figure~\ref{fig:globalreduction})  
starting from $G$ and its subterms.  
It unfolds recursors with finite index sets. 
\end{enumerate}
The definition of the size of judgement 
$\judgementsize{\Gamma \proves G_1 \equivwf G_2}$ 
follows \cite[\S~2.4.3]{DependentBook}. 
We use $\inductionreductionsize{\G}$ and $\inductiontermsize{G}$ 
because of \trule{WfRecF} in Figure~\ref{fig:type-wf}. 
Note that $\termsize{G}$ corresponds to $\reductionsize{\G}$, 
 while $\inductiontermsize{G}$ corresponds to
$\inductionreductionsize{\G}$.  In $\inductionreductionsize{\G}$, 
we incorporate the number of reductions of unfolded 
recursors. Because the reduction of expressions 
strongly normalises, we choose the size of $e$ to be $1$ and 
the length of reductions of (closed) $e$ is assumed to be $0$.

The termination of type equality checking then is proved by the following main
lemma.  

\begin{lem}[Size of equality judgements]\label{lem:size}
The weight of any premise of a rule is always 
less than the weight of the conclusion. 
\end{lem}
\begin{proof}
Our proof is by induction on the length of reduction 
sequences and the size of terms. \\[1mm]
{\bf Case \trule{WfBase}.} The case $\WHNF{G_1}=G_1$ and 
$\WHNF{G_2}=G_2$ are obvious by definition. Hence we assume  
$\WHNF{G_1} \not = G_1$. Thus there exists at least one step reduction 
such that $G_1\redsym G_1'$. Hence by definition, 
$\inductionreductionsize{\WHNF{G_i}} < \inductionreductionsize{G_i}$ ($i=1,2$).
Similarly we have $\inductiontermsize{\WHNF{G_i}} \leq 
\inductiontermsize{G_i}$ by definition.  
Note that for any $G$, 
we have $\inductiontermsize{G}< \omega$. 
Hence we have 
\[
\begin{array}{lll}
  & 
\judgementsize{\Gamma \vdash \WHNF{G_1}\equivwf  \WHNF{G_2}}\\
= 
& \omega \cdot
\inductionreductionsize{\WHNF{G_1}} + \inductiontermsize{\WHNF{G_1}} + \omega \cdot
\inductionreductionsize{\WHNF{G_2}} + \inductiontermsize{\WHNF{G_2}} + 1\\ 
<  
& \omega \cdot \inductionreductionsize{G_1} + \inductiontermsize{G_1}+
  \omega \cdot \inductionreductionsize{G_2} + \inductiontermsize{G_2} + 2 \\
=& \judgementsize{\Gamma \vdash G_1\equiv  G_2}
\end{array}
\]
{\bf Case \trule{WfIO}.} Similar with \trule{WfBra} below.\\[1mm]
{\bf Case \trule{WfBra}.} 
\[
\begin{array}{ll}
& \Sigma_{k\in K}\judgementsize{\Gamma \vdash G_{1k}\equiv  G_{2k}}\\
= & 
\Sigma_{k\in K}(\judgementsize{\Gamma \vdash G_{1k}\equivwf G_{2k}}+1)\\
= & 
\Sigma_{k\in K}(
\omega\cdot(\inductionreductionsize{G_{1k}}+\inductionreductionsize{G_{2k}})+
\inductiontermsize{G_{1k}}+\inductiontermsize{G_{2k}}+2)\\
= & 
\Sigma_{k\in K}(
\omega\cdot(\reductionsize{G_{1k}}+\reductionsize{G_{2k}})+
\termsize{G_{1k}}+\termsize{G_{2k}}+2)\\
< & 
\Sigma_{k\in K}(
\omega\cdot(\reductionsize{G_{1k}}+\reductionsize{G_{2k}})+
\termsize{G_{1k}}+\termsize{G_{2k}}+2)+5\\
= & 
\judgementsize{
\Gamma \vdash \TO{\p}{\q}\colon \{l_k: \G_{1k}\}_{k\in K}\equivwf  
\TO{\p}{\q}\colon \{l_k: \G_{2k}\}_{k\in K}}
\end{array}
\] 
We note that $G_{1k}$ and $G_{2k}$ cannot 
be recursors 
since $\Gamma \proves G_{ik}\RHD \Type$ by the kinding rule, 
hence $\reductionsize{G_{ik}}=\inductionreductionsize{G_{ik}}$ and 
$\termsize{G_{ik}}=\inductiontermsize{G_{ik}}$ in 
the above third equation.  
\\[1mm]
{\bf Case \trule{WfPRec}.} Similar with \trule{WfBra} above. \\[1mm]
{\bf Cases \trule{WfRVar},\trule{WfEnd}.} By definition. \\[1mm]
{\bf Case \trule{WfRec}.} The case $I$ is infinite. 
{\small
\[
\begin{array}{ll}
& 
\judgementsize{\Gamma \vdash \G_1 \equiv \G_2} + 
\judgementsize{\Gamma,\AT{\ii}{\II} \vdash \G_1' \equiv \G_2'} \\
= & 
\judgementsize{\Gamma \proves G_1 \equivwf G_2}+1 + 
\judgementsize{\Gamma,\AT{\ii}{\II} \vdash \G_1' \equiv \G_2'}+1 \\
= & 
\omega\cdot(\inductionreductionsize{G_1}+\inductionreductionsize{G_2})+
\inductiontermsize{G_1}+\inductiontermsize{G_2}+2
 + 
\omega\cdot(\inductionreductionsize{G_1'}+\inductionreductionsize{G_2'})+
\inductiontermsize{G_1'}+\inductiontermsize{G_2'}+2\\
< & 
\omega\cdot(\inductionreductionsize{G_1}+\inductionreductionsize{G_2})+
\inductiontermsize{G_1}+\inductiontermsize{G_2}+4
 + 
\omega\cdot(\inductionreductionsize{G_1'}+\inductionreductionsize{G_2'})+
\inductiontermsize{G_1'}+\inductiontermsize{G_2'}+4 + 1\\
= & 
\judgementsize{
\Gamma \vdash
\GR{\G_1}{\AT{\ii}{\II}}{\xx}{\G_1'}
\equivwf
\GR{\G_2}{\AT{\ii}{\II}}{\xx}{\G_2'}
}
\end{array}
\] 
}
The case $I$ is finite. 
{\small
\[
\begin{array}{ll}
& 
\judgementsize{\Gamma \vdash \G_1 \equiv \G_2} + 
\judgementsize{\Gamma,\AT{\ii}{\II} \vdash \G_1' \equiv \G_2'} \\
= & 
\omega\cdot(\inductionreductionsize{G_1}+\inductionreductionsize{G_2})+
\inductiontermsize{G_1}+\inductiontermsize{G_2}+2
 + 
\omega\cdot(\inductionreductionsize{G_1'}+\inductionreductionsize{G_2'})+
\inductiontermsize{G_1'}+\inductiontermsize{G_2'}+2\\
< & 
\Sigma_{n\in I}
(
\omega\cdot(m_{1n} + \inductionreductionsize{G_{1n}''} + 
m_{2n} + \inductionreductionsize{G_{2n}''}))\\
& +  \Sigma_{n\in I}
(
\omega \cdot 
(\inductionreductionsize{G_1}+\inductionreductionsize{G_2})+
\inductiontermsize{G_1}+\inductiontermsize{G_2}+4+2
)\\
&  +  
\Sigma_{n\in I}
(
\omega \cdot 
(\inductionreductionsize{G_1'}+\inductionreductionsize{G_2'})+
\inductiontermsize{G_1'}+\inductiontermsize{G_2'}+4+2
)+1\\
= & 
\judgementsize{
\Gamma \vdash
\GR{\G_1}{\AT{\ii}{\II}}{\xx}{\G_1'}
\equivwf
\GR{\G_2}{\AT{\ii}{\II}}{\xx}{\G_2'}
}
\end{array}
\] 
}
where we assume ${m_{i\n}}$ is the length of the reduction 
sequence from  $\GR{\G_i}{\AT{\ii}{\II}}{\xx}{\G_i'}\, \n$ 
(in Figure~\ref{fig:globalreduction}), and 
$\GR{\G_i}{\AT{\ii}{\II}}{\xx}{\G_i'}\, \n 
 \redsym^\ast G_{i\n}'' \not\redsym$ for all $\n\in I$. 
\\[1mm]
{\bf Case \trule{WfRecF}.} Assume $I=[0,\ldots,\m]$. 
\[
\begin{array}{ll}
& 
\judgementsize{\Gamma \vdash G_1 \equiv  G_2} 
+ 
\Sigma_{1\leq \n \leq m}
\judgementsize{
\Gamma \vdash
\GR{\G_1}{\AT{\ii}{\II}}{\xx}{\G_1'} \n 
\equiv
\GR{\G_2}{\AT{\ii}{\II}}{\xx}{\G_2'} \n 
}\\
= & 
\omega\cdot \inductionreductionsize{G_1} 
+ \inductiontermsize{G_1}
+ \omega\cdot \inductionreductionsize{G_2} 
+ \inductiontermsize{G_2} + 2\\
& +  
\Sigma_{1\leq \n \leq m}
(\judgementsize{
\Gamma \vdash
\GR{\G_1}{\AT{\ii}{\II}}{\xx}{\G_1'} \n 
\equivwf 
\GR{\G_2}{\AT{\ii}{\II}}{\xx}{\G_2'} \n 
} + 1)
\\
= & 
\omega\cdot((m_{11}-1) + \inductionreductionsize{G_{11}''} + 
(m_{21}-1) + \inductionreductionsize{G_{21}''}))
+ \inductiontermsize{G_1} 
+ \inductiontermsize{G_2} + 2\\
& + \Sigma_{1\leq \n \leq m}
(
\omega\cdot(m_{1n} + \inductionreductionsize{G_{1n}''} + 
m_{2n} + \inductionreductionsize{G_{2n}''}))\\
& 
+ \Sigma_{1\leq \n \leq m}
(
\omega \cdot 
(\inductionreductionsize{G_1}+\inductionreductionsize{G_2})+
\inductiontermsize{G_1}+\inductiontermsize{G_2}+4+2
)\\
& +  
\Sigma_{1\leq \n \leq m}
(
\omega \cdot 
(\inductionreductionsize{G_1'}+\inductionreductionsize{G_2'})+
\inductiontermsize{G_1'}+\inductiontermsize{G_2'}+4+2
) \\
< & 
\Sigma_{n\in I}
(
\omega\cdot(m_{1n} + \inductionreductionsize{G_{1n}''} + 
m_{2n} + \inductionreductionsize{G_{2n}''}))\\
& 
+  \Sigma_{n\in I}
(
\omega \cdot 
(\inductionreductionsize{G_1}+\inductionreductionsize{G_2})+
\inductiontermsize{G_1}+\inductiontermsize{G_2}+4+2
)\\
& +  
\Sigma_{n\in I}
(
\omega \cdot 
(\inductionreductionsize{G_1'}+\inductionreductionsize{G_2'})+
\inductiontermsize{G_1'}+\inductiontermsize{G_2'}+4+2
) +1 \\
= & 
\judgementsize{
\Gamma \vdash
\GR{\G_1}{\AT{\ii}{\II}}{\xx}{\G_1'}
\equivwf
\GR{\G_2}{\AT{\ii}{\II}}{\xx}{\G_2'}
}
\end{array}
\] 
where we assume ${m_{i\n}}$ is the length of the reduction 
sequence from  $\GR{\G_i}{\AT{\ii}{\II}}{\xx}{\G_i'}\, \n$ 
(in Figure~\ref{fig:globalreduction}), and 
$\GR{\G_i}{\AT{\ii}{\II}}{\xx}{\G_i'}\, \n 
 \redsym^\ast G_{i\n}'' \not\redsym$ for all $\n\in I$. 

\noindent{\bf Case \trule{WfApp}} First, 
since \trule{WfApp} has a lower priority than \trule{WfBase}, 
we have $G_i\tii_i\not\redsym$.  
Hence
$\inductionreductionsize{G_i \tii_i}=\inductionreductionsize{G_i}$. 
We also note that:
\begin{enumerate}[(1)]
\item 
If $G_i$ is not recursor, then 
$\inductiontermsize{G_i}=\termsize{G_i}$; and  

\item 
If $G_i$ is a recursor, then $\tii_i$ contains free variables (since 
\trule{WfApp} has a lower priority than \trule{WfBase}, if $\tii_i$ is
closed, it reduces to $\m$ for some $\m$), 
hence 
$\inductiontermsize{G_i\, \tii_i}=\inductiontermsize{G_i}+2$. 
\end{enumerate} 
There is no case such that $G_i$ is a recursor 
and $\tii_i$ is some natural number $\n$ since, if so, we 
can apply \trule{WfBase}. Thus, by (1,2), 
$\inductiontermsize{G_i\, \tii_i}=\inductiontermsize{G_i}+2$. 
Hence we have: 
\[
\begin{array}{lll}
  & \judgementsize{\Gamma \vdash \G_1\equivwf \G_2}\\
= & \omega \cdot 
(\inductionreductionsize{G_1}+\inductionreductionsize{G_2})+
\inductiontermsize{G_1}+\inductiontermsize{G_2}+1\\
< & 
\omega \cdot 
(\inductionreductionsize{G_1}+\inductionreductionsize{G_2})+
\inductiontermsize{G_1}+\inductiontermsize{G_2}+4+1\\
= & \judgementsize{\Gamma \vdash \G_1 \tii_1 \equivwf \G_2 \tii_2}
\end{array}
\] 
as required.
\end{proof}

\begin{PROP}[Termination for type equality checking]
\label{pro:equiv-termination}
Assuming that proving \\
the predicates 
$\Gamma\models \PRED$ 
appearing in type equality derivations is decidable,  
then type-equality checking of $\Gamma\proves G \equiv G'$
terminates. 
Similarly for other types. 
\end{PROP}
\begin{proof} 
By Lemma \ref{lem:size} and termination 
of kinding and well-formed environment checking.
\end{proof}


We first formally define annotated processes which are processes 
with explicit type annotations for bound names and variables (see \S~\ref{par:annotated}). 
\begin{center}
\begin{tabular}{rcl}
 \PP  & ::= & \sr\uu{\p_0,..,\p_\n} {\AT{\y}{T}}\PP   
      \sep  \ssa\uu\p{\AT{\y}{T}}\PP   
      \sep  \Pin{\ccc}{\p}{\AT{\x}{T}}{\PP} 
      \sep  \nuc{\AT{a}{\ENCan{G}}}\PP\\[1mm]
      & \sep & $\mu \AT{X}{\tau}.\PP$ 
      \sep  $\RECSEQP{\PP}{\AT{\ii}{\II}}{\X}{\QQ}$ 
      \sep  $\X^\tau$ \\
\end{tabular}
\end{center}

\begin{thm}[Termination of type checking]
\label{thm:termination}
Assuming that proving the 
predicates $\Gamma\models \PRED$ appearing in kinding, 
equality, projection and typing derivations  
is decidable, 
then type-checking of annotated process $P$, i.e. 
$\Gamma\proves P\rhd \emptyset$ terminates. 
\end{thm}
\begin{proof} 
First, it is straightforward to show that 
kinding checking, well-formed environment checking and 
projection are decidable 
as long as deciding the predicates
$\Gamma\models \PRED$ appearing in the rules 
is possible. 

Secondly, we note that by the standard argument from indexed dependent types 
\cite{DependentBook,DBLP:conf/popl/XiP99}, 
for the dependent $\lambda$-applications (\tftrule{TApp} in Figure~\ref{fig:process-typing1}), we do not require 
equality of terms (i.e.~we only need the equality 
of the indices by the semantic consequence judgements). 

Thirdly, by the result from \cite[Corollary 2, page 217]{GH05}, 
the type isomorphic checking $\tau \WB \tau'$ terminates so that 
the type isomorphic checking 
in $\tftrule{Tvar}$ in Figure~\ref{fig:process-typing1} 
(between $\tau$ in the environment and 
$\tau'$ of $X^{\tau'}$) always terminates. 

Forth, it is known that type checking for annotated 
terms with session types terminates with subtyping 
\cite[\S~5.2]{GH05} and multiparty session types \cite{CHY07}. 

Hence the rest of the proof consists in eliminating the type equality rule
$\tftrule{Teq}$ in order to make the rules syntax-directed.    
We include the type equality check into 
$\tftrule{TInit,TReq,TAcc}$ 
(between the global type and its projected session
type), the input rules $\tftrule{TIn,TRecep}$ (between the session type 
and the type annotating $x$), 
$\tftrule{TRec}$ 
(between the session type and the type annotating $X$),   
and 
$\tftrule{TRec}$ 
(between the session type and the type annotating $X$ in $\Gamma$).  
We show the four syntax-directed rules. 
The first rule is the initialisation. 
\[ 
\begin{prooftree}
{
\begin{array}{@{}c@{}}
\Gamma \vdash \uu:\mar{\G} \quad 
\Gamma \vdash \PP \rhd \D, \y:T
\quad \Gamma \vdash \G \proj{\p_0} \equiv T\\[1mm]
\Gamma \proves \p_i\rhd \Nat \quad 
\Gamma \models \pid(\G)=\{\p_0..\p_\n\}\\[1mm]
\end{array}
}
\justifies
{\Gamma \vdash \sr\uu{\p_0,..,\p_\n} \y\PP \rhd \D} 
\using\tfrule{TInit}
\end{prooftree}
\]
Then it is straightforward to check 
$\Gamma \vdash \uu:\mar{\G}$ terminates.  
Checking $\Gamma \vdash \PP \rhd \D, \y:T$ also terminates by inductive
hypothesis. 
Checking $\Gamma \vdash \G \proj{\p_0} \equiv T$ terminates 
since the projection terminates and 
checking $\alpha \equiv \beta$ (for any type $\alpha$ and $\beta$) 
terminates by Lemma \ref{lem:size}.
Checking $\Gamma \proves \p_i\rhd \Nat$ terminates since 
the kinding checking terminates. Finally 
checking $\Gamma \models \pid(\G)=\{\p_0..\p_\n\}$ terminates 
by assumption. 

The second rule is the session input. 
\[
\begin{prooftree}
{\Gamma \vdash \PP\rhd \D,\ccc:\T,\y:\T' \quad 
\Gamma\proves T_0\equiv \T'}
\justifies
{\Gamma \vdash \Pin{\ccc}{\p}{\AT{\y}{T_0}}{\PP}\rhd \D,\ccc:\Lin{\p}{\T'}{\T}}
\using\tfrule{TRecep}
\end{prooftree}
\]
Then checking $\Gamma \vdash \PP\rhd \D,\ccc:\T,\y:\T'$ terminates by 
inductive hypothesis and checking $\Gamma\proves T_0\equiv \T'$ 
terminates by Lemma \ref{lem:size}.

The third and forth rules are recursions. 
\[
\begin{prooftree}
{\Gamma, \X:\Ty \vdash P \rhd \Ty' \quad \Gamma\proves \Ty \equiv \Ty'}
\justifies
{\Gamma \vdash \mu{\AT{\X}{\Ty}}.P \rhd \Ty'} \using\tfrule{TRec}
\end{prooftree}
\begin{prooftree}
{\Gamma, \X:\Ty \vdash \Env \quad \Gamma\vdash \Ty \WB \Ty' \quad 
\Gamma\vdash \Ty' \equiv \Ty_0}
\justifies
{\Gamma, \X:\Ty \vdash \X^{\Ty'} \rhd \Ty_0} 
\using\tfrule{TVar}
\end{prooftree}
\]
In \tftrule{TRec}, we assume $P\not = X$ (such a term is meaningless). 
Then \tftrule{TRec} terminates by inductive hypothesis 
and Lemma \ref{lem:size}, while \tftrule{TVar} terminates 
by termination of isomorphic checking (as said above) 
and Lemma \ref{lem:size}. 

Other rules are similar, hence we conclude the proof. 
\end{proof}
To ensure the termination of $\Gamma\models \PRED$, possible solutions include
the restriction of predicates to linear equalities over natural numbers without
multiplication (or to other decidable arithmetic subsets) or the restriction of
indices to finite domains, cf.~\cite{DBLP:conf/popl/XiP99}.

\subsection{Type soundness and progress}
\noindent 

\noindent The following lemma states that mergeability is sound 
with respect to the branching subtyping 
$\subT$. 
By this, we can safely replace the third clause 
$\sqcup_{k\in K}  \G_k\proj{\qq}$ of the branching case from the projection definition by 
$\sqcap \{T\ | \ \forall k\in K.T \subT (G_k\proj{\qq}) \}$.  
This allows us to prove subject reduction by including subsumption 
as done in \cite{CHY07}.

\begin{lem}[Soundness of mergeability]
\label{lem:mergeability}
Suppose $G_1\proj{\p} \mergeop G_2\proj{\p}$ and 
$\Gamma \proves G_i$. Then there exists $G$ such that 
$G\proj{\p}=\sqcap \{T\ | \ T \subT G_i\proj{\p} \ (i=1,2)\}$
where $\sqcap$ denotes the maximum element with respect to $\subT$. 
\end{lem}
\begin{proof}
The only interesting case is when $G_1\proj{\p}$ and $G_2\proj{\p}$ take 
a form of the branching type. Suppose 
$G_1 = \TO{\p'}{\p}\colon\! \indexed{l}{{\G'}}{k}{K}$
and 
$G_2 = \TO{\p'}{\p}\colon\! \indexed{l}{{\G''}}{j}{J}$
with $G_1\proj{\p} \mergeop G_2\proj{\p}$. 
Let $\G'_{k}\proj{\p}=T_{k}$ and 
$\G''_{j}\proj{\p}=T_{j}'$. 
Then by the definition of $\mergeop$ in \S~\ref{subsec:endpoint}, 
we have $G_1\proj{\p}= \langle\p',\{l_k:T_{k}\}_{k\in K}\rangle$ 
and
$G_2\proj{\p}=\&\langle\p',\{l_j:T_{j}'\}_{j\in J}\rangle$ 
with 
$\forall i\in (K \cap J). T_i\mergeop T_i'$ and
$\forall k\in (K \setminus J),\forall j\in (J \setminus K).l_k \not = l_j$. 
By the assumption and inductive hypothesis on $T_k\mergeop T_j'$, 
we can set 
\[
T=\&\langle\p',\{l_k:T''_{i}\}_{i\in I}\rangle
\] 
such that $I=K \cup J$; and (1) if $i\in K \cap J$, then  
$T''_{i}=T_{i}\sqcap T'_{i}$; (2) if $i\in K, i\not\in J$, then  
$T''_{i}=T_{i}$; and (3)
if $i\in J, i\not\in K$, then  
$T''_{i}=T'_{i}$. Set 
$G_{0i}\proj{\p}=T''_i$. Then we can obtain 
\[ 
G=\TO{\p'}{\p}\colon\! \indexed{l}{{\G_{0}}}{i}{I}
\]
which satisfies 
$G\proj{\p}=\sqcap \{T\ | \ T \subT G_i\proj{\p} \ (i=1,2)\}$, 
as desired. 
\end{proof}

\noindent
As session environments record channel states, they evolve when communications proceed.
This can be formalised by introducing
a notion of session environments reduction. These rules are formalised below
modulo $\equiv$. 
{\small
\begin{iteMize}{$\bullet$}
\item
$\sered{\set{\si{\s}{\pv}:\Lout{\qv}{\U}{\T},\si{\s}{\qv}:\Lin{\pv}{\U}{\T'}}}{\!\set{\si\s\pv: \T, \si\s{\qv}:\T'}}$
\item
$\sered{\set{\si{\s}{\pv}:\Lout{\pv}{\U}{\Lin{\pv}{\U}{\T'}}}}
{\!\set{\si\s\pv: \T'}}$
\item
$\sered{\set{\si\s\pv:\Lsel{\qv}{\T_k}}}{\set{\si\s\pv:\LselSingle{\qv}{l_j};T_j}}$
\item
$\sered{\set{\si\s\pv:\LselSingle{\qv}{l_j};T,
\si\s{\qv}:\branchtype}}{\set{\si\s\pv:\T, \si\s{\qv}:\T_j}}$
\item
$\sered{\D\cup\D''}{\D'\cup\D''}$ if $\sered{\D}{\D'}$. 
\end{iteMize}
}
\noindent
The first rule corresponds to the reception of
a value or channel by the participant $\qv$; 
the second rule formalises 
the reception of a value or channel sent by itself $\pv$;
the third rule 
treats the case of the choice of label $l_j$ while the forth rule
propagate these choices to the receiver (participant $\qv$). 

For the subject reduction theorem, we need to define 
{\em the coherence of the session 
environment \D}, which means that each end-point type is dual with
other end-point types.  

\begin{DEFINITION}\label{cd}
A session environment \D\ is {\em coherent for the session \s}
(notation \coe{\D}{\s}) if $ \sii:\T\in\Delta$ and
$\pro{\T}\q\not=\End$ imply $ \siq:\T'\in\Delta$ and
$\dual{\pro{\T}\q}{\pro{\T'}\p}$.  A session environment \D\ is
{\em coherent} if it is coherent for all sessions 
which occur in it.
\end{DEFINITION}
The definitions for $\pro{\T}\q$ and $\dual{}{}$ are defined in
Appendix~\ref{app:runtime}. Intuitively, $\pro{\T}\q$ is a projection
of $T$ onto $\q$ which is similarly defined as $\pro{\G}\q$; and 
$\dual{\pro{\T}\q}{\pro{\T'}\p}$ means actions in $\T$ onto 
$\q$ and actions in $\T'$ onto $\p$ are dual (i.e. input matches 
output, and branching matches with selections, and vice versa). 
Note that two projections of a same global type are always dual: 
let $\G$ a global type and $\p,\q \in \G$
with $\p \neq \q$. Then
$\dual{\pro{(\pro{\G}{\p})}{\q}}{\pro{(\pro{\G}{\q})}{\p}}$, 
i.e. session environments corresponding to global
type are always coherent.

Using the above notion we can state type preservation under reductions
as follows:
\begin{thm}[Subject Congruence and
  Reduction]\label{thm:sr}\mbox{\;}

\begin{enumerate}[\em(1)]
\item
If $\derqq\Ga{\Sigma}\PP\D$ and $\PP\equiv{\PP'}$, then
\derqq\Ga{\Sigma}{\PP'}{\D}. 
\item 
If \derqq\Ga{\Sigma}\PP\Ty\ and \redM\PP{\PP'} with 
\Ty\ coherent, then
\derqq\Ga{\Sigma}{\PP'}{\Ty'}\ for some $\Ty'$ such that
$\seredstar\Ty\Ty'$ with $\Ty'$ coherent.
\end{enumerate}
\end{thm}
\begin{proof}
  We only list the crucial cases of the proof of subject reduction: the recursor
  (where mathematical induction is required), the initialisation, the input and
  the output.  The proof of subject congruence is essentially as the same as
  that in \cite{CHY07,BettiniCDLDY08LONG}.  Our proof for (2) works by induction
  on the length of the derivation \redM\PP{\PP'}. The base case is trivial. We
  then proceed by a case analysis on the first reduction step
  \red\PP{\PP''}\redM{}{\PP'}. We omit the hat from principal values and
  $\Sigma$ for readability.
\\[1mm]
{\bf Case [ZeroR].}  Trivial. \\[1mm]
{\bf Case [SuccR].}  Suppose 
that we have
$\derqq\Ga{}{\GR{P}{\ii}{X}{Q}\APP \n+1}{\tau}$
and 
$\red{\GR{P}{\ii}{X}{Q}\APP \n+1}
{Q\sub{\n}{\ii}\sub{\GR{P}{\ii}{X}{Q}\APP \n}{X}}$. 
Then there exists $\Ty'$ such that 
\begin{eqnarray}
& 
\label{eq:rec1}
\Gamma, \ii : \minus{\II},\X:\Ty'\sub{\ii}{\jj}  \vdash \QQ  \rhd
  \Ty'\sub{\ii+1}{\jj} \\
& 
\label{eq:rec2}
\derqq\Ga{}{P}{\Ty'\subst{0}{\ii}}\\ 
& 
\label{eq:rec3}
\Gamma\vdash \Pi \AT{\jj}{\II}.\Ty \RHD \Pi\AT{\jj}{\II}.\K
\end{eqnarray}
with 
$\Ty \equiv (\Pi\AT{\jj}{\II}.\Ty') \n+1 \equiv \Ty'\subst{\n+1}{\jj}$
and $\Ga{\models}{\n+1:\II}$. 
By Substitution Lemma (Lemma \ref{lem:substitution} (\ref{subs01})), 
noting $\Ga{\models}{\n: \minus{\II}}$,  
we have: 
$\Gamma, \X:\Ty'\sub{\ii}{\jj}\sub{\n}{\ii}  \vdash \QQ\sub{\n}{\ii}  \rhd
  \Ty'\sub{\ii+1}{\jj}\sub{\n}{\ii}$, 
which means that 
\begin{eqnarray}
\label{eq:subs}
\Gamma, \X:\Ty'\sub{\n}{\jj}  \vdash \QQ\sub{\n}{\ii}  \rhd
  \Ty'\sub{\n+1}{\jj} 
\end{eqnarray}
Then we use an induction on $\n$. \\[1mm]
{\em Base Case $\n=0$:} 
By applying Substitution Lemma (Lemma \ref{lem:substitution}
(\ref{subs02})) to  (\ref{eq:subs}) with (\ref{eq:rec2}), we have 
$\Gamma  \vdash \QQ\sub{1}{\ii}\sub{P}{X}  \rhd
  \Ty'\sub{1}{\jj}$.  \\[1mm]
{\em Inductive Case $\n\geq 1$:} 
By the inductive hypothesis on $\n$, 
we assume: 
$\derqq\Ga{}{\GR{P}{\ii}{X}{Q}\APP \n}{\Ty'\sub{\n}{\jj}}$. 
Then by applying Substitution Lemma (Lemma \ref{lem:substitution}) to
(\ref{eq:subs}) with this hypothesis, 
we obtain
$\Gamma\vdash \QQ\sub{\n}{\ii}\sub{\GR{P}{\ii}{X}{Q}\ \n}{X}  \rhd
 \Ty'\sub{\n+1}{\jj}$. \\[1mm]
\\[1mm]
{\bf Case [Init].}
$$
\sr\Ia{\p_0,..,\p_\n}{\y}{\PP}\redsym (\nu \s)(
        \PP\sub{\si\s {\p_0}}{\y} \pc s : \qbot \pc \sj{\Ia}{\p_1}{\s} \pc ...\pc
        \sj{\Ia}{\p_\n}{\s})
$$
We assume that
$\derqq{\Ga}{\emptyset}{\sr\Ia{\p_0,..,\p_\n}{\y}{\PP}}{\D}$. Inversion of
\tftrule{TInit} and \tftrule{TSub} gives that $\D'\leq\D$ and:
{\small
\begin{eqnarray}
& \forall i\neq 0, & \Gamma \proves \p_i\rhd \Nat\label{case-pi}\\
& & \Gamma \vdash \Ia:\mar{\G} \label{case-Ia}\\
& & \Gamma \models \pid(\G)=\{\p_0..\p_\n\} \label{case-pid}\\
& & \Gamma \vdash \PP \rhd \D', \y:\G \proj{\p_0} \label{case-PP}\\
\text{From~(\ref{case-PP}) and Lemma~\ref{lem:substitution}~(\ref{subs2}),}
& & \derqq{\Ga}{}{\PP\sub{\si\s {\p_0}}{\y}}{\D', \si\s {\p_0}:\G \proj{\p_0}} \label{case-sub}\\
\text{\hspace{-3ex}From Lemma~\ref{lem:basic}~(\ref{lem:agreeEnv}) and \trule{QInit},}
& & \derqq{\Ga}{s}{s : \qbot}{\emptyset} \label{case-sbot}\\
\text{From (\ref{case-pi}), (\ref{case-Ia}), (\ref{case-pid}) and
  \tftrule{TReq},}
& \forall i\neq 0, & \Gamma \vdash \sj{\Ia}{\pp_i}{\s} \rhd \s[\pp_i]:\G \proj{\pp_i} \label{case-sppi}
\end{eqnarray}
}



Then \tftrule{TPar} on (\ref{case-sub}), (\ref{case-sbot}) and (\ref{case-sppi}) gives:
$$\Gamma \vdash \PP\sub{\si\s {\p_0}}{\y} \pc \sj{\Ia}{\p_1}{\s} \pc ...\pc
\sj{\Ia}{\p_\n}{\s} %
\rhd \D', \s[\pp_0]:\G \proj{\pp_0}, ... , \s[\pp_n]:\G
\proj{\pp_n}
$$
From \trule{GInit} and \trule{GPar}, we have:
$$\Gamma \vdash _ s \PP\sub{\si\s {\p_0}}{\y} \pc \sj{\Ia}{\p_1}{\s} \pc ...\pc
\sj{\Ia}{\p_\n}{\s} \pc s : \qbot%
\rhd \D', \s[\pp_0]:\G \proj{\pp_0}, ... , \s[\pp_n]:\G
\proj{\pp_n}
$$
From Lemma \ref{lem:mergeability} we know that $\coe{(\s[\pp_0]:\G \proj{\pp_0},
  ... , \s[\pp_n]:\G\proj{\pp_n})}{\s}$. We can then use \trule{GSRes} to get:
$$\Gamma \vdash _ \emptyset (\nu \s)(\PP\sub{\si\s {\p_0}}{\y} \pc \sj{\Ia}{\p_1}{\s} \pc ...\pc
\sj{\Ia}{\p_\n}{\s} \pc s : \qbot) \rhd \D'
$$
We conclude from 
\tftrule{TSub}.
\\[1mm]
{\bf Case [Join].} 
$$
\sj{\Ia}{\pp}{\s} \pc \ssa\Ia{\pp}{\y}{\PP}\redsym 
        \PP\sub{\si \s {\p}}{\y}
$$
We assume that $\Gamma \vdash _{} \sj{\Ia}{\p}{\s} \pc \ssa\Ia{\p}{\y}{\PP}\rhd
\D$. Inversion of \tftrule{TPar} and \tftrule{TSub} gives that $\D=\D',\s[\pp]:\T$ and :
\begin{eqnarray}
& & \Gamma \vdash \sj{\Ia}{\pp}{\s} \rhd \s[\pp]:\G \proj{\pp}\label{case-Join-spp}\\
& & \T\geq\G\proj{\pp}\label{case-Join-subt}\\
& & \Gamma \vdash \ssa\uu\pp\y\PP \rhd \D'\label{case-Join-PP}\\
\text{By inversion of \tftrule{TAcc} from~(\ref{case-Join-PP})}
& & \Gamma \vdash \PP \rhd \D', \y:\G \proj{\pp}\label{case-Join-projPP}\\
\text{From~(\ref{case-Join-projPP}) and Lemma~\ref{lem:substitution}~(\ref{subs2}),}
& & \derqq{\Ga}{}{\PP\sub{\si\s {\p}}{\y}}{\D', \si\s {\p}:\G \proj{\p}} \label{case-Join-sub}
\end{eqnarray}

We conclude by \tftrule{TSub} from (\ref{case-Join-sub}) and
(\ref{case-Join-subt}).
\smallskip

\noindent{\bf Case [Send].} 
\[\red{\out{\si{\s}{\q}}{\va}{\p}{\PP} \Par \mqueue{\s}{\queue}}
    {\PP \Par \mqueue{\s}{\qtail{\valheap{\va}{\p}{\q}}}}\]
By inductive hypothesis, 
$\derqq{\Ga}{\Sigma}{\out{\si{\s}{\q}}{\e}{\p}{\PP} \Par
\mqueue{\s}{\queue}}{\D}$ with  
$\Sigma=\set{\s}$. Since this is derived by [GPar], we have:
\begin{eqnarray}
\der{\Ga}{\out{\si{\s}{\q}}{\va}{\p}{\PP}}{\D_1}\label{L21}\\
\derqq{\Ga}{\{s\}}{\mqueue{\s}{\queue}}{\D_2}\label{L22}
\end{eqnarray}
where $\D=\D_2\Dcomp\D_1$.
  From (\ref{L21}), we have
\begin{eqnarray}
\nonumber \D_1=\D_1',\si{\s}{\q}:\oT{\ptilde{\ST}}{\p};\T\label{L23}\\
\de\Ga{\va}{\ST}\label{L24}\\
\der\Ga\PP{\D_1',\si{\s}{\q}:\T}.\label{L25}
\end{eqnarray}
Using \trule{QSend} on (\ref{L22}) and (\ref{L24}) we derive
\begin{eqnarray}
\derqq{\Ga}{\{s\}}{ \mqueue{\s}{\qtail{
\valheap{\va}{\p}{\q}}}}{\D_2\Tcomp\set{ \si{\s}{\q} :
\oT{\SST}{\p}}}.\label{L26}
\end{eqnarray}
Using \trule{GInit} on (\ref{L25}) we derive
{\begin{eqnarray}
\derqq\Ga{\emptyset}\PP{\D_1',\si{\s}{\q}:\T}\label{L251}
\end{eqnarray}}
and then using \trule{GPar} on (\ref{L251}) and (\ref{L26}), we conclude
\begin{eqnarray*}
\derqq{\Ga}{\{\s\}}{\PP\pc\mqueue{\s}{\qtail{
\valheap{\va}{\q}{\p}}}}{(\D_2\Tcomp\set{\si{\s}{\q} :
\oT{\SST}{\p}})\Dcomp (\D_1',\si{\s}{\q}:\T)}.\label{L27}
\end{eqnarray*}

Note that $(\D_2\Tcomp\set{ \si{\s}{\q} : \oT{\SST}{\p}})\Dcomp
(\D_1',\si{\s}{\q}:\T){=}\D_2\Dcomp( \D_{1}',\si{\s}{\q} : \oT{\SST}{\p};\T).$

\medskip

\noindent{\bf Case [Recv].} 
\[ \inp{\sii}{\x}{\q}{\PP} \Par
\qpop{\valheaps{\va}{\set\p}{\q}} \redsym
\PP\subst{\ptilde{\va}}{\ptilde{\x}} \Par \s:\h\]
By inductive hypothesis, $\derqq{\Ga}{\Sigma}{\inp{\sii}{\x}{\q}{\PP} \Par
\qpop{\valheaps{\va}{\set\p}{\q}} }{\D}$ with $\Sigma=\set\s$.   
Since this is derived by [GPar], we have: 
\begin{eqnarray}
\der{\Ga}{\inp{\sii}{\x}{\q}{\PP}}{\D_1}\label{L31}\\
\derqq{\Ga}{\{\s\}}{\qpop{\valheaps{\va}{\p}{\q}}}{\D_2}\label{L32}
\end{eqnarray}
 
where $\D=\D_2\Dcomp\D_1$. From (\ref{L31}) we have
\begin{eqnarray}
\nonumber \D_1=\D_1',\si\s\p:\iT{\ptilde{\ST}}\q;\T\label{L33}\\
\der{\Ga,\ptilde{\x}:\ptilde{\ST}}\PP{\D_1',\si{\s}{\p}:\T}\label{L34}
\end{eqnarray}
From (\ref{L32}) we have
\begin{eqnarray}
\nonumber \D_2=\set{\siq :\oT{\SST'}{\p}}\Dcomp
 \D_2'\label{L36}\\
\derq{\Ga}{\stdqueue}{\D_2'}\label{L37}\\
 \de{\Ga}{\values}{\SST'}.\label{L38}
\end{eqnarray}
The coherence of $\D$ implies $\SST=\SST'$.
From (\ref{L34}) and (\ref{L38}), together with Substitution lemma, we obtain
$\der{\Ga}{\PP\subst{\ptilde{\va}}{\ptilde{\x}}}{\D_1',\si{\s}{\p}:\T}$, which implies by rule \trule{GInit}
\begin{eqnarray}
\derqq{\Ga}{\emptyset}{\PP\subst{\ptilde{\va}}{\ptilde{\x}}}{\D_1',\si{\s}{\p}:\T}.\label{L35}
\end{eqnarray}
Using rule \trule{GPar} on (\ref{L35}) and (\ref{L37}) we conclude
\begin{eqnarray*}
\derqq{\Ga}{\{\s\}}{\PP\subst{\ptilde{\va}}{\ptilde{\x}}\pc
\s:\h}{
 \D_2'\Dcomp(\D_1',\si{s}{\p}:\T}).\end
{eqnarray*}
Note that
${(\set{\siq :\oT{\SST}{\p}}\Dcomp
 \D_2')\Dcomp(\D_1',\s [\p]:\iT{\ptilde{\ST}}\q;\T)}~\Rightarrow{
 \D_2'\Dcomp(\D_1',\si{\s}{\p}:\T)}.$ 
\end{proof}
\noindent 
Note that communication safety~\cite[Theorem 5.5]{CHY07}
and session fidelity~\cite[Corollary 5.6]{CHY07} 
are corollaries of the above theorem.  

A notable fact is,  
in the presence of the asynchronous initiation primitive, we can still obtain 
{\em progress} in a single multiparty session as in \cite[Theorem
5.12]{CHY07}, i.e.~if a program $P$ starts from one session, 
the reductions at session channels do not get a stuck.
Formally 
\begin{enumerate}[(1)]
\item 
We say $P$ is {\em simple} if 
$P$ is typable and derived by  
$\Gamma \proves^\star P \rhd \Delta$ 
where the session typing in the premise and
the conclusion of each prefix rule
is restricted to at most a singleton. 
More concretely, (1) we eliminate $\Delta$ from 
\tfrule{TInit}, \tfrule{TAcc}, \tfrule{TOut}, \tfrule{TIn},
\tfrule{TSel} and \tfrule{TBra}, (2) we delete 
\tfrule{TDeleg} and \tfrule{TRecep}, 
(3) we restrict $\tau$ and $\Delta$ in 
\tfrule{TPRec}, \tfrule{TEq}, \tfrule{TApp}, \tfrule{TRec} and \tfrule{TVar}
contain at most only one session typing, and 
(4) we set $\Delta=\emptyset$  and $\Delta'$ contains 
at most only one session typing; or vice-versa in \tfrule{TPar}. 
\item 
We say $P$ is {\em well-linked} when for
each $P \redsym^\ast Q$, whenever $Q$ has an active prefix whose subject is
a (free or bound) shared channels, then it is always reducible. 
This condition eliminates the element 
which can hinder progress is when interactions
at shared channels cannot proceed. See \cite[\S~5]{CHY07} more detailed
definitions. 
\end{enumerate}
The proof of the following theorem 
is essentially identical with \cite[Theorem 5.12]{CHY07}.

\begin{thm}[Progress]
\label{thm:progress}
If $P$ is well-linked and without any
element from the runtime syntax and 
	$\Gamma \vdash^\star P \rhd \emptyset$. 
Then for all $P\redsym^\ast Q$, either $Q\equiv \inact$ or 
$Q\redsym R$ for some $R$. 
\end{thm}

\section{Typing examples}
\label{subsec:typingexample}
 \noindent
In this section, we give examples of typing derivations for the protocols
mentioned in
\S~\ref{sec:introduction} and \S~\ref{subsec:examples}.

\subsection{Repetition example - \S~\ref{sec:introduction}
  (\ref{alice-bob-carol})} 
This example illustrates the repetition of a message
pattern. The global type for this protocol is 
$G(n)=\FOREACH{\ii}{<n}{\GS{\Alice}{\Bob}{\Nat}.\GS{\Bob}{\Carol}{\Nat}}$. 
Following the projection from Figure~\ref{fig:projection}, \Alice 's end-point
projection of $G(n)$ has the following form:
{\small
\[
\begin{array}{rcll}
  G(n)\!\proj\! \Alice & = & \RECSEQ{& \End\\
    & & &}{i}{\xx}{\IF\ \Alice\!=\!\Alice\!=\!\Bob\ \THEN\ 
    (\ldots) \\
    & & & \ELSE\ \IF\ \Alice\!=\!\Alice\ \THEN\ ({\Lout{\Bob}{\Nat}{\IF\
        \Alice\!=\!\Bob\!=\!\Carol\ \THEN\ \ldots}}) \\
    & & & \ELSE\ \IF\ \Alice=\Bob\ \THEN\ \ldots \\
    & & & \ELSE\ \ldots} \APP n
\end{array}
\]}

For readability, we omit from our examples the impossible cases created by
the projection algorithm. The number of cases can be automatically trimmed to
only keep the ones whose resolutions depend on free variables.

In this case, the projection yields the following local type:
{\small
$$
G(n)\proj \Alice=
\RECSEQ{\End}{i}{\xx}{\Lout{\Bob}{\Nat}{\xx}} \APP n
$$
}
\noindent Before typing, we first define some abbreviations:
{\small
\[
\begin{array}{rcl}
\Alice(n)& =& \sr{a}{\Alice,\Bob,\Carol}{y}{(\RECSEQP{\inact}{\ii}{\X}{\outS{\y}{\Bob,\e[i]}}
  \X \APP n)}\\
\Delta(n) & = & \{y:(G(n)\proj \Alice) \}\\
\Gamma & = & \AT{n}{\Nat}, \AT{a}{\ENCan{G(n)}}
\end{array}
\]}
Our goal is to prove the typing judgement 
\[ \Gamma \proves \Alice(n) \rhd \emptyset
\]
We start from the leafs.
\[
\begin{prooftree}
{
\begin{prooftree}
{
\begin{prooftree}
{\Gamma, \ii : \minus{\II},\X:\Delta(\ii)  \vdash \Env }
\justifies
{\Gamma, \ii : \minus{\II},\X:\Delta(\ii)  \vdash \X  \rhd
y:\Delta(\ii)}\using\tftrule{TVar}
\end{prooftree}
}
\justifies
{\Gamma, \ii : \minus{\II},\X:\Delta(\ii)  \vdash \outS{\y}{\Bob,\e[i]}\X  \rhd
y:\Lout{\Bob}{\Nat}{\Delta(\ii)}
}\using\tftrule{TOut}
\end{prooftree}}
\justifies
{\Gamma, \ii : \minus{\II},\X:\Delta(\ii)  \vdash \outS{\y}{\Bob,\e[i]}\X  \rhd
\Delta(\ii+1)
}\using\tftrule{TEq}
\end{prooftree}
\]
The \tftrule{TEq} rule can be used because
types $\Delta(\ii+1)$ and
$y:\!\Lout{\Bob}{\Nat}{(\RECSEQ{\End}{j}{\xx}{\Lout{\Bob}{\Nat}{\xx}} \APP \ii)}$
are equivalent: they have the same weak-head normal form (we use the rule
\trule{WfBase}).


Since we have the trivial $\Gamma\vdash \inact  \rhd \Delta(0)$, we can apply
the rules \tftrule{TApp} and \tftrule{TPRec}.
\[
\begin{prooftree}
{\begin{prooftree}
{\begin{array}{c}
\Gamma, \ii : \minus{\II},\X:\Delta(\ii)  \vdash \outS{\y}{\Bob,\e[i]}\X  \rhd
 \Delta(\ii+1) \\
\Gamma\vdash \inact  \rhd \Delta(0) \quad \Gamma,\AT{\ii}{\II}\vdash \Delta(\ii) \RHD \K
\end{array}
}
\justifies
{\Gamma \vdash (\RECSEQP{\inact}{\ii}{\X}{\outS{\y}{\Bob,\e[i]}}
  \X)\rhd \Pi \AT{\ii}{\II}.\Delta(\ii)}\using\tftrule{TPRec}
\end{prooftree}
}\justifies
{\Gamma \vdash (\RECSEQP{\inact}{\ii}{\X}{\outS{\y}{\Bob,\e[i]}}
  \X \APP n)\rhd \Delta(n)}\using\tftrule{TApp}
\end{prooftree}
\]
We conclude with \tftrule{TInit}.
\[
\begin{prooftree}
{
\begin{array}{@{}c@{}}
\Gamma \vdash a:\mar{\G(n)} \quad 
\Gamma \vdash (\RECSEQP{\inact}{\ii}{\X}{\outS{\y}{\Bob,\e[i]}}
  \X \APP n)\rhd \Delta(n)
\end{array}
}
\justifies
{\Gamma \vdash \Alice(n) \rhd \emptyset} 
\using\tftrule{TInit}
\end{prooftree}
\]




$\Bob(n)$ and $\Carol(n)$ can be similarly typed.


\subsection{Sequence example - \S~\ref{sec:introduction}
  (\ref{ex:sequence})}
The sequence example consists of $n$ participants organised in the following way
(when $n\geq 2$): the starter $\W[n]$ sends the first message, the final worker
$\W[0]$ receives the final message and the middle workers first receive a
message and then send another to the next worker. We write below the result of
the projection for a participant $\W[\pp]$ (left) and the end-point type that
naturally types the processes (right):
{\small
\[
\begin{array}{l|l}
\begin{array}{llll}
   \RECSEQ{\End }{i}{\xx}{\\
 \ \IF\ \ \pp=\W[i+1]\ \THEN \ \Lout{\W[i]}{\Nat}\xx \\
 \ \ELSE\ \IF\ \ \pp=\W[i]\ \THEN \ \Lin{\W[i+1]}{\Nat}\xx \\
 \ \ELSE\ \xx \ } \\
 n
\end{array}
&
\begin{array}{lll}
\IF \ &
\p=\W[n]  \ \THEN \ \Lout{\W[n-1]}{\Nat}{}\ELSE\\
\IF &
\p=\W[0] \ \THEN \ {\Lin{\W[1]}{\Nat}{}}\ELSE\\
\IF & \p=\W[\ii]\  \THEN \  {\Lin{\W[i+1]}{\Nat}{}}{\Lout{\W[i-1]}{\Nat}{\!}}
\end{array}
\end{array}
\]
} %

\noindent
This example illustrates the main challenge faced by the type checking algorithm.
In order to type this example, we need to prove the equivalence
of these two types. For any concrete instantiation of $\p$ and $\n$, the standard weak
head normal form equivalence rule \trule{WfBase} is sufficient.  Proving the
equivalence for all $\p$ and $\n$ requires either (a) to bound the domain $\II$
in which they live, and check all instantiations within this finite domain using
rule \trule{WfRecF}; or (b) to prove the equivalence through the meta-logic
rule \trule{WfRecExt}.  In case (a), type checking terminates, while case (b)
allows to prove stronger properties about a protocol's implementation.

\subsection{Ring - Figure~\ref{fig:examples}(a)}

The typing of the ring pattern is similar to the one of the sequence.
The projection of this global session type for $\W[\pp]$ gives the following
local type: 
{\small
\[
\begin{array}{lll}
   \RECSEQ{(\GS{\W[\nn]}{\W[0]}{\Nat}.\End)\proj \W[\pp]\\
   \quad }{i}{\xx}{ \IF\ \pp=\W[\nn-i-1]\ \THEN \ \Lout{\W[\nn-\ii]}{\Nat}\xx\\
  \quad \quad\quad\quad\quad \ \ELSE\IF\ \ \pp=\W[\nn-\ii]\ \THEN \ \Lin{\W[\nn-\ii-1]}{\Nat}\xx \\
 \quad \quad\quad\quad\quad\ \ELSE\IF\ \ \xx \ } & n
\end{array}
\]
}

On the other hand, user processes can be easily type-checked with an end-point
type of the following form:
{\small
\[
\begin{array}{lll}
\IF \ &
\p=\W[0] \ \THEN \ \Lout{\W[1]}{\Nat}{\Lin{\W[n]}{\Nat}{}}\\
\ELSE\IF &
\p=\W[n]  \ \THEN \ \Lin{\W[n-1]}{\Nat}{\Lout{\W[0]}{\Nat}{}}\\
\ELSE\IF & 1 \leq i+1 \leq n-1  \ \mathsf{and}\ \p=\W[\ii+1]\\
& \THEN \  {\Lin{\W[i]}{\Nat}{}}{\Lout{\W[i+2]}{\Nat}{}}
\end{array}
\]}

\noindent
Proving the equivalence between these types is similar as the one the sequence:
we rely on rules \trule{WfBase} and \trule{WfRecF} when the domain of $n$ is
bounded, or on the meta-logic rule \trule{WfRecExt}.

\subsection{Mesh pattern - Figure~\ref{fig:examples}(b)} 
The mesh example describes nine different participants behaviours (when
$n,m\geq 2$). The participants in the first and last rows and columns, except
the corners which have two neighbours, have three neighbours. The other
participants have four neighbours. The specifications of the mesh are defined by
the communication behaviour of each participant and by the links the
participants have with their neighbours. 
The term below is the result of the projection of the global type for
participant $\pp$

{\small
\[
\begin{array}{lll}
\RECSEQ{  ~~ (\RECSEQ{\End\proj \pp}{k}{\zz}{
    \IF \ \pp=\W[0][k+1]\ \THEN \ \Lout{\W[0][k]}{\Nat}\zz\\
     \qquad\qquad\quad \quad\quad\quad\quad~~ \ \ELSE\IF\ \ \pp=\W[0][k]\ \THEN \ \Lin{\W[0][k+1]}{\Nat}\zz \\
 \qquad\qquad\quad \quad\quad\quad\quad~~\ \ELSE\ \ \zz \ })  m\\
     \quad }{i}{\xx}{ \\
       \qquad (\RECSEQ{~( \IF \ \pp=\W[i+1][0]\ \THEN \ \Lout{\W[i][0]}{\Nat}\xx\\
     \quad \quad\quad\quad \ \ELSE\IF\ \ \pp=\W[i][0]\ \THEN \ \Lin{\W[i+1][0]}{\Nat}\xx \\
 \quad \quad\quad\quad \ \ELSE\ \ \xx \ )\\       
           \quad \qquad  }{j}{\yy}{ \\         
          \qquad \qquad  ~\IF \ \pp=\W[i+1][j+1]\ \THEN \ \Lout{\W[i][j+1]}{\Nat}\\
          \qquad\qquad\qquad \IF \ \pp=\W[i+1][j+1]\ \THEN \ \Lout{\W[i+1][j]}{\Nat} \yy\\
          \qquad\qquad\qquad  \ELSE\IF\ \ \pp=\W[i+1][j]\ \THEN \ \Lin{\W[i+1][j+1]}{\Nat}\yy \\
                    \qquad\qquad\qquad \ELSE\ \yy\\
     \quad \quad\quad\quad \ \ELSE\IF\ \ \pp=\W[i][j+1]\ \THEN \ \Lin{\W[i+1][j+1]}{\Nat}\yy \\
      \qquad\qquad\qquad \IF \ \pp=\W[i+1][j+1]\ \THEN \ \Lout{\W[i+1][j]}{\Nat} \yy\\
          \qquad\qquad\qquad  \ELSE\IF\ \ \pp=\W[i+1][j]\ \THEN \ \Lin{\W[i+1][j+1]}{\Nat}\yy \\
                    \qquad\qquad\qquad \ELSE \ \yy\\
        \qquad\qquad \ELSE\IF \ \pp=\W[i+1][j+1]\ \THEN \ \Lout{\W[i+1][j]}{\Nat} \yy\\
          \qquad\qquad  \ELSE\IF\ \ \pp=\W[i+1][j]\ \THEN \ \Lin{\W[i+1][j+1]}{\Nat}\yy \\
            \qquad\qquad  \ELSE \ \yy}\\
       \qquad  m) \\
     \quad n} 
\end{array}
\]
}

From Figure~\ref{fig:examples}(c),
the user would design the end-point type as follows:
{\small\[
\begin{array}{lll}
\IF \ &
\p=\W[n][m] \ \THEN \ \Lout{\W[n-1][m]}{\Nat}{\Lout{\W[n][m-1]} {\Nat}{}}\\
\ELSE\IF &
\p=\W[n][0]  \ \THEN \ \Lin{\W[n][1]}{\Nat}{\Lout{\W[n-1][0]}{\Nat}{}}\\
\ELSE\IF &
\p=\W[0][m]  \ \THEN \ \Lin{\W[1][m]}{\Nat}{\Lout{\W[0][m-1]}{\Nat}{}}\\
\ELSE\IF &
\p=\W[0][0]  \ \THEN \ \Lin{\W[1][0]}{\Nat}{\Lin{\W[0][1]}{\Nat}{}}\\
\ELSE\IF &
 1 \leq k+1 \leq m-1 ~ \mathsf{and}\ \p=\W[n][k+1]\\
 & \THEN \ \Lin{\W[n][k+2]}{\Nat}{\Lout{\W[n-1][k+1]}{\Nat}{\Lout{\W[n][k]}{\Nat}}}\\ 
\ELSE\IF &
 1 \leq k+1 \leq m-1 ~ \mathsf{and}\ \p=\W[0][k+1]\\
 & \THEN \ \Lin{\W[1][k+1]}{\Nat}{\Lin{\W[0][k+2]}{\Nat}\Lout{\W[0][k]}{\Nat}}\\ 
\ELSE\IF &
 1 \leq i+1 \leq n-1  ~\mathsf{and}\ \p=\W[i+1][m]\\
 & \THEN \ \Lin{\W[i+2][m]}{\Nat}{\Lout{\W[i][m]}{\Nat}{\Lout{\W[i+1][m-1]}{\Nat}}}\\ 
\ELSE\IF &
 1 \leq i+1 \leq n-1  ~\mathsf{and}\ \p=\W[i+1][0]\\
 & \THEN \ \Lin{\W[i+2][0]}{\Nat}{\Lin{\W[i+1][1]}{\Nat}{\Lout{\W[i][0]}{\Nat}}}\\ 
 \ELSE\IF &
 1 \leq i+1 \leq n-1  ~\mathsf{and}\ 1 \leq j+1 \leq m-1 ~\mathsf{and}~ \p=\W[i+1][j+1]\\
 & \THEN \ \Lin{\W[i+2][j+1]}{\Nat}{\Lin{\W[i+1][j+2]}{\Nat}{\Lout{\W[i][j+1]}{\Nat}{\Lout{\W[i+1][j]}{\Nat}}}}\\ 
\end{array}
\]}

Each case denotes a different local behaviour in the mesh pattern. 
We present the following meta-logic proof of the typing equivalence through
\trule{WfRecExt} in the two cases of the 
top-left corner and bottom row, in order to demonstrate how our
system types the mesh session. The other cases are left to the reader. 

Let $T[\p][n][m]$ designate the first original type and $T'[\p][n][m]$ the
second type. To prove the type equivalence, we want to check that for all $\n,
\m \geq 2$ and $\p$, we have: 
\begin{center}
\small
$(\prod n.\prod m.T[\p][n][m])\n \m \redsym^\ast T_{\n, \m}\not\redsym$ iff 
$(\prod n.\prod m.T'[\p][n][m]) \n \m \redsym^\ast T_{\n, \m}\not\redsym$. 
\end{center}

For $\p=\W[n][m]$, which implements the top-left corner, the generator type
reduces several steps and gives the end-point type
$\Lout{\W[\n-1][\m]}{\Nat}{\Lout{\W[\n][\m-1]} {\Nat}{\inact}}$, which is the same to
the one returned in one step by the case analysis of the type built by the
programmer. For $\p=\W[0][k+1]$, we analyse the case 
where $1 \leq k+1 \leq \m-1$. The generator type returns the end-point type
\Lin{\W[1][k+2]}{\Nat}{\Lin{\W[0][k+2]}{\Nat}\Lout{\W[0][k]}{\Nat}}.
One can observe that the end-point type returned for $\p=\W[0][k+1]$ in the type of the
programmer is the same as the one returned by the generator. Similarly for all
the other cases. 

By \tftrule{TOut, TIn}, we have: 
{\small
\[ 
a:\mar{\G} \vdash \outS{y}{\W[n-1][m], f(n-1,m)}\outS{y}{\W[n][m-1], f(n,m-1)}{\inact}
  \rhd \D,\y:\G \proj{\p_{\text{top-left}}}
  \]
  \[
a:\mar{\G} \vdash \inpS{y}{\W[1][k+1], z_1}\inpS{y}{\W[0][k+2], z_2}\outS{y}{\W[0][k], f(0,k)}\inact \rhd \D', \y :\G \proj{\p_{\text{bottom}}}\\
\]
}
where $\G \proj{\p}$ is obtained from the type above. 


\subsection{FFT example - Figure~\ref{fig:fft}}
\noindent 
\label{subsec:fft:typing}
We prove
type-safety and deadlock-freedom for the FFT processes.  
Let $P_{\text{fft}}$ be the following process:
{\small
\[
\begin{array}{lll}
\Pi n.(\nu a)(
\mathbf{R} \ \sr{a}{\p_0..\p_{2^n-1}}{y}{P(2^n-1,\p_0,x_{\overline{\p_0}},y,r_{\p_0})}\\[1mm]
\quad \lambda{\ii}.\lambda {\Y}.
{(\sr{a}{\p_{\ii+1}}{y}{P({\ii+1},\p_{\ii+1},x_{\overline{\p_{\ii+1}}},y,r_{\p_{\ii+1}})}
\sep \Y)}\APP 2^n-1)
\end{array}
\] 
}

\noindent
As we reasoned above, each $P(n,\p,x_{\overline{\p}},y,r_\pp)$ is
straightforwardly typable by an end-point type which can be proven to be
equivalent with the one projected from the global type $G$ from
Figure~\ref{fig:fft}(c). Automatically checking the equivalence for all $n$ is
not easy though: we need to rely on the finite domain restriction using
\trule{WfRecF} or to rely on a meta-logic proof through \trule{WfRecExt}.
The following theorem says once $P_{\text{fft}}$ is applied 
to a natural number $\m$, its evaluation always 
terminates with the answer at $r_\pp$.

\begin{thm}[Type safety and deadlock-freedom of FFT]
\label{theorem:fft}
For all $\m$, $\emptyset \proves P_{\text{fft}} \APP \m\rhd \emptyset$;  
and if $P_{\text{fft}} \APP  \m \longrightarrow^\ast Q$, then $Q\longrightarrow^\ast
(\Poutend{r_0}{0}{\X_0}\pc\ldots\pc\Poutend{r_{2^\m-1}}{0}{\X_{2^\m-1}})$
where the $\Poutend{r_\pp}{0}{\X_\pp}$ are the actions sending the
final values $\X_\pp$ on external channels $r_\pp$. 
\end{thm}
\begin{proof}
For the proof, we first show 
$P_{\text{fft}} \APP  \m$ 
is typable by a single, multiparty dependent session
(except the answering channel at $r_\pp$).  
Then the result is immediate as a corollary
of progress (Theorem \ref{thm:progress}). 

To prove that 
the processes
are typable against the given global type,  
we start from the end-point projection. 

We assume index $n$ to be a parameter as in Figure~\ref{fig:fft}. 
The main loop is an iteration over the $n$ steps of the
algorithm. Forgetting for now the content of the main loop, the generic
projection for machine $\pp$ has the following skeleton:\\[1mm]
{\small
$\begin{array}{ll}
& \Pi n.(\RECSEQ{(\RECSEQ{\End}{l}{\xx}{(\ldots)}\APP n) \\
& }{k}{\uuu}{\\
& \quad \Pifthenelse{\pp=k}{\Lout{k}{\U}{\Lin{k}{\U}{\uuu}}}{\uuu}})\\
& ~ 2^n\\
\end{array}$} 
\\[1mm]
A simple induction gives us through \trule{WfRecExt} the equivalent type:\\[1mm]
{\small
$\begin{array}{ll}
& \Pi n.\Lout{\pp}{\U}{\Lin{\pp}{\U}{(\RECSEQ{\End}{l}{\xx}{(\ldots)}\APP n)}} \APP 2^n\\
\end{array}$} 
\\[1mm]
We now consider the inner loops. The generic projection gives:\\[1mm]
{\small
$\begin{array}{@{}l@{}l}
& \ldots \\
& \quad (\RECSEQ{\xx}{i}{\yy}{ \\
    & \qquad (\RECSEQ{\yy}{j}{\zz}{ \\
          & \quad \qquad \Pifthenelse{\pp=i*2^{n-l}+2^{n-l-1}+j=i*2^{n-l}+j}{\ldots\\
          & \quad \qquad }{\Pifthenelse{\pp=i*2^{n-l}+2^{n-l-1}+j}{\Lout{i*2^{n-l}+j}{\U}{\ldots}\\
          & \quad \qquad }{\Pifthenelse{\pp=i*2^{n-l}+j}{\Lin{i*2^{n-l}+2^{n-l-1}+j}{\U}{\ldots}\\
          & \quad \qquad }{\Pifthenelse{\ldots}{\ldots}{\ldots}}}}}\\
    & \qquad)\APP  2^{n-l-1}} \\
  & \quad) \APP 2^l \\
& \ldots \\
\end{array}$}
\\[1mm]
An induction over $\pp$ and some simple arithmetic over binary numbers gives us
through \trule{WfRecExt} the only two branches that can be taken:\\[1mm]
{\small
$\begin{array}{@{}l@{}l}
& \ldots \\
&\Pifthenelse{\bit{n-l}(\pp)=0 \\
& }{\Lin{\pp+2^{n-l-1}}{\U}{\Lout{\pp+2^{n-l-1}}{\U}{\Lout{\pp}{\U}{\Lin{\pp}\U{\xx}}}}\\
& }{\Lout{\pp-2^{n-l-1}}{\U}{\Lin{\pp-2^{n-l-1}}{\U}{\Lout{\pp}{\U}{\Lin{\pp}\U{\xx}}}}}\\
& \ldots \\
\end{array}$} 
\\[1mm]
The first branch corresponds to the upper part of the butterfly while the second
one corresponds to the lower part.
For programming reasons (as seen in the processes, the natural implementation
include sending a first initialisation message with the $x_k$ value), we want to
shift the self-receive $\Lin{\pp}\U{}$ from the initialisation to the beginning
of the loop iteration at the price of adding the last self-receive to the end:
$\Lin{\pp}\U{\End}$. 
The resulting equivalent type up to $\equiv$ is: \\[2mm]
{\small
$
\begin{array}{ll}
\Pi n.\Lout{\pp}{\U}{} \\
 \ (\LR{\Lin{\pp}\U{\End}}{l}{\xx} {\\[1mm]
 \  \Pifthenelse{\bit{n-l}(\pp)=0\\[1mm]
    \ } {\Lin{\pp}\U{\Lin{\pp+2^{n-l-1}}{\U}{\Lout{\pp+2^{n-l-1}}{\U}{\Lout{\pp}{\U}{\xx}}}}\\[1mm]
    \ } {\Lin{\pp}\U{\Lout{\pp-2^{n-l-1}}{\U}{\Lin{\pp-2^{n-l-1}}{\U}{\Lout{\pp}{\U}{\xx}}}}}}) \APP n
\end{array}
$
}\\[1mm]
From this end-point type, it is straightforward to type 
and implement the processes defined in 
Figure~\ref{fig:fft}(d) in \S~\ref{subsec:fft}. 
Hence we conclude the proof. 
\end{proof}

\subsection{Web Service} 
\label{sec:applications}
\noindent 
This section demonstrates the expressiveness of our type theory. We 
program and type a real-world Web service usecase: Quote Request (C-U-002) is the most complex scenario 
described in \cite{CDLRequirements}, the
public document authored by the W3C Choreography Description Language Working
Group \cite{CDL}. 




\begin{figure}[ht]
\begin{center}
\begin{tabular}{c}
\xymatrix@C=25pt@R=1pt{
 & *+[F] {\Supp[0]}\ar@{<->}[r]\ar@{<->}[ddr]
 &  *+[F] {\Manu[0]}\ar@{}[r]|\equiv & \txt{\Manu[0][0]} \\
 &  &  \\
  *+[F]{\Buyer} \ar@{<->}[r]\ar@{<->}[uur]\ar@{<->}[ddr]
&  *+[F]{\Supp[1]} \ar@{<->}[ddr]
&  *+[F] {\Manu[1]}\ar@{}[r]|\equiv & \txt{\Manu[0][1]\\ \Manu[2][1]}\\
 &  &   \\
&  *+[F]{\Supp[2]} \ar@{<->}[r]\ar@{<->}[uur]
&  *+[F] {\Manu[2]}\ar@{}[r]|\equiv & \txt{\Manu[1][2]\\ \Manu[2][2]}\\
 & {:} & {:}  \\
}\\
\end{tabular}
\end{center}
\caption{The Quote Request usecase (C-U-002) \cite{CDLRequirements} \label{fig:CDL}}
\end{figure}

\paragraph{\bf Quote Request usecase}
The usecase is described below (as published in \cite{CDLRequirements}). 
A buyer interacts with multiple suppliers who in turn
interact with multiple manufacturers in order to obtain quotes for some
goods or services. 
The steps of the interaction are: 
\begin{enumerate}[(1)]
\item A buyer requests a quote from a set of suppliers.
All suppliers receive the request for quote and send requests
for a bill of material items to their respective manufacturers.

\item 
The suppliers interact with their manufacturers to build their quotes
for the buyer. The eventual quote is sent back to the buyer. 

\item EITHER
\begin{enumerate}
\item The buyer agrees with one or more of the quotes and places the order or orders. OR 
\item The buyer responds to one or more of the quotes by modifying 
and sending them back to the relevant suppliers.
\end{enumerate}
\item EITHER
\begin{enumerate}
\item The suppliers respond to a modified quote 
by agreeing to it and sending a confirmation message back to the buyer. OR

\item The supplier responds by modifying the quote and sending it back
to the buyer and the buyer goes back to STEP 3. OR 

\item 
The supplier responds to the buyer rejecting the modified quote. OR 

\item 
The quotes from the manufacturers need to be renegotiated by the supplier. Go to STEP 2. 
\end{enumerate}
\end{enumerate}
%
%
The usecase, depicted in figure~\ref{fig:CDL}, may seem simple, 
but it contains many challenges. 
The Requirements 
in Section 3.1.2.2 of \cite{CDLRequirements} include:  
{\bf [R1]} the ability to repeat the same set of interactions 
between different parties using a single definition
and to compose them;
{\bf [R2]} 
the number of participants 
may be bounded at design time or at runtime; and  
{\bf [R3]} 
the ability to {\em reference a global description from within a
global description} to support {\em recursive behaviour} as 
denoted in {\sc Step} 4(b, d).   
The following works through a parameterised 
global type specification that satisfies these requirements.   

\paragraph{\bf Modular programming using global types}
We develop the specification of the usecase program modularly, 
starting from smaller global types. Here, 
$\texttt{Buyer}$ stands for the buyer, 
$\texttt{Supp}[\ii]$ for a supplier, 
and $\texttt{Manu}[\jj]$ for a manufacturer. Then we alias manufacturers by $\texttt{Manu}[\ii][\jj]$ to identify that $\texttt{Manu}[\jj]$ 
is connected to  $\texttt{Supp}[\ii]$ (so a single $\texttt{Manu}[\jj]$ can have multiple
aliases $\texttt{Manu}[\ii'][\jj]$, see figure~\ref{fig:CDL}). Then, using the idioms presented in \S~1, {\sc Step} 1 
is defined as:
\[ 
\begin{array}{lll}
\G_1 = 
\FOREACH{\ii}{\II}{
\GS{\texttt{Buyer}}{\texttt{Supp}[\ii]}{\mathsf{Quote}}.\End}
\end{array}
\]
For {\sc Step} 2, we compose a nested loop and the subsequent action
within the main loop ($J_i$ gives all $\texttt{Manu}[\jj]$ connected
to $\texttt{Supp}[\ii]$):  
{\small
\[ 
\begin{array}{ll}
\multicolumn{2}{l}{G_2 = \mathtt{foreach}({\ii}:{\II})\{
\MERGE{G_2[i],\
\GS{\texttt{Supp}[i]}{\texttt{Buyer}}{\mathsf{Quote}}.\End}\}} \\[1mm]
G_2[i]= \mathtt{foreach}({\jj}:{J_i})\{ & 
\GS{\texttt{Supp}[\ii]}{\texttt{Manu}[\ii][\jj]}{\mathsf{Item}}.\\
& \GS{\texttt{Manu}[\ii][\jj]}{\texttt{Supp}[\ii]}{\mathsf{Quote}}.\End\}
\end{array}
\]
}
%
%

\noindent
$G_2[i]$ represents the second loop between the $i$-th supplier and its
manufacturers.  
Regarding {\sc Step} 3,  
the specification involves buyer preference for certain suppliers.  
Since this can be encoded using dependent types 
(like the encoding of $\mathsf{if}$),
we omit this part and assume 
the preference is given by the (reverse) ordering of $\II$ in order to focus on the
description of the
interaction structure.
\[
\begin{array}{rcl}
G_3 & = & \GR{\ \ty\ }{i}{\yy}{} \TO{\texttt{Buyer}}{\texttt{Supp}[i]} : \{ \\
& & \quad \begin{array}{llll}
\mathsf{ok}: & \End \\
\mathsf{modify}: & \multicolumn{3}{l}{\GS{\texttt{Buyer}}{\texttt{Supp}[i]}{\mathsf{Quote}}} \\
& \TO{\texttt{Supp}[i]}{\texttt{Buyer}}: \{
&  \mathsf{ok}: & \End \\
&& \mathsf{retryStep3}: & \yy\\
&&  \mathsf{reject}: & \End\}\}\ \ii
\end{array}
\end{array}
\]
In the innermost branch, $\mathsf{ok}$, $\mathsf{retryStep3}$ and 
$\mathsf{reject}$ correspond to {\sc Step} 4(a), (b) and (c)
respectively. Type variable $\ty$ is for (d). 
We can now compose all these subprotocols together. Taking
$G_{23}\ =\ \GM{\ty}{\MERGE{G_2,\ G_3}}$ 
and assuming $\II=[0..i]$, the full global type is 
\[
\lambda i.\lambda \VEC{J}.\MERGE{G_1,G_{23}} 
\]
where we have $i$ suppliers, and $\VEC{J}$ gives the $J_i$ (continuous) index sets of the $\texttt{Manu}[\jj]$s connected with each $\texttt{Supp}[\ii]$. 

\paragraph{End-point types}
We show the end-point type for suppliers, who engage in
the most complex interaction structures among the participants.  
The projections corresponding to 
$G_1$ and $G_2$ are 
straightforward:
\[
\begin{array}{lllll}
G_1\proj{\texttt{Supp}[\n]} = \Linn{\texttt{Buyer}}{\mathsf{Quote}} \\[1mm]
G_2\proj{\texttt{Supp}[\n]} = \mathtt{foreach}({\jj}:{J_i})\{ \Loutt{\texttt{Manu}[\n][\jj]}{\mathsf{Item}}; \\
\qquad\qquad\qquad\qquad {\Linn{\texttt{Manu}[\n][\jj]}{\mathsf{Quote}}\}; {\Loutt{\texttt{Buyer}}{\mathsf{Quote}}}}
\end{array}
\]
For $G_3\proj{\texttt{Supp}[\n]}$, 
we use the branching injection and 
mergeability theory developed in 
\S~\ref{subsec:endpoint}. 
After the relevant application of \trule{TEq}, we can obtain 
the following projection:
\[
\begin{array}{l}
\begin{array}{l}
\&\langle\texttt{Buyer}, \{ \\
\\
\\
\\
\\
\end{array}
\begin{array}{ll}
\mathsf{ok}: & \End\\
\mathsf{modify}: & \Lin{\texttt{Buyer}}{\mathsf{Quote}}{} \oplus\langle\texttt{Buyer}, \{ \\
& \begin{array}{ll}
\qquad \mathsf{ok}: & \End \\
\qquad \mathsf{retryStep3}: & T\\
\qquad \mathsf{reject}: & \End\}\rangle\}\rangle\\
\end{array}
\end{array}
\end{array}
\]
where $T$ is a type for the invocation 
from $\texttt{Buyer}$: 
\[
\begin{array}{l}
\IF \ \n\leq i \ \THEN \ \&\langle \texttt{Buyer}, \{ \mathsf{closed}:\End, 
\ \mathsf{retryStep3}:\ty \} \rangle \\ 
\ELSE\IF \ i=\n \ \THEN \ \ty
\end{array}
\]
To tell the other suppliers whether the loop is being reiterated 
or if it is finished, 
we can simply insert the following closing notification
$\FOREACH{\jj}{\II\setminus \ii}
{\TO{\texttt{Buyer}}{\texttt{Supp}[j]}:\{\mathsf{close}:\}}
$
before each $\End$, 
and a similar retry notification (with label $\mathsf{retryStep3}$) before $\ty$.
%
Finally, each end-point type is formed by the following composition: 
\[ 
\begin{array}{l}
\MERGE{G_1\proj{\texttt{Supp}[\n]}, 
{\GM{\ty}{\MERGE{G_2\proj{\texttt{Supp}[\n]},\ G_3\proj{\texttt{Supp}[\n]}}}}})
\end{array}
\]
Following this specification, the projections can be implemented in various
end-point languages (such as CDL or BPEL).  

\section{Conclusion and related work}
\label{sec:related}
This paper studies a parameterised multiparty session type theory
which combines three well-known theories: indexed dependent types
\cite{DBLP:conf/popl/XiP99}, dependent types with
recursors \cite{DBLP:conf/mfps/Nelson91}
and multiparty session types \cite{BettiniCDLDY08LONG,CHY07}.
The resulting typing system is decidable (under
an appropriate assumption about the index arithmetic).
It offers great expressive power
for describing complex communication topologies
and guarantees safety properties of processes running under
such topologies.
We have explored the impact of
parameterised type structures for
communications
through implementations of the above web service usecases and of several parallel
algorithms in Java and C with session
types~\cite{HU07TYPE-SAFE,HKOYH10}, including
the N-body (with a ring topology),
the Jacobi method
(with sequence and mesh topologies)
and the FFT \cite{NYPHK11,NYH12}.
We observe (1) a clear coordination
of the communication behaviour of each party with the
construction of the whole multiparty protocol, thus reducing
programming errors and ensuring deadlock-freedom; and (2)
a performance benefit against the original binary session version,
reducing the overhead of multiple binary session
establishments (see also \cite{NYPHK11,NYH12}).
Full implementation and integration of our theory
into~\cite{HU07TYPE-SAFE,CorinDFBL09,HKOYH10} is on-going work.

\subsection{Related work}
We focus on the works on dependent types and other typed process
calculi which are related to multiparty session types; for further
comparisons of session types with other service-oriented calculi and
behaviour typing systems, see \cite{DL10} for a wide ranging survey of
the related literature.

\paragraph*{\bf Dependent types}
\noindent
The first use of primitive recursive functionals for dependent types
is in Nelson's $\mathcal{T}^\pi$~\cite{DBLP:conf/mfps/Nelson91}
for the $\lambda$-calculus, which is 
a finite
representation of $\mathcal{T}^\infty$ by Tait and Martin L\"of
\cite{Tait,PerMartin}.
$\mathcal{T}^\pi$ 
can type functions previously untypable in ML, and the finite
representability of dependent types makes it possible to have a
type-reconstruction algorithm. We also use the ideas from DML's
dependent typing system in \cite{DBLP:conf/popl/XiP99,DependentBook}
where type dependency is only allowed for index sorts, so that
type-checking can be reduced to a constraint-solving problem over
indices.
Our design choice to combine both systems 
gives (1) the simplest formulation of sequences of global and end-point types and
processes described by the primitive recursor; (2) a precise specification
for parameters appearing in the participants based on index sorts; and (3) a
clear
integration with the full session types
and general recursion, whilst ensuring decidability of type-checking (if the
constraint-solving problem is decidable).
From the basis of these works,
our type equivalence
does not have to rely on
behavioural equivalence between processes, but only on the strongly
normalising {\em types} represented by recursors.

Dependent types have been also studied in the context
of process calculi, where the dependency centres
on locations (e.g.~\cite{Hennessy07}),
and channels (e.g.~\cite{Yoshida04}) for mobile agents or
higher-order processes.
An effect-based session typing system
for corresponding assertions to specify fine-grained
communication specifications is studied
in \cite{BAG05} where effects can appear both in types and
processes.
None of these works investigate families of global specifications using
dependent types.
Our main typing rules require a careful treatment for type
soundness not found in the previous works, due to the simultaneous instantiation
of terms and indices by the recursor, with reasoning by mathematical induction
(note that type soundness was left open in \cite{DBLP:conf/mfps/Nelson91}).

\paragraph*{\bf Types and contracts for multiparty interactions}
The first papers on multiparty session types were \cite{BC07} and
\cite{CHY07}. The former uses a distributed calculus where
each channel connects a master end-point to one or more slave
endpoints; instead of global types, they use only
local types.
Since the first
work \cite{CHY07} was proposed, this theory has been used in the
different contexts such as distributed protocol implementation and
optimisation~\cite{SivaramakrishnanNZE10},
security~\cite{CorinDFBL09,ccdr10},
design by contract~\cite{BHTY10},
parallel algorithms~\cite{NYPHK11,NYH12}, web services~\cite{YDBH10},
multicore programming~\cite{YoshidaVPH08},
an advanced progress guarantee~\cite{BettiniCDLDY08LONG},
messaging optimisation~\cite{esop09},
structured exceptions~\cite{CGY10},
buffer and channel size analysis for multiparty
interactions~\cite{Bufferfull},
medical guidelines~\cite{NYH09} and
communicating automata \cite{DY12},
some of which initiated industrial collaborations, cf.~\cite{HondaMBCY11}.
Our typing system can be smoothly integrated
with other works as no changes to the runtime typing components have been made
while expressiveness has been greatly improved.

%
%

The work \cite{carbone.honda.yoshida:esop07}
presented an {\emph{executable global processes}}
for web interactions
based on binary session types.
Our work provides flexible, programmable  global descriptions
as {\em types}, offering a progress for parameterised multiparty
session, which is not ensured in \cite{carbone.honda.yoshida:esop07}.


The work \cite{B10} provides a programming idiom of roles, defining
different classes of participants, and a different type system for
parameterised session types. There is no investigation of the system
expressivity for the 3D-Mesh pattern as we have presented in this
paper through the Fast Fourier Transformation example.  The static
type system follows the typing strategy and programming methodology of
multiparty session types: programmers first define the global type of
the intended pattern and then define each of the roles; the roles are then
validated through projection of the global type onto the principals by
type-checking.

Recent formalisms for typing multiparty interactions include
\cite{CP09,CairesV09}.
These works treat different aspects of dynamic session structures.
{\em Contracts} \cite{CP09} 
can type more processes
than session types, thanks to the flexibility of process syntax
for describing protocols.
However, typable processes themselves in \cite{CP09} may
not always satisfy the properties of session types such as progress: it is
proved later by checking whether the type meets a certain form.  Hence
proving progress with contracts effectively
requires an exploration of all possible paths (interleaving, choices)
of a protocol.  The most complex example of~\cite[\S~3]{CP09} (a group
key agreement protocol from~\cite{AST98}), which is typed as
$\pi$-processes with delegations, can be specified and articulated
by a single parameterised
global session type as:
{\small
\begin{align*}
\label{group}
 \Pi \AT{n}{I}.(& \FOREACH{i}{<n}{\GS{\W[n-i]}{\W[n-i+1]}{\Nat}};\\
               & \FOREACH{i}{<n}{\GS{\W[n-i]}{\W[n+1]}{\Nat}.\GS{\W[n+1]}{\W[n-i]}{\Nat}})
\end{align*}
}
%
Once the end-point process conforms to this specification, we can
automatically guarantee communication safety and progress.

{\em  Conversation Calculus} \cite{CairesV09} supports the dynamic
joining and leaving of participants.
We also introduced a dynamic role-based multiparty session type
discipline in previous work \cite{DY11}, where an arbitrary number of participants can
interact in a running session via a universal polling operator. This
work was extended with simple relations between roles
in \cite{Poon11} to dynamically handle
the complex topologies presented in this paper.
Although the formalism in \S~\ref{subsec:semantics} can operationally capture some dynamic
features, the aim of the present work is
not
the type-abstraction
of dynamic interaction patterns. Our purpose is
to capture, in a single type description, a family
of protocols over arbitrary numbers of participants,
to be instantiated at runtime.
Parameterisation gives freedom not possible with previous session
types: once typed, a parametric process is ensured that its arbitrary
well-typed instantiations, in terms of both topologies and process
behaviours, satisfy the safety and progress properties of typed
processes, without the cost of complex runtime support (as in \cite{DY11}).
Parameterisation, composition and repetition are common idioms in
parallel algorithms and choreographic/conversational interactions, all
of which are uniformly treatable in our dependent type theory. Here
types offer a rigorous structuring principle which can economically
abstract rich interaction structures, including parameterised ones.
\section*{Acknowledgements}

We thank the reviewers,
Lasse Nielsen and Dimitris Mostrous for their useful comments for the
paper and  Kohei Honda for discussions, and
Michael Emmanuel, Yiannos Kryftis, Nicholas Ng, Olivier Pernet and Hok Shun Poon
for their implementation work of the examples presented
in this paper.
The work is partially supported by EPSRC EP/G015635/1 and EP/F003757/1, and the NSF Ocean Observatories Initiative. The last author was also supported by EPSRC PhD Plus and EPSRC Knowledge Transfer Secondment.

\bibliographystyle{abbrv}

\bibliography{session}

\appendix

\section{Kinding and typing rules}
In this Appendix section, we give the definitions of kinding rules and
typing rules 
that were omitted in the main sections.

\subsection{Kinding and subtyping}
\label{app:kind}
Figure~\ref{fig:localkindsystem} defines the kinding rules for local types.  
Figure~\ref{fig:subtyping}
presents the subtyping rules which are used for typing runtime processes. 
The rules for the type isomorphism can be given 
by replacing $\subT$ by $\WB$.

\begin{figure}
\centering
\begin{tabular}{c}
\begin{prooftree}
  {\Gamma \vdash \p \rhd \Nat \quad \Gamma \vdash \T \RHD \LType
  \quad \Gamma\vdash \U  \RHD \SType \text{ or } \LType }
\justifies
{\Gamma \vdash \Lout{\p}{\U}{\T}  \RHD \LType} \using\tfrule{KLOut}
\end{prooftree}\\
\\
\begin{prooftree}
{\Gamma \vdash \p \rhd \Nat \quad \Gamma \vdash \T \RHD \LType 
  \quad \Gamma \vdash \U  \RHD \SType \text{ or } \LType}
\justifies
{\Gamma \vdash \Lin{\p}{\U}{\T}  \RHD \LType} \using\tfrule{KLIn}
\end{prooftree}\\
\\
\begin{prooftree}
{\Gamma \vdash \T \RHD  \Pi\AT{\ii}{\II}.\K \quad \Gamma \models \tii : I}
\justifies
{\Gamma \vdash \T\APP\tii \RHD \K\sub{\tii}{\ii}} \using\tfrule{KLApp}
\end{prooftree}
\ 
\begin{prooftree}
{\Gamma \vdash \p \rhd \Nat \quad \forall k \in K, \Gamma \vdash \T_k \RHD \LType }
\justifies
{\Gamma \vdash \Lsel{\p}{T_k} \RHD \LType} \using\tfrule{KLSel}
\end{prooftree}\\
\\
\begin{prooftree}
{\Gamma \vdash \p \rhd \Nat \quad \forall k \in K, \Gamma \vdash \T_k \RHD \LType }
\justifies
{\Gamma \vdash \Lbranch{\p}{T_k} \RHD \LType} \using\tfrule{KLBra}
\end{prooftree}\\
\\
\begin{prooftree}
{\Gamma \vdash \T \RHD \K\subst{0}{j} \quad 
 \Gamma, \ii:\II^- \vdash \T' \RHD \K\subst{i+1}{j}}
\justifies
{\Gamma \vdash \LR{\T}{\AT{\ii}{\II^-}}{\xx}{\T'}  
\RHD \Pi\AT{\jj}{\II}.\K} \using\tfrule{KLRcr}
\end{prooftree}\\
\\ 
\begin{prooftree}
{\Gamma \vdash \kappa}
\justifies
{\Gamma \vdash \xx \RHD \kappa} \using\tfrule{KVar}
\end{prooftree}
\quad
\begin{prooftree}
{\Gamma \vdash \T \RHD \LType }
\justifies
{\Gamma \vdash \LM{\xx}{\T} \RHD \LType} \using\tfrule{KLRec}
\end{prooftree}
\quad
\begin{prooftree}
{\Gamma \vdash \Env }
\justifies
{\Gamma \vdash \End \RHD \LType} \using\tfrule{KLEnd}
\end{prooftree}\\
\\
\end{tabular}

\caption{Kinding rules for local types} \label{fig:localkindsystem}
\end{figure}

\begin{figure}
\centering 
\begin{tabular}{c}
\begin{prooftree}
{\Gamma \vdash \T \subT \T'}
\justifies
{\Gamma \vdash  \Lout{\p}{\U}{\T}\subT  \Lout{\p}{\U}{\T'}} \using\trule{TSubOut}
\end{prooftree}\quad 
\begin{prooftree}
{\Gamma \vdash \T \subT \T'}
\justifies
{\Gamma \vdash  \Lin{\p}{\U}{\T}\subT  \Lin{\p}{\U}{\T'}} \using\trule{TSubIn}
\end{prooftree}\\
\\
\begin{prooftree}
{\forall k\in K\subseteq J,\ \Gamma \vdash \T_k \subT \T_k'}
\justifies
{ \Gamma \vdash \LselI{\p}{\T_k}{k}{K}\subT  \LselI{\p}{\T_\jj'}{\jj}{J}} \using\trule{TSSel${}_\subT$}
\end{prooftree}\\
\\
\begin{prooftree}
{\forall k\in J\subseteq K,\ \Gamma \vdash \T_k \subT \T_k'}
\justifies
{ \Gamma \vdash \LbranchI{\p}{\T_k}{k}{K} \subT \LbranchI{\p}{\T_\jj'}{\jj}{J}} \using\trule{TBra${}_\subT$}
\end{prooftree}\\
\\
\begin{prooftree}
{\Gamma \vdash \T_1 \subT \T_2 \quad \Gamma,\ii:\II \vdash \T_1' \subT \T_2'}
\justifies
{\Gamma \vdash \LR{\T_1}{\AT{\ii}{\II}}{\xx}{\T_1'} \subT \LR{\T_2}{\AT{\ii}{\II}}{\xx}{\T_2'}} \using\trule{TSubPRec}
\end{prooftree}\\
\\
\begin{prooftree}
{\Gamma \vdash T\sub{\LM{\xx}{\T}}{\xx} \subT \T'}
\justifies
{\Gamma \vdash \LM{\xx}{\T} \subT  \T'} \using\trule{TLSubRec}
\end{prooftree}\quad 
\begin{prooftree}
{\Gamma \vdash T'\subT T\sub{\LM{\xx}{\T}}{\xx}}
\justifies
{\Gamma \vdash T'\subT  \LM{\xx}{\T }} \using\trule{TRSubRec}
\end{prooftree}\\
\\
\begin{prooftree}
{\Gamma \vdash \T \subT \T' \quad \Gamma \models \tii:I=\tii':I}
\justifies
{\Gamma \vdash \T\APP \tii \subT \T'\APP \tii'} \using\trule{TSubProj}
\end{prooftree}\\
\\
\begin{prooftree}
{\Gamma \vdash \Env}
\justifies
{\Gamma \vdash \End\subT \End} \using\trule{TSubEnd}
\end{prooftree} \quad 
\begin{prooftree}
{\Gamma \vdash \Env}
\justifies
{\Gamma \vdash \xx\subT \xx} \using\trule{TSubRVar}
\end{prooftree}
\end{tabular}
\caption{Subtyping} \label{fig:subtyping}
\end{figure}



\section{Typing system for runtime processes}
\label{app:runtime}
This appendix defines a typing system for runtime processes
(which contain queues).   
Most of the definitions are from \cite{BettiniCDLDY08LONG}.

\begin{center}
\begin{tabular}{c}
\begin{tabular}{lrclr}
  Message & $\TQ$ & ::= &
         \oT\UT{\pv} & \emph{message send}\\
         &     & \sep  & \seltypes &\emph{message selection}\\
         &     & \sep  & $\TQ;\TQ'$ &\emph{message sequence}\\
\\
  Generalised \quad\quad\quad & $\TG$ & ::= &
         \T
  & \emph{session}\\
         &     & \sep & \TQ &\ \emph{message}\\
         &     & \sep & $\TQ;\T$ &\ \emph{continuation}\\[0.5mm]
\end{tabular}
\end{tabular}
\end{center}

\noindent
{\em Message types} are the types for queues: they
represent the messages contained in the queues. The {\em message
send type} \oT\UT{\pv}\ expresses the communication to $\pv$ 
of a value or of a channel of type \UT. The {\em
message selection type} \seltypes{} represents the communication to
participant $\pv$ of the label $l$ and $\TQ;\TQ'$
represents sequencing of message types (we assume associativity for $;$). For
example $\oplus\anglep{1}{\kf{ok}}$ is the message type for the
message $(2,1,\kf{ok})$. A {\em generalised type} is either
a session type, or a message type, or a message type followed by a
session type. Type $\TQ;\T$ represents the continuation of the type $\TQ$
associated to a queue with the type $\T$ associated to a pure process. An
example of generalised type is
$\oplus\anglep{1}{\kf{ok}};!\langle
3,\kf{string}\rangle; ?\langle 3,\kf{date}\rangle; \End$.

In order to take into account the structural congruence between
queues (see Figure~\ref{tab:structcong}) we consider message types
modulo the equivalence relation \equivT{}{}\ induced by the
rules shown as follows 
(with $\natural\in\set{!,\oplus}$ and $Z\in\set{\UT,l}$): 
\begin{center}
$\equivT{\oTG{Z}{\pv};\oTGp{Z}{\qv};\TQ}{\oTGp{Z}{\qv};\oTG{Z}{\pv};\TQ}$
\quad if $\pv\not = \qv$
\end{center}
The equivalence relation on message types extends to generalised types by:
\[\equivT{\TQ}{\TQ'}\text{ implies }\equivT{\TQ;\TG}{\TQ';\TG}\]
\indent 
We say that two session environments 
$\D$ and $\D'$ are equivalent (notation $\equivT{\D}{\D'}$) if
$\ccc:\TG\in\D$ and $\TG\not=\End$ imply $\ccc:\TG'\in\D'$ with $\equivT{\TG}{\TG'}$ and vice versa. This equivalence relation is used in rule \trule{Equiv} (see Figure~\ref{fig:runtime-process-typing}).

\begin{figure}[h!]
\[
\begin{array}{c}
 \begin{prooftree}
\der{\Ga}{\PP}{\D} \justifies \derqq{\Ga}{\emptyset}{\PP }{\D}
\using\trule{GInit}
\end{prooftree}
\ 
\begin{prooftree}
\derqq{\Ga}{\Sigma}{\PP}{\D}\quad \D \WB \D'
\justifies
\derqq{\Ga}{\Sigma}{\PP }{\D'}\using\trule{Equiv}
\end{prooftree}
\begin{prooftree}
\derqq{\Ga}{\Sigma}{\PP}{\D}\quad \D \subT \D'
\justifies
\derqq{\Ga}{\Sigma}{\PP }{\D'}\using\trule{Subs}
\end{prooftree}
\\\\
\begin{prooftree}
\derqq{\Ga}{\Sigma}{\PP}{\D}\quad \derqq{\Ga}{\Sigma'}{\Q}{\D'}
\quad \Sigma\cap\Sigma'=\emptyset
 \justifies
\derqq{\Ga}{\Sigma\cup\Sigma'}{\PP\Par \Q}{\D \Dcomp \D'}
 \using \trule{GPar}
 \end{prooftree}
\quad 
\begin{prooftree}
\derqq{\Ga}{\Sigma}{\PP}{\D}\quad \coe\D\s \justifies
\derqq{\Ga}{\Sigma\setminus\s}{(\nu\s)\PP }{\ms\D\s}
\using\trule{GSRes}
\end{prooftree}
\end{array}
\]
\caption{Run-time process typing} \label{fig:runtime-process-typing}
\end{figure}

\begin{figure}
\centering
$
\begin{array}{cc}
\begin{prooftree}
{{\Ga}\proves \Env}
\justifies 
{\derq{\Ga}{\s:\qbot}{\emptyset}} \using\trule{QInit}
\end{prooftree}
\\ \\ 
\begin{prooftree}
  \derq{\Ga}{\stdqueue}{\D} \qquad \de{\Ga}{\values}{\SST}
\justifies \derq{\Ga}{ \mqueue{\s}{\qtail{
      \valheap{\va}{\pv}{\qv}}}}{\D\Tcomp\set{ s[\qv] : \oT{\SST}{\pv}
} } \using\trule{QSend}
\end{prooftree}
\\ \\
\begin{prooftree}
\derq{\Ga}{\stdqueue}{\D} \justifies \derq{\Ga}{
\mqueue{\s}{\qtail{\delheap{s'[\pv']}{\pv}{\qv}}}}{\D,s'[\pv']:\T'\Tcomp\set{
s[\qv] : \oT{\T'}{\pv}} } \using\trule{QDeleg}
\end{prooftree}\\
\\
\begin{prooftree}
\derq{\Ga}{ \stdqueue}{\D}\quad j\in K \justifies \derq{\Ga}{
\mqueue{\s}{\qtail{\labheap{l_j}{\pv}{\qv}}}}{\D\Tcomp\set{ s[\qv] :
\LselI{\pv}{\T_k}{k}{K}}} \using\trule{QSel}
\end{prooftree}\\
\\
\end{array}
$
\caption{Queue typing} \label{fig:queue-typing}
\end{figure}

We start by defining the typing rules for single queues, in which the turnstile
$\vdash$ is decorated with \set\s\ (where \s\ is the session name of the current
queue) and the session environments are mappings from channels to message types.
The empty queue has empty session environment. Each message adds an
output type to the current type of the channel which has the role of
the message sender.  Figure~\ref{fig:queue-typing} lists the typing rules for
queues, where \Tcomp\ is defined by:

\begin{center}
$\D \Tcomp\set{s[\qv]: \TQ} = \begin{cases}
 \D' ,s[\qv]: \TQ';\TQ     & \text{if }\D=\D',  s[\qv]:\TQ', \\
 \D,s[\qv]: \TQ    & \text{otherwise}.
\end{cases}$
\end{center}
For example we can derive $\derq{}{\s:(3,1,\kf{ok})}{\set{\si\s 1:\oplus\anglep{1}{\kf{ok}}}}$.

In order to type pure processes in parallel with queues,
we need to use generalised types in session environments and
further typing rules. Figure~\ref{fig:runtime-process-typing}
lists the typing rules for processes containing queues.
The judgement $\derqq{\Ga}{\Sigma}{\PP}{\D}$ means that
$\PP$ contains the queues whose session names are in $\Sigma$.
Rule \trule{GInit} promotes the typing of a pure process
to the typing of an arbitrary process, since a pure process does not
contain queues. When two arbitrary processes are put in parallel (rule
\trule{GPar}) we need to require that each session name is associated
to at most one queue (condition $\Sigma\cap\Sigma'=\emptyset$). In
composing the two session environments we want to put in sequence a
message type and a session type for the same channel with role. 
For this reason we define the composition \Dcomp\ between 
generalised types as:

\begin{center}
$\TG \Dcomp \TG'=\begin{cases}
  \TG \Tcomp \TG'    & \text{if $\TG$ is a message type}, \\[-.2em]
    \TG'\Tcomp \TG   & \text{if $\TG'$ is a message type}, \\[-.2em]
    \bot  & \text{otherwise}\\
\end{cases}$
\end{center}

\noindent
where $\bot$ represents failure of typing.

We extend \Dcomp\ to session environments as expected:

\begin{center}
$\D \Dcomp \D' = \D\backslash \dom{\D'} \cup \D' \backslash \dom{\D}
\cup \set{\ccc: \TG \Dcomp \TG' \sep
  \ccc:\TG \in \D~\&~\ccc:\TG' \in \D'}$.
\end{center}

\noindent Note that \Dcomp\ is commutative, i.e., $\D \Dcomp
\D'=\D' \Dcomp \D$. Also if we can derive message types only for
channels with roles, we consider the channel variables in the
definition of \Dcomp\ for session environments since we want to
get for example $\set{y:\End}\Dcomp\set{y:\End}=\bot$ (message
types do not contains $\End$). 




In rule \trule{GSRes} we require the coherence of the session 
environment \D\ with respect to the session name \s\
to be restricted (notation \coe\D\s). This coherence 
is defined in Definition~\ref{cd} using the notions of projection 
of generalised types and of duality, 
introduced respectively in Definitions~\ref{def:genproj} and~\ref{dd}. 

\begin{DEFINITION}\label{def:genproj} 
\rm 
The {\em projection
of the generalised local type \TG\ onto \q}, denoted by $\pro\TG\q$,
is defined by:
\[
\begin{array}{rcl}
\pro{(\oT\UT{\pv};\TG')}\q&=&\begin{cases}
  !\UT;\pro{\TG'}\q  & \text{if }\q = \pv, \\
   \pro{\TG'}\q & \text{otherwise}.
\end{cases}\\
\pro{(\seltypes;\TG')}\q&=&\begin{cases}
  \oplus l;\pro{\TG'}\q  & \text{if }\q=\pv, \\
   \pro{\TG'}\q & \text{otherwise}.
\end{cases}\\
\pro{(\iT\UT{\p};\T)}\q&=&\begin{cases}
  ?\UT;\pro{\T}\q  & \text{if }\q=\p, \\
   \pro{\T}\q & \text{otherwise.}
\end{cases}\\
\pro{(\seltype)}\q&=&\begin{cases}
  \seltypeT  & \text{if }\q = \p, \\
  \sqcup_{i\in I}\pro{\T_i}\q & \text{otherwise.}
\end{cases}\\
\pro{(\branchtype)}\q&=&\begin{cases}
  \branchtypeT  & \text{if }\q=\p, \\
\sqcup_{i\in I}\pro{\T_i}\q 
&  \text{otherwise.}
\end{cases}\\
\pro{(\mu \xx.\T)}\q&=& \mu \xx.(\pro{\T}\q) \quad
\pro\xx\q=\xx\quad\pro\End\q=
         \End\\
       \end{array}
       \]

\noindent where $\sqcup_{i\in I}\pro{\T_i}\q$ is defined as $\sqcup_{i\in I}
T_i$ in Definition~\ref{def:mergeability} replacing by:  

\[ 
\begin{array}{lll}
\&\langle\{l_k:T_k\}_{i\in I}\rangle \ \mergecup 
\&\langle\{l_j:T_j'\}_{j\in J}\rangle \ = \\ 
\quad \&\langle\{l_k:T_k\mergecup T_k'\}_{k\in K\cap J}
\cup \{l_k:T_k\}_{k\in K\setminus J}
\cup \{l_j:T_j'\}_{j\in J\setminus K}\rangle\\[1mm]
\end{array}
\]
\end{DEFINITION}

\begin{DEFINITION}\label{dd}
The {\em duality relation} between projections of generalised
types is the minimal symmetric relation which satisfies:

\begin{center}
$\dual\End\End\qquad \dual\xx\xx\qquad \dual\T{\T'}\implies
\dual{\mu \xx.\T}{\mu \xx.\T'}$ 

$\dual\TG{\T}\implies\dual{!\UT;\TG}{?\UT;\T}$

$\forall i\in I ~\dual{\T_i}{\T_i'}\implies \dual \seltypeTp
\branchtypeTp $

$\exists i\in I ~l=l_i~\&~\dual{\TG}{\T_i}\implies \dual {\oplus
l;\TG} \branchtypes $

\end{center}
\end{DEFINITION}

\end{document}